\documentclass[a4paper, UKenglish, cleveref, thm-restate]{lipics-v2019}

\usepackage{microtype}          
\usepackage{macros}             
\usepackage{stmaryrd, scalerel}
\usepackage{algorithmicx}
\usepackage{algpseudocode}
\usepackage{tikz}               
\usetikzlibrary{petri, positioning}

\usepackage{blindtext}

\newcommand{\atomic}{\text{atomic}}

\newcommand{\cutoffvar}{\ell}		
\newcommand{\kwayvar}{\Delta}		

\hideLIPIcs

\author{Michael Blondin}{Département d'informatique, Université de
  Sherbrooke, Sherbrooke,
  Canada}{michael.blondin@usherbrooke.ca}{https://orcid.org/0000-0003-2914-2734}{Supported
  by a Quebec--Bavaria project funded by the Fonds de recherche du
  Québec (FRQ), by a Discovery Grant from the Natural Sciences and
  Engineering Research Council of Canada (NSERC), and by the Fonds de
  recherche du Québec – Nature et technologies (FRQNT)}

\author{Javier Esparza}{Fakultät für Informatik, Technische
  Universität München, Garching bei München,
  Germany}{esparza@in.tum.de}{https://orcid.org/0000-0001-9862-4919}{}

\author{Blaise Genest}{Univ Rennes, CNRS, IRISA, France}
       {blaise.genest@irisa.fr}{https://orcid.org/0000-0002-5758-1876}{}

\author{Martin Helfrich}{Fakultät für Informatik, Technische
  Universität München, Garching bei München, Germany}
       {helfrich@in.tum.de}{https://orcid.org/0000-0002-3191-8098}{}

\author{Stefan Jaax}{Fakultät für Informatik, Technische Universität
  München, Garching bei München, Germany}{jaax@in.tum.de}{https://orcid.org/0000-0001-5789-8091}{}

\authorrunning{M. Blondin, J. Esparza, B. Genest, M. Helfrich, and S. Jaax}

\Copyright{M. Blondin, J. Esparza, B. Genest, M. Helfrich, and S. Jaax}


\title{Succinct Population Protocols for \newline Presburger Arithmetic}

\titlerunning{Succinct Population Protocols for Presburger Arithmetic} 

\funding{\emph{Javier Esparza, Martin Helfrich, Stefan Jaax}: Supported by an
  ERC Advanced Grant (787367: PaVeS)}

\nolinenumbers

\EventEditors{Christophe Paul and Markus Bl\"{a}ser}
\EventNoEds{2}
\EventLongTitle{37th International Symposium on Theoretical Aspects of Computer Science (STACS 2020)}
\EventShortTitle{STACS 2020}
\EventAcronym{STACS}
\EventYear{2020}
\EventDate{March 10--13, 2020}
\EventLocation{Montpellier, France}
\EventLogo{}
\SeriesVolume{154}
\ArticleNo{36}

\begin{document}

\maketitle

\begin{abstract}
  In~\cite{AAE06}, Angluin~\etal\ proved that population protocols
  compute exactly the predicates definable in Presburger arithmetic (PA),
  the first-order theory of addition. As part of this result, they
  presented a procedure that translates any formula $\varphi$ of 
  quantifier-free PA with remainder predicates (which has the same
  expressive power as full PA) into a population
  protocol with $2^{\O(\poly(|\varphi|))}$ states that computes
  $\varphi$. More precisely, the number of states of the protocol is
  exponential in both the bit length of the largest coefficient in the
  formula, and the number of nodes of its syntax tree.

In this paper, we prove that every formula $\varphi$ of quantifier-free PA with remainder predicates is computable by a leaderless population protocol with $\O(\poly(|\varphi|))$ states. 
Our proof is based on several new constructions, which may be of
independent interest. Given a 
formula $\varphi$ of quantifier-free PA with remainder predicates, a first construction produces a succinct protocol (with $\O(|\varphi|^3)$ leaders) that computes $\varphi$; this completes the work initiated in \cite{BEJ18}, where we constructed such protocols for a fragment of PA. For large enough inputs, we can get rid of these leaders.
If the input is not large enough, then it is small, and we design another construction producing a succinct protocol with {\em one} leader that computes $\varphi$. 
Our last construction gets rid of this leader for small inputs.
\end{abstract}

\section{Introduction}
\newcommand{\maxcoef}[1]{\textit{mxc}(#1)}
\newcommand{\at}[1]{\textit{at}(#1)}

Population protocols~\cite{AADFP04,AADFP06} are a model of distributed
computation by indistinguishable, mobile finite-state agents, intensely investigated in recent years (see e.g.~\cite{AlistarhG18,ElsasserR18}). Initially
introduced to model networks of passively mobile sensors, they have also been applied to the analysis of chemical reactions under the name of chemical reaction networks (see \eg~\cite{SCWB08}). 

In a population protocol, a collection of agents, called a \emph{population}, randomly interact in  pairs to decide whether their initial 
configuration satisfies a given property, \eg\ whether there are initially more agents in some state $A$ than in some state $B$. 
Since agents are indistinguishable and finite-state, their 
configuration at any time moment is completely characterized by the mapping that assigns to each state the number of agents that currently populate it.  
A protocol is said to \emph{compute a predicate} 
if for every initial configuration
where the predicate holds,
the agents eventually reach consensus 1,
and they eventually reach consensus 0 otherwise.

In a seminal paper, Angluin~\etal\ proved that population protocols
compute exactly the predicates definable in Presburger arithmetic (PA)~\cite{AAE06}. As part of the result, for every Presburger predicate Angluin~\etal\ construct a leaderless protocol that computes it. 
The construction uses the quantifier elimination procedure for PA: every Presburger formula $\varphi$ can be transformed into an equivalent boolean combination of \emph{threshold predicates} of the form $ \vec{\alpha} \cdot \vec{x} > \beta$ and 
\emph{remainder predicates} of the form $\vec{\alpha} \cdot \vec{x} \equiv \beta~(\text{mod } m)$,  where $\vec{\alpha}$ is an integer vector, and $\beta, m$ are integers \cite{Haase18}. Slightly abusing language, we call the set of these boolean combinations \textit{quantifier-free Presburger arithmetic} (QFPA)\footnote{Remainder predicates cannot be directly expressed in Presburger arithmetic without quantifiers.}. Using that PA and QFPA have the same expressive power, Angluin~\etal\ first construct protocols for all threshold and remainder predicates, and then show that the predicates computed by protocols are closed under negation and conjunction. 

The construction of \cite{AAE06} is simple and elegant, but it produces large protocols. Given a formula $\varphi$ of QFPA, let $n$ be the number of bits of the largest coefficient of $\varphi$ in absolute value, and let $m$ be the number of atomic formulas of $\varphi$, respectively. The number of states of the protocols of \cite{AAE06} grows exponentially in both $n$ and $m$. In terms of $|\varphi|$ (defined as the sum of the number of variables, $n$, and $m$)
they have $\O(2^{\poly(|\varphi|)})$ states. This raises the question of whether \emph{succinct protocols} with $\O(\poly(|\varphi|))$ states exist for every formula $\varphi$ of QFPA. We give an affirmative answer by proving that every formula of QFPA has a succinct and leaderless protocol.

Succinct protocols are the state-complexity counterpart of  \emph{fast protocols}, defined as protocols running in polylogarithmic parallel time in the size of the population. Angluin \etal~showed  that every predicate has a fast protocol with a leader \cite{AngluinAE08a}, but  Alistarh \etal, based on work by Doty and Soloveichik \cite{DotyS18}, proved that  in the leaderless case some predicates need linear parallel time \cite{AlistarhAEGR17}. Our result shows that, unlike for time complexity, succinct protocols can be obtained for every QFPA formula in both the leaderless case and the case with leaders.

The proof of our result overcomes a number of obstacles. Designing succinct leaderless protocols is particularly hard for inputs with very few input agents, because there are less resources to simulate leaders. So we produce two completely different families of protocols, one for small inputs with $\O(|\varphi|^3)$ agents, and a second for large inputs with  $\Omega(|\varphi|^3)$ agents, and combine them appropriately.\smallskip

\parag{Large inputs} The family for large inputs is based on our previous work \cite{BEJ18}. However, in order to obtain leaderless protocols we need a new succinct construction for boolean
combinations of atomic predicates.  This obstacle is overcome by designing new protocols for threshold and remainder predicates that work under \emph{reversible dynamic
  initialization}.  Intuitively, agents are allowed to dynamically ``enter'' and
``leave'' the protocol through the initial states (dynamic
initialization). Further, every interaction can be undone  (reversibility), until a certain condition
is met, after which the protocol converges to the correct output
for the current input. We expect protocols with reversible dynamic initialization to prove
useful in other contexts, since they allow a protocol designer to cope with ``wrong''
non-deterministic choices. 

\smallskip\parag{Small inputs} The family of protocols for small inputs is designed from scratch. We exploit that there are few inputs of small size. 
So it becomes possible to design one protocol for
each possible size of the population, and combine them appropriately. Once the population size is
fixed, it is possible to design agents that check if they have interacted with all other agents. This
is used to simulate the \emph{concatenation operator} of sequential
programs, which allows for boolean combinations and succinct
evaluation of linear combinations. 

\smallskip\parag{Relation to previous work} In \cite{BEJ18}, we designed succinct protocols with leaders for systems of linear equations. More precisely, we constructed a protocol with $\O((m + k)(n+\log m))$ states and $\O(m(n+\log m))$ leaders that computes a given predicate $A\vec{x} \geq \vec{c}$, where $A \in \Z^{m \times k}$ and $n$ is the number of bits of the largest entry in $A$ and $\vec{c}$, in absolute value. Representing $A\vec{x} \geq \vec{c}$ as a formula $\varphi$ of QFPA, we obtain a protocol with $\O(|\varphi|^2)$ states and $\O(|\varphi|^2)$ leaders that computes $\varphi$.  However,  in \cite{BEJ18} no succinct protocols for formulas with remainder predicates are given, and the paper makes extensive use of leaders. 

\smallskip\parag{Organization} Sections \ref{sec:prelims} and \ref{sec:popprot} introduce basic notation and definitions. Section \ref{sec:main} presents the main result. Sections \ref{sec:large} and \ref{sec:small} present the constructions of the protocols for large and small inputs, respectively.
Section \ref{sec:conc} presents conclusions. For space reasons, several proofs are only sketched. Detailed proofs are given in the appendices of this paper.

\section{Preliminaries}\label{sec:prelims}
\parag{Notation}
We write $\Z$ to denote the set of integers, $\N$ to denote the set of
non negative integers $\{0, 1, \ldots\}$, $[n]$ to denote $\{1,
2, \ldots, n\}$, and $\N^E$ to denote the set of all multisets over
$E$, \ie\ unordered vectors with components labeled by
$E$. The \emph{size} of a multiset $\vec{v} \in \N^E$ is defined as
$|\vec{v}| \defeq \sum_{e \in E} \vec{v}(e)$. 
The set of all multisets over $E$ with size $s \geq 0$ is $E^\ms{s} \defeq 
\left\{\vec{v} \in \N^E \colon |\vec{v}| = s \right\}$. We sometimes write
multisets using set-like notation, \eg\ $\multiset{a, 2 \cdot b}$
denotes the multiset $\vec{v}$ such that $\vec{v}(a) = 1$, $\vec{v}(b)
= 2$ and $\vec{v}(e) = 0$ for every $e \in E \setminus
\{a, b\}$. The empty multiset $\multiset{}$ is instead denoted
$\vec{0}$ for readability. For every $\vec{u}, \vec{v} \in \N^E$, we
write $\vec{u} \geq \vec{v}$ if $\vec{u}(e) \geq \vec{v}(e)$ for every
$e \in E$. Moreover, we write $\vec{u} \mplus \vec{v}$ to denote the
multiset $\vec{w} \in \N^E$ such that $\vec{w}(e) \defeq \vec{u}(e) +
\vec{v}(e)$ for every $e \in E$. The multiset $\vec{u} \mminus
\vec{v}$ is defined analogously with $-$ instead of $+$, provided that
$\vec{u} \geq \vec{v}$.

\parag{Presburger arithmetic}
\emph{Presburger arithmetic} (PA) is the first-order theory of $\N$ with
addition, \ie\ $\mathsf{FO}(\N, +)$. For example, the PA
formula $\psi(x, y, z) = \exists x' \exists z' (x = x' + x') \land (y
= z + z') \land \neg(z' = 0)$ states that $x$ is even and that $y >
z$. It is well-known that 
for every formula of PA there is an equivalent formula of
quantifier-free Presburger arithmetic (QFPA)~\cite{Pr29}, the theory 
with syntax given by the
grammar
$$\varphi(\vec{v}) ::= \vec{a} \cdot \vec{v} > b \mid \vec{a} \cdot
\vec{v} \equiv_c b \mid \varphi(\vec{v}) \land \varphi(\vec{v}) \mid
\varphi(\vec{v}) \lor \varphi(\vec{v}) \mid \neg \varphi(\vec{v}) $$
where $\vec{a} \in \Z^X$, $b \in \Z$, $c \in \N_{\geq 2}$, and
$\equiv_c$ denotes equality modulo $c$. For example, the formula
$\psi(x,y,z)$ above is equivalent to $(x \equiv_2 0) \land (y - z \geq
1)$. Throughout the paper, we refer to any formula of QFPA, or the 
predicate $\N^X \to \{0, 1\}$ it denotes,
as a \emph{predicate}. Predicates of the form $\vec{a} \cdot \vec{v} > b$
and $\vec{a} \cdot \vec{v} \equiv_c b$ are \emph{atomic}, and they are called
\emph{threshold} and \emph{remainder} predicates respectively. The
\emph{max-norm} $\norm{\varphi}$ of a predicate $\varphi$ is the
largest absolute value among all coefficients occurring within
$\varphi$. The \emph{length} $\len(\varphi)$ of a predicate $\varphi$
is the number of boolean operators occurring within $\varphi$. The
\emph{bit length} of a predicate $\varphi$, over variables $X$, is
defined as $|\varphi| \defeq \len(\varphi) + \log\norm{\varphi} +
|X|$.
We lift these definitions to sets of predicates in the natural way: given a finite set $P$ of predicates, we define its \emph{size}
$\size(P)$ as the number of predicates in $P$, its \emph{length} as $\len(P) \defeq
\sum_{\varphi \in P} \len(\varphi)$, its \emph{norm} as $\norm{P}
\defeq \max\{\norm{\varphi} : \varphi \in P\}$, and its \emph{bit
	length} as $|P| \defeq \allowbreak \size(P) + \len(P)+\log
\norm{P}+|X|$. Note that $\len(P)=0$ iff $P$ only contains atomic
predicates.

\section{Population protocols}\label{sec:popprot}
A \emph{population protocol} is a tuple $\PP = (Q, T, L, X, I, O)$
where
\begin{itemize}
\item $Q$ is a finite set whose elements are called \emph{states};

\item $T \subseteq \{(\vec{p}, \vec{q}) \in \N^Q \times \N^Q :
  |\vec{p}| = |\vec{q}|\}$ is a finite set of \emph{transitions} containing 
  the set $\{(\vec{p}, \vec{p}) : \vec{p} \in \N^Q, |\vec{p}| = 2\}$; 

\item $L \in \N^Q$ is the \emph{leader multiset};

\item $X$ is a finite set whose elements are called \emph{input
  variables};

\item $I \colon X \to Q$ is the \emph{input mapping};

\item $O \colon Q \to \{0, 1, \bot\}$ is the \emph{output mapping}.
\end{itemize}

For readability, we often write $t \colon \vec{p}
\mapsto \vec{q}$ to denote a transition $t = (\vec{p}, \vec{q})$.
Given $\kwayvar \geq 2$, we say that $t$ is \emph{$\kwayvar$-way} if $|\vec{p}| \leq \kwayvar$.

In the standard syntax of population protocols $T$ is a subset of $\N^2 \times \N^2$, and $O \colon Q \to \{0, 1\}$. These differences are discussed at the end of this section.

\parag{Inputs and configurations} 
An \emph{input} is a multiset $\vec{v} \in \N^X$ such that $|\vec{v}| \geq 2$, and a \emph{configuration} is a multiset $C \in
\N^Q$ such that $|C| \geq 2$. 
Intuitively, a configuration rep\-resents a
population of agents where $C(q)$ denotes the number of agents in
state $q$. The \emph{initial configuration $C_{\vec{v}}$ for
input $\vec{v}$} is defined as  $C_{\vec{v}} \defeq L \mplus \multiset{\vec{v}(x) \cdot I(x) : x
  \in X}$. 

The \emph{support} and \emph{$b$-support} of a configuration $C$ are
respectively defined as $\supp{C} \defeq \{q \in Q : C(q) > 0\}$ and
$\supp{C}_b = \{q \in \supp{C} : O(q) = b\}$. The \emph{output} of a
configuration $C$ is defined as $O(C) \defeq b$ if $\supp{C}_b \neq
\emptyset$ and $\supp{C}_{\neg b} = \emptyset$ for some $b \in \{0,
1\}$, and $O(C) \defeq \bot$ otherwise. Loosely speaking, if $O(q) = \bot$ then agents
in state $q$ have no output, and a population has output $b \in \{0, 1\}$ if all agents \emph{with
output} have output $b$. 

\parag{Executions} A transition $t = (\vec{p}, \vec{q})$ is \emph{enabled} 
in a configuration $C$ if $C \geq \vec{p}$, and \emph{disabled}
otherwise. Because of our assumption on $T$, every configuration enables at least one transition.
If $t$ is enabled in $C$, then it can be \emph{fired}
leading to configuration $C' \defeq C \mminus \vec{p} \mplus \vec{q}$,
which we denote $C \trans{t} C'$. For every set of transitions $S$, we
write $C \trans{S} C'$ if $C \trans{t} C'$ for some $t \in S$. We
denote the reflexive and transitive closure of $\trans{S}$ by
$\trans{S^*}$. If $S$ is the set of all transitions of the protocol
under consideration, then we simply write $\trans{}$ and $\trans{*}$.

An execution is a sequence of configurations $\sigma = C_0 C_1 \cdots$
such that $C_i \trans{} C_{i+1}$ for every $i \in \N$. We write
$\sigma_i$ to denote configuration $C_i$. The \emph{output} of an
execution $\sigma$ is defined as follows. If there exist $i \in \N$
and $b \in \{0, 1\}$ such that $O(\sigma_i) = O(\sigma_{i+1}) = \cdots
= b$, then $O(\sigma) \defeq b$, and otherwise $O(\sigma) \defeq
\bot$.

\parag{Computations} An execution $\sigma$ is \emph{fair} if for every
configuration $D$ the following holds: $$ \text{if } |\{i \in \N :
\sigma_i \trans{*} D\}| \text{ is infinite, then } |\{i \in \N :
\sigma_i = D\}| \text{ is infinite.}$$ In other words, fairness
ensures that an execution cannot avoid a configuration forever. We say
that a population protocol \emph{computes} a predicate $\varphi \colon
\N^X \to \{0, 1\}$ if for every $\vec{v} \in \N^X$ and every fair
execution $\sigma$ starting from $C_{\vec{v}}$, it is the
case that $O(\sigma) = \varphi(\vec{v})$. Two protocols are \emph{equivalent} if they
compute the same predicate. It is known that population protocols compute precisely the 
Presburger-definable predicates~\cite{AAE06,EGLM17}. 

\begin{example}\label{ex:flock}
  Let $\PP_n = (Q, T, \vec{0}, \{x\}, I, O)$ be the protocol where $Q
  \defeq \{0, 1, 2, 3, \ldots, 2^n\}$, $I(x) \defeq 1$, $O(a) = 1
  \defiff a = 2^n$, and $T$ contains a transition, for each $a, b \in
  Q$, of the form $\multiset{a, b} \mapsto \multiset{0, a + b}$ if $a
  + b < 2^n$, and $\multiset{a, b} \mapsto \multiset{2^n, 2^n}$ if $a
  + b \geq 2^n$. It is readily seen that $\PP_n$ computes $\varphi(x)
  \defeq (x \geq 2^n)$. Intuitively, each agent stores a number,
  initially 1. When two agents meet, one of them stores the sum of
  their values and the other one stores 0, with sums capping at
  $2^n$. Once an agent reaches this cap, all agents eventually get
  converted to~$2^n$.

  Now, consider the protocol $\PP'_n = (Q', T', \vec{0}, \{x\}, I',
  O')$, where $Q' \defeq \{0, 2^0, 2^1, \ldots, 2^n\}$, $I'(x) \defeq
  2^0$, $O'(a) = 1 \defiff a = 2^n$, and $T'$ contains a transition
  for each $0 \leq i < n$ of the form $\multiset{2^i, 2^i} \mapsto
  \multiset{0, 2^{i+1}}$, and a transition for each $a \in Q'$ of the
  form $\multiset{a, 2^n} \mapsto \multiset{2^n, 2^n}$. Using similar
  arguments as above, it follows that $\PP'_n$ also computes
  $\varphi$, but more succinctly: While $\PP_n$ has $2^n + 1$ states,
  $\PP'_n$ has only $n + 1$ states.
\end{example} 

\parag{Types of protocols} A protocol $\PP = (Q, T, L, X, I, O)$ is
\begin{itemize}
\item \emph{leaderless} if $|L|=0$, and \emph{has $|L|$ leaders} otherwise;
\item \emph{$\kwayvar$-way} if all its transitions are $\kwayvar$-way;
\item\emph{simple} if there exist $\q{f}, \q{t} \in Q$ such that $O(\q{f}) = 0$, $O(\q{t}) = 1$ 
and $O(q) = \bot$ for every $q \in Q \setminus \{\q{f}, \q{t}\}$ (\ie, the output is  determined by the number of agents in $\q{f}$ and $\q{t}$.)
\end{itemize}
Protocols with leaders and leaderless protocols compute the same predicates  \cite{AAE06}. Every $\kwayvar$-way protocol can be transformed into an equivalent 2-way protocol with a polynomial increase in the number of transitions \cite{BEJ18}. Finally, every protocol can be transformed into an equivalent simple protocol with a polynomial increase in the number of states (see \Cref{appendix_output}).

\section{Main result}\label{sec:main}
The main result of this paper is the following theorem:

\begin{theorem} \label{thm:main}
  Every predicate $\varphi$ of QFPA
  can be computed by a leaderless population protocol $\PP$ with
  $\O(\poly(|\varphi|))$ states. Moreover, $\PP$ can be constructed in polynomial time.
\end{theorem}

To prove Theorem~\ref{thm:main}, we first provide a construction that
uses $\cutoffvar \in \O(|\varphi|^3)$ leaders. If there 
are at least $|\vec{v}| \geq \ell$ input agents 
$\vec{v}$ ({\em large inputs}), we will show how the protocol can be made leaderless by having agents encode both their state and the state of some leader. Otherwise, $|\vec{v}| < \ell$ ({\em
  small inputs}), and we will resort to a special construction, with a
single leader, that only works for populations of bounded size. We
will show how the leader can be simulated collectively by the agents.
Hence, we will construct succinct protocols computing $\varphi$ for
large and small inputs, respectively. Formally, we prove:

\begin{lemma} \label{lem:main}
  Let $\varphi$ be a predicate over variables $X$. There exist $\cutoffvar \in \O(|\varphi|^3)$ and leaderless protocols $\PP_{\geq \cutoffvar}$ and $\PP_{< \cutoffvar}$ with $\O(\poly(|\varphi|))$ states such that:
  \begin{enumerate}[(a)]
  \item $\PP_{\geq \cutoffvar}$ computes predicate $(|\vec{v}| \geq \cutoffvar)
    \rightarrow \varphi(\vec{v})$, and \label{itm:lem:main:large}

  \item $\PP_{< \cutoffvar}$ computes predicate $(|\vec{v}| < \cutoffvar) \rightarrow
    \varphi(\vec{v})$. \label{itm:lem:main:small}
  \end{enumerate}
\end{lemma}

\cref{thm:main} follows immediately from the lemma: it suffices to
take the conjunction of both protocols, which only yields a quadratic
blow-up on the number of states, using the classical product
construction~\cite{AADFP04}.  The rest of the paper is dedicated to
proving \cref{lem:main}.  Parts~\eqref{itm:lem:main:large}
and~\eqref{itm:lem:main:small} are shown in Sections~\ref{sec:large}
and~\ref{sec:small}, respectively.

In the remainder of the paper, whenever we claim the existence of some
protocol $\PP$, we also claim polynomial-time constructibility of
$\PP$ without mentioning it explicitly.

\section{Succinct protocols for large populations}\label{sec:large}
\newcommand{\Ipt}{\mathsf{In}}
\newcommand{\ipt}{\mathsf{in}}
\newcommand{\Opt}{\mathsf{Out}}
\newcommand{\opt}{\mathsf{out}}
\newcommand{\upc}{{\upharpoonleft}}
\newcommand{\divi}[2]{ #1 \, \div \, #2}
\newcommand{\rem}[2]{ #1  \, \text{mod} \, #2}

We show that, for every predicate $\varphi$, there exists a constant
$\cutoffvar \in \O(|\varphi|^3)$ and a succinct protocol $\PP_{\geq \cutoffvar}$ computing $(|\vec{v}| \geq \cutoffvar) \rightarrow \varphi(\vec{v})$.
Throughout this section, we say that $n \in \N$ is \emph{large} if $n \geq \cutoffvar$, and that a protocol \emph{computes $\varphi$ for large inputs} if it computes $(|\vec{v}| \geq  \cutoffvar) \rightarrow \varphi(\vec{v})$.

We present the proof in a top-down manner, by means of a chain of
statements of the form ``$A \leftarrow B$, $B \leftarrow C$,
$C \leftarrow D$, and $D$''. Roughly speaking, and using notions that will be defined in the forthcoming subsections:
\begin{itemize}
\item \cref{subsec:helpers} introduces protocols with helpers, a
  special class of protocols with leaders. The section shows:
  $\varphi$ is computable for large inputs by a succinct
  leaderless protocol~(A), if it is computable for large inputs by a succinct
  protocol with helpers~(B).

\item \cref{subsec:finsets} defines protocols that simultaneously
  compute a set of predicates. The section proves: (B) holds if the set $P$
  of atomic predicates occurring within $\varphi$ is simultaneously
  computable for large inputs by a succinct protocol with helpers~(C).

\item \cref{subsec:finat} introduces protocols with reversible dynamic
  initialization. The section shows: (C) holds if each atomic predicate of $P$
  is computable for large inputs by a succinct protocol with helpers
  and reversible dynamic initialization~(D).

\item \cref{subsec:findyn} shows that~(D) holds by exhibiting succinct
  protocols with helpers and reversible dynamic initialization that
  compute atomic predicates for large inputs.
\end{itemize}
\smallskip
Detailed proofs and some formal definitions of this section are found in \Cref{app:large}.

\subsection{From protocols with helpers to leaderless protocols}
\label{subsec:helpers}

Intuitively, a protocol with helpers is a protocol with leaders satisfying an additional
property: adding more leaders does not change the predicate computed
by the protocol. Formally, let $\PP = (Q, T, L, X, I, O)$ be a population
protocol computing a predicate $\varphi$. We say that $\PP$ is a
protocol \emph{with helpers} if for every $L' \succeq L$ the protocol
$\PP' = (Q, T, L', X, I, O)$ also computes $\varphi$, where $L'
\succeq L \; \defeq \; \forall q \in Q \colon (L'(q) = L(q)=0 \vee
L'(q) \geq L(q) > 0)$. If $|L| = \cutoffvar$, then we say that $\PP$ is a protocol \emph{ with $\cutoffvar$ helpers}.

\begin{restatable}{theorem}{thmRemoveHelpers} \label{thm:remove:helpers}
  Let $\PP = (Q, T, L, X, I, O)$ be a
  $\kwayvar$-way population protocol with $\cutoffvar$-helpers
  computing some predicate $\varphi$. There exists a 2-way leaderless
  population protocol with $\O(\cutoffvar \cdot |X| + (\kwayvar \cdot |T|+|Q|)^2)$ states that computes $(|\vec{v}|  \geq \cutoffvar) \rightarrow \varphi(\vec{v})$.
\end{restatable}

\begin{proof}[Proof sketch]
  By~\cite[Lemma~3]{BEJ18}, $\PP$ can be transformed into a 2-way
  population protocol (with helpers\footnote{Lemma~3 of~\cite{BEJ18}
    deals with leaders and not the more specific case of
    helpers. Nonetheless, computation under helpers is preserved as
    the input mapping of $\PP$ remains unchanged in the proof of the
    lemma.})  computing the same predicate $\varphi$, and with at most
  $|Q| + 3\kwayvar \cdot |T|$ states. Thus, we assume $\PP$ to be 2-way in
  the rest of the sketch.

  For simplicity, assume $X = \{x\}$ and $L = \multiset{3 \cdot q, 5
    \cdot q'}$; that is, $\PP$ has 8 helpers, and initially 3 of them
  are in state $q$, and $5$ are in $q'$. We describe a leaderless
  protocol $\PP'$ that simulates $\PP$ for every input $\vec{v}$ such
  that $|\vec{v}| \geq |L| = \cutoffvar$. Intuitively, $\PP'$ runs in two phases:
  \begin{itemize}
  \item In the first phase each agent gets assigned a number between 1
    and 8, ensuring that each number is assigned to at least one agent
    (this is the point at which the condition
	 $|\vec{v}| \geq \cutoffvar$ is needed).
	 At the end of the phase, each agent is in a state of the
    form $(x, i)$, meaning that the agent initially represented one
    unit of input for variable $x$, and that it has been assigned
    number $i$.  To achieve this, initially every agent is placed in
    state $(x, 1)$. Transitions are of the form $ \multiset{ (x, i),
      (x, i)} \mapsto \multiset{ (x, i+1), (x, i)}$ for every $1 \leq
    i \leq 7$. The transitions guarantee that all but one agent is
    promoted to $(x, 2)$, all but one to $(x, 3)$, etc. In other
    words, one agent is ``left behind'' at each step.

  \item In the second phase, an agent's state is a multiset:
    agents in state $(x, i)$ move to state
    $\multiset{I(x), q}$ if $1 \leq i \leq 3$, and to state
    $\multiset{I(x), q'}$ if $4 \leq i \leq 8$.  Intuitively, after
    this move each agent has been assigned two jobs: simultaneously
    simulate a regular agent of $\PP$ starting at state $x$,
    \emph{and} a helper of $L$ starting at state $q$ or $q'$. Since in
    the first phase each number is assigned to at least one agent,
    $\PP'$ has at least 3 agents simulating helpers in state $q$, and
    at least 5 agents simulating helpers in state $q'$. There may be
    many more helpers, but this is harmless, because, by definition,
    additional helpers do not change the computed
    predicate. \smallskip

    The transitions of $\PP'$ are designed according to this double
    role of the agents of $\PP'$. More precisely, for all multisets
    $\vec{p}, \vec{q}, \vec{p}', \vec{q}'$ of size two,
    $\multiset{\vec{p}, \vec{q}} \mapsto \multiset{\vec{p}',
      \vec{q}'}$ is a transition of $\PP'$ if{}f $(\vec{p} + \vec{q})
    \trans{} (\vec{p}' + \vec{q}')$ in $\PP$.
    \qedhere
  \end{itemize}
\end{proof}

\subsection{From multi-output protocols to protocols with helpers}
\label{subsec:finsets}

A \emph{$k$-output population protocol} is a tuple $\mathcal{Q} = (Q,
T, L, X, I, O)$ where $O \colon [k] \times Q \to \{0, 1, \bot\}$ and
$\mathcal{Q}_i \defeq (Q, T, L, X, I, O_i)$ is a population protocol
for every $i \in [k]$, where $O_i$ denotes the mapping such that
$O_i(q) \defeq O(i, q)$ for every $q \in Q$. Intuitively, since each
$\mathcal{Q}_i$ only differs by its output mapping, $\mathcal{Q}$ can
be seen as a single population protocol whose executions have $k$
outputs. More formally, $\mathcal{Q}$ \emph{computes} a set of
predicates $P = \{\varphi_1, \varphi_2, \ldots, \varphi_k\}$ if
$\mathcal{Q}_i$ computes $\varphi_i$ for every $i \in
[k]$. Furthermore, we say that $\mathcal{Q}$ is \emph{simple} if
$\mathcal{Q}_i$ is simple for every $i \in [k]$. Whenever the number
$k$ is irrelevant, we use the term \emph{multi-output population
  protocol} instead of $k$-output population protocol.

\begin{restatable}{proposition}{PropMult}
\label{prop:mult}
Assume that every finite set $A$ of {\em atomic} predicates is computed by some
$|A|$-way multi-output protocol with $\O(|A|^3)$ helpers and states, and
$\O(|A|^5)$ transitions. Every
QFPA predicate $\varphi$ is computed by some
simple $|\varphi|$-way protocol with
$\O(|\varphi|^3)$ helpers and states, and
$\O(|\varphi|^5)$ transitions.
\end{restatable}

\begin{proof}[Proof sketch]
Consider a binary tree decomposing the boolean operations of $\varphi$.
We design a protocol for $\varphi$ by induction on the height of the tree.

The case where the height is $0$, and $\varphi$ is atomic, is trivial.
We sketch the induction step for the case where the root is labeled with $\land$, that is $\varphi = \varphi_1 \land \varphi_2$, the other cases are similar.
By induction hypothesis, we have simple protocols $\PP_1,\PP_2$ computing $\varphi_1,\varphi_2$, respectively. Let $\q{t}_j,\q{f}_j$ be the output states of $\PP_j$ for $j \in \{1,2\}$ such that $O_j(\q{t}_j) = 1$ and $O_j(\q{f}_j) = 0$. We add two new states $\q{t}, \q{f}$ (the output states of the new protocol) and an additional helper starting in state $\q{f}$. To compute $\varphi_1 \land \varphi_2$ we add the following transitions for every $b_1 \in \{\q{t}_1,\q{f}_1\}, b_2 \in \{\q{t}_2,\q{f}_2\}$, and $b \in \{\q{t},\q{f}\}$:
$\multiset{b_1,b_2,b} \mapsto \multiset{b_1,b_2,\q{t}}$ if $b_1 = \q{t}_1 \land b_2 = \q{t}_2$, and $\multiset{b_1,b_2,b} \mapsto \multiset{b_1,b_2,\q{f}}$ otherwise.
The additional helper computes the conjunction as desired.
\end{proof}

\subsection{From reversible dynamic initialization to multi-output protocols}
\label{subsec:finat}

Let $P = \{\varphi_1, \ldots, \varphi_k\}$ be a set of $k \geq 2$
atomic predicates of arity $n \geq 1$ over a set $X = \{x_1, \ldots, x_n\}$ of variables. 
We construct a multi-output protocol $\PP$ for $P$ of size 
$\poly(|\varphi_1|+ \cdots + |\varphi_k|)$.

Let $\PP_1, \ldots, \PP_k$ be protocols for 
$\varphi_1, \ldots, \varphi_k$. Observe that $\PP$ cannot be a ``product protocol'' that executes $\PP_1, \ldots, \PP_k$ synchronously. Indeed, the states of such a $\PP$ are tuples $(q_1, \ldots, q_k)$ of states of $\PP_1, \ldots, \PP_k$,  and so $\PP$ would have exponential size in $k$.
Further, $\PP$ cannot execute $\PP_1, \ldots, \PP_k$  asynchronously in parallel, because, given an input $\vec{x} \in \N^n$, it must dispatch $k \cdot \vec{x}$ agents ($\vec{x}$ to the input states of each $\PP_j$), but it only has $\vec{x}$. 
Such a $\PP$ would need $(k-1)|\vec{x}|$ helpers, which is not possible, because a protocol of size  $\poly(|\varphi_1|+ \cdots + |\varphi_k|)$ can only use $\poly(|\varphi_1|+ \cdots + |\varphi_k|)$ helpers, whatever the input $\vec{x}$. 

The solution is to use a more sophisticated parallel asynchronous computation. Consider two copies of inputs, denoted
$\overline{X} = \{\overline{x}_1, \ldots, \overline{x}_n\}$ and $\underline{X} = \{\underline{x}_1, \ldots, \underline{x}_n\}$.
For each predicate $\varphi$ over $X$, 
consider predicate $\tilde{\varphi}$ over 
$\overline{X} \cup \underline{X}$
satisfying $\tilde{\varphi}(\overline{\vec{x}},\underline{\vec{x}}) = 
\varphi(k\overline{\vec{x}}+\underline{\vec{x}})$
for every $(\overline{\vec{x}},\underline{\vec{x}}) \in \N^{\overline{X} \cup \underline{X}}$.  
We obtain 
$\tilde{\varphi}(\overline{\vec{x}},\underline{\vec{x}}) = 
\varphi(\vec{x})$ whenever $k \overline{\vec{x}} + \underline{\vec{x}} = \vec{x}$, \eg\ for
$\overline{\vec{x}}:=  {\lfloor}\frac{{\vec{x}}}{k}{\rfloor}$ and $\underline{\vec{x}} :=\rem{{\vec{x}}}{k}$. 
With this choice, $\PP$ needs to dispatch a total of 
$k \left( |\overline{\vec{x}} + \underline{\vec{x}}| \right)
\leq |\vec{x}| + n \cdot (k-1)^2$ 
agents to compute $\tilde{\varphi}_1(\overline{\vec{x}},\underline{\vec{x}}), \ldots, \tilde{\varphi}_k(\overline{\vec{x}},\underline{\vec{x}})$. 
That is, $n \cdot (k-1)^2$ helpers are sufficient to compute $\PP$. 
Formally, we define $\tilde{\varphi}$ in the following way:

\begin{quote}
For $\displaystyle \varphi(\vec{x}) = \left(\sum_{i=1}^n \alpha_i x_i > \beta \right)$, we define $\displaystyle \tilde{\varphi}(\overline{\vec{x}}, \underline{\vec{x}}) := \left(\sum_{i=1}^n (k \cdot \alpha_i) \overline{x}_i + \alpha_i\underline{x}_i > \beta \right)$
\end{quote}
\noindent and similarly for modulo predicates.
For instance, if $\varphi(x_1,x_2)=3 x_1 - 2x_2 > 6$
and $k=4$, then $\tilde{\varphi}(\overline{x}_1,\underline{x}_1,\overline{x}_2,\underline{x}_2)=12 \overline{x}_1 + 3 \underline{x}_1 -
8 \overline{x}_2 - 2 \underline{x}_2 > 6$. 
As required, $\tilde{\varphi}(\overline{\vec{x}},\underline{\vec{x}})=
\varphi(k\overline{\vec{x}}+\underline{\vec{x}})$. 

Let us now describe how
the protocol $\PP$ computes $\tilde{\varphi}_1(\overline{\vec{x}}, \underline{\vec{x}}), \ldots, \tilde{\varphi}_k(\overline{\vec{x}}, \underline{\vec{x}})$.
Let $\tilde{\PP}_1, \ldots, \tilde{\PP}_k$ be protocols computing $\tilde{\varphi}_1, \ldots, \tilde{\varphi}_k$.  
Let $X= \{\q{x_1}, \ldots, \q{x_n}\}$ be the input states of $\PP$, 
and let
$\q{\overline{x}^j_1}, \ldots, \q{\overline{x}^j_n}$ and $\q{\underline{x}^j_1}, \ldots, \q{\underline{x}^j_n}$ be the input states of $\tilde{\PP}_j$
for every $1 \leq j \leq k$. 
Protocol $\PP$ 
repeatedly chooses an index $1 \leq i \leq n$, and executes one of these two actions: (a)~take $k$ agents from $\q{x_i}$, and dispatch them to $\q{\overline{x}^1_i}, \ldots, \q{\overline{x}^k_i}$ (one agent to each state); or (b)~take one agent from $\q{x_i}$ and $(k-1)$ helpers, and dispatch them to $\q{\underline{x}^1_i}, \ldots, \q{\underline{x}^k_i}$. The index and the action are chosen nondeterministically.
Notice that if 
for some input $\q{x}_i$, 
all $\ell$ agents of $\q{x_i}$ are dispatched, then 
$k \q{\overline{x}^j_i} + \q{\underline{x}^j_i}=\ell$ for all $j$. 
If all agents of $\q{x_i}$ are dispatched for every $1 \leq i \leq n$, then we say that the {\em dispatch is correct}.

The problem is that, because of the nondeterminism, 
the dispatch may or may not be correct.
Assume, \eg, that $k=5$ and $n=1$. Consider the input $x_1=17$, and assume that $\PP$ has $n \cdot (k-1)^2=16$ helpers. $\PP$ may correctly dispatch
$\overline{x}_1 =  {\lfloor}\frac{17}{5}{\rfloor}=3$ agents to each of $\q{\overline{x}^1_1}, \ldots, \q{\overline{x}^1_5}$
and $\underline{x}_1 = (\rem{17}{5})=2$ to each of $\q{\underline{x}^1_1}, \ldots, \q{\underline{x}^1_5}$; this gives a total of $(3+2)\cdot 5 = 25$ agents, consisting of the $17$ agents for the input plus $8$
helpers. However, it may also wrongly dispatch $2$ agents to each of $\q{\overline{x}^1_1}, \ldots, \q{\overline{x}^1_5}$ and
$4$ agents to each of $\q{\underline{x}^1_1}, \ldots, \q{\underline{x}^1_5}$, with a total of $(2+4) \cdot 5 = 30$ agents, consisting of $14$ input agents plus $16$ helpers. In the second case, each $\PP_j$
wrongly computes $\tilde{\varphi}_j(2, 4) = \varphi_j(2\cdot 5 + 4) = \varphi_j(14)$, instead of the correct value $\varphi_j(17)$.

To solve this problem we ensure that $\PP$ can always recall agents already dispatched 
to $\tilde{\PP}_1, \ldots, \tilde{\PP}_k$ as long as the dispatch is not yet correct. 
This allows $\PP$ to ``try out'' dispatches until it dispatches correctly, which eventually happens by fairness. 
For this we design $\PP$ so that (i) the atomic protocols $\tilde{\PP}_1, \ldots, \tilde{\PP}_k$ can work with inputs agents that arrive over time ({\em dynamic initialization}), and (ii) $\tilde{\PP}_1, \ldots, \tilde{\PP}_k$ can always return to their initial configuration and send agents back to $\PP$, unless the  dispatch is correct ({\em reversibility}).
To ensure that $\PP$ stops redistributing after dispatching a correct distribution, it suffices to replace each reversing transition $\vec{p} \mapsto \vec{q}$ by transitions $\vec{p} + \multiset{\q{x_i}} \mapsto \vec{q} + \multiset{\q{x_i}}$, one for each $1 \leq i \leq n$: All these transitions become disabled when $\q{x_1}, \ldots, \q{x_n}$ are not populated. 

\parag{Reversible dynamic initialization} Let us now formally introduce the class of \emph{protocols with reversible dynamic initialization} that enjoys all properties needed for our construction. A simple protocol with \emph{reversible dynamic initia\-li\-zation} (\emph{RDI-protocol} for short) is a tuple $\PP = (Q, T_\infty, T_\dagger, L, X, I , O)$, where $\PP_\infty = (Q, T_\infty, L, X, I , O)$ is a simple population protocol, and $T_\dagger$ is the set of transitions making the system reversible, called the \emph{RDI-transitions}.

Let $T \defeq T_\infty \cup T_\dagger$, and let $\Ipt \defeq \{\ipt_x
: x \in X\}$ and $\Opt \defeq \{\opt_x : x \in X\}$ be the sets of
\emph{input} and \emph{output transitions}, respectively, where
$\ipt_x \defeq (\vec{0}, \multiset{I(x)})$ and $\opt_x \defeq
(\multiset{I(x)}, \vec{0})$. An \emph{initialization sequence} is a
finite execution $\pi \in \left( T \cup \Ipt \cup \Opt\right)^*$ from
the \emph{initial configuration} $L'$ with $L' \succeq L$. The \emph{effective input} of
$\pi$ is the vector $\vec{w}$ such that $\vec{w}(x) \defeq
|\pi|_{\ipt_x} - |\pi|_{\opt_x}$ for every $x \in X$. Intuitively, a
RDI-protocol starts with helpers only, and is dynamically initialized
via the input and output transitions.

\newcommand{\DiEq}[1]{\left[#1\right]}

Let $\q{f}, \q{t} \in Q$ be the unique states of $\PP$ with $O(\q{f}) = 0$ and
$O(\q{t}) = 1$. For every configuration $C$, let $\DiEq{C} \defeq \left\{C' : C'(\q{f}) + C'(\q{t}) = C(\q{f}) + C(\q{t}) \text{ and }
C'(q) = C(q) \text{ for all } q \in Q \setminus \{\q{f}, \q{t}\} \right\}$. Intuitively, all configurations $C' \in
\DiEq{C}$ are equivalent to $C$ in all but the output states.

An RDI-protocol is required to be \emph{reversible}, that is
for every initialization sequence $\pi$ with effective input $\vec{w}$,
and such that $L' \trans{\pi} C$ for some $L' \succeq L$, the following holds:
\begin{itemize}
\item if $C \trans{T^*} D$ and $D' \in \DiEq{D}$, then $D' \trans{T^*} C'$ for some $C' \in \DiEq{C}$, and

\item $C(I(x)) \leq \vec{w}(x)$ for all $x \in X$.
\end{itemize}
Intuitively, an RDI-protocol can never have more agents in an input state than the effective number of agents it received via the input and output transitions. 
Further, an RDI-protocol can always reverse all sequences that do not contain input or output transitions.
This reversal does not involve the states $\q{f}$ and $\q{t}$, which have a special role as output states. Since RDI-protocols have a default output, we need to ensure that the default output state is populated when dynamic initialization ends, and reversal for $\q{f}$ and $\q{t}$ would prevent that.

An RDI-protocol $\PP$ \emph{computes} $\varphi$ if for every initialization sequence $\pi$ with effective input $\vec{w}$ such that $L' \trans{\pi} C$ for some $L' \succeq L$, the standard population protocol $\PP_\infty$ computes $\varphi(\vec{w})$ from $C$ (that is with $T_\dagger$ disabled). Intuitively, if the dynamic initialization terminates, the RDI-transitions $T_\dagger$ become disabled, and then the resulting standard protocol $\PP_\infty$ converges to the output corresponding to the dynamically initialized input.

\begin{restatable}{theorem}{thmAtomicMultHelpers} \label{thm:atomic:mult:helpers}
Assume that for every atomic predicate $\varphi$, there exists a
$|\varphi|$-way RDI-protocol with
$\O(|\varphi|)$ helpers,
$\O(|\varphi|^2)$ states and
$\O(|\varphi|^3)$ transitions
that computes $\varphi$.
For every finite set  $P$ of atomic predicates, there exists a $|P|$-way
simple {\em multi-output} protocol, with
$\O(|P|^3)$  helpers and states, and
$\O(|P|^5)$ transitions, that computes $P$.
\end{restatable}

\subsection{Atomic predicates under reversible dynamic initialization}
\label{subsec:findyn}

Lastly, we show that atomic predicates are succinctly computable
by RDI-protocols:

\begin{restatable}{theorem}{thmAtomicHelpers} \label{thm:atomic:helpers}
Every atomic predicate $\varphi$ over variables $X$ can be computed by a simple $|\varphi|$-way population protocol with reversible dynamic initialization that has
$\O(|\varphi|)$ helpers,
$\O(|\varphi|^2)$ states, and
$\O(|\varphi|^3)$
transitions.
\end{restatable}

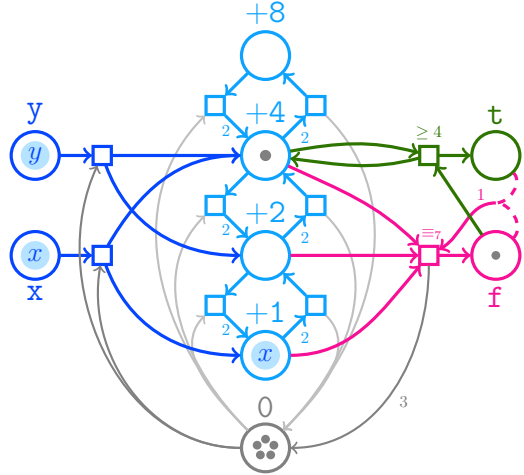
\begin{figure}[htb]
	\vspace*{-15pt} 
	\centering\definecolor{lipics-triadic1}{HTML}{FDC711}
\definecolor{lipics-triadic2}{HTML}{11E7FD}
\definecolor{lipics-triadic3}{HTML}{FD119E}
\definecolor{lipics-triadic4}{HTML}{336600}
\colorlet{colNeutral}{gray}
\colorlet{colPos}{lipics-triadic2!70!blue}
\colorlet{colModulo}{lipics-triadic3!90!purple}
\colorlet{colEqual}{lipics-triadic4!90!green}
\colorlet{colOut}{gray}
\colorlet{colPosDark}{blue!70!lipics-triadic2}
\colorlet{colNegDark}{purple!90!lipics-triadic3}
\begin{tikzpicture}[->, node distance=0.7cm, auto, very thick, scale=0.6, transform shape, font=\huge]
  \tikzset{every place/.style={inner sep=1pt, minimum size=30pt,
      token distance=50pt}}
  \tikzset{storage/.style={every token/.append style={inner
        sep=1pt, minimum size=18pt, font=\LARGE}, token
      distance=13pt}}
  
  \tikzset{every node/.style={colPos}}
  \tikzset{every token/.append style={color=colPos!30, text=colPosDark, minimum size=7pt}}
  \node[place, label=$\q{+1}$] (p1) {}
  	[children are tokens, storage]
  	child {node [token] {$x$}}
  ;
  \node[transition]            (pu1) [above right= of p1]  {};
  \node[transition]            (pd1) [above left=  of p1]  {};
  \node[place, label=$\q{+2}$] (p2)  [above right= of pd1] {};
  \node[transition]            (pu2) [above right= of p2]  {};
  \node[transition]            (pd2) [above left=  of p2]  {};
  
  \tikzset{every token/.append style={color=colNeutral}}
  \node[place, label=$\q{+4}$,tokens=1] (p4)  [above right= of pd2] {};
  \tikzset{every token/.append style={color=colPos!30, text=colPosDark}}
  \node[transition]            (pu4) [above right= of p4]  {};
  \node[transition]            (pd4) [above left=  of p4]  {};
  \node[place, label=$\q{+8}$] (p8)  [above right= of pd4] {};
  \tikzset{every node/.style={colNeutral}}
  \tikzset{every token/.append style={color=colNeutral, text=colNeutral, minimum size=5.5pt}}
  \node[place, tokens=5, label=$\q{0}$] (z) [below= 27pt of p1] {};
  
  \tikzset{every node/.style={colPosDark}}
  \tikzset{every token/.append style={color=colPos!30, text=colPosDark}}
  \node[transition]           (tx) [left=80pt of p2] {};
  \node[place, label=below:$\q{x}$] (x)  [left=     of tx] {}
  	[children are tokens, storage]
  	child {node [token] {$x$}}
  ; 
  \node[transition]           (ty) [left=80pt of p4] {};
  \node[place, label=$\q{y}$] (y)  [left=of ty]      {}
  	[children are tokens, storage]
  	child {node [token] {$y$}}
  ;
   
  \tikzset{every node/.style={colModulo}}
  \tikzset{every token/.append style={colModulo!30, text=colModulo}}
  \node[transition, label=above:{\large $\ \equiv_7$}]           (tmod) [right=80pt of p2] {};
  
  \tikzset{every node/.style={colEqual}}
  \tikzset{every token/.append style={colEqual!30, text=colEqual}}
  \node[transition, label=above:{\large $\geq 4$}]           (tequal) [right=80pt of p4] {};
 
  \node[place, label=$\q{t}$] (t)  [right=     of tequal] {}; 
  \tikzset{every node/.style={colModulo}}
  \tikzset{every token/.append style={color=colNeutral, text=colNeutral}}
  \node[place, label=below:$\q{f}$] (f)  [right=     of tmod,tokens=1] {}; 
  \coordinate [below=0.5cm of t] (TorF) {};

  \tikzset{every node/.style={colPos}}
  \path[->, colPos, very thick, font=\large]
  (p1)  edge[] node[swap] {$2$} (pu1)
  (pu1) edge[] node[swap]             {}    (p2)
  (p2)  edge[] node[swap] {$2$} (pu2)
  (pu2) edge[] node[swap]             {}    (p4)
  (p4)  edge[] node[swap] {$2$} (pu4)
  (pu4) edge[] node[swap]             {}    (p8)
  ;
  \path[->, colPos, very thick, font=\large]
  (p8)  edge[] node[swap]             {}    (pd4)
  (pd4) edge[] node[swap, yshift=5pt] {$2$} (p4)
  (p4)  edge[] node[swap]             {}    (pd2)
  (pd2) edge[] node[swap, yshift=5pt] {$2$} (p2)
  (p2)  edge[] node[swap]             {}    (pd1)
  (pd1) edge[] node[swap, yshift=5pt] {$2$} (p1)
  ;
  
  \path[->, colNeutral!50, thick]
  (pu1) edge[out=-45, in=45] node[swap] {} (z)
  (pu2) edge[out=-45, in=45] node[swap] {} (z)
  (pu4) edge[out=-45, in=45] node[swap] {} (z)
  ;
  \path[->, colNeutral!50, thick]
  (z) edge[in=-135, out=135] node[swap] {} (pd4)
  (z) edge[in=-135, out=135] node[swap] {} (pd2)
  (z) edge[in=-135, out=135] node[swap] {} (pd1)
  ;
  
  \path[->, colPosDark, very thick]
  (x)  edge[]                node[] {} (tx)
  (tx) edge[out=60,  in=180] node[] {} (p4)
  (tx) edge[out=-70, in=180] node[] {} (p1)
  ;
  \path[->, colNeutral, thick]
  (z) edge[out=180, in=-100] node[] {} (tx)
  ;
  
  \path[->, colPosDark, very thick]
  (y)  edge[]                node[] {} (ty)
  (ty) edge[out=0,  in=180]  node[] {} (p4)
  (ty) edge[out=-70, in=180] node[] {} (p2)
  ;
  \path[->, colNeutral, thick]
  (z) edge[out=180, in=-115] node[] {} (ty)
  ;
  
  \tikzset{every node/.style={colModulo}}
  \path[->, colModulo, very thick, font=\large]
  (p4) edge[out=-25, in=135]  node[] {} (tmod)
  (p2) edge[]                 node[] {} (tmod)
  (p1) edge[out=0, in=-130]   node[] {} (tmod)
  (t)  edge[-, out=-45, in=0, dashed]			  node[] {} (TorF)
  (f)  edge[-, out=45, in=0, dashed]  		  node[] {} (TorF)
  (TorF)  edge[out=180, in=30]  node[swap, yshift=9pt, xshift=16pt] {$1$} (tmod)
  (tmod)  edge   			  node[] {} (f)
  ;
  \path[->, colModulo, thick]
  (tmod) edge[out=-90, in=0,colNeutral] node[colNeutral, font=\large] {$3$} (z)
  ;
  
  \tikzset{every node/.style={colEqual}}
  \path[->, colEqual, very thick]
  (p4) edge[out=10, in=165]                node[] {} (tequal)
  (f) edge                				   node[] {} (tequal)
  (tequal) edge[out=-165, in=-7]           node[] {} (p4)
  (tequal) edge           				   node[] {} (t)
  ;
\end{tikzpicture}%
%
	\vspace*{-6pt} 
	\caption{Partial representation of the protocol computing $5x + 6y
		\geq 4 \pmod{7}$ as a Petri net, where places (circles), transitions
		(squares) and tokens (smaller filled circles) represent
		respectively states, transitions and agents. Non-helper agents
		remember their input variable (labeled here within tokens).
		The depicted configuration is obtained from input $x = 2$, $y = 1$ by
		firing the bottom leftmost transition (dark blue).}\label{fig:petri:remainder}
\end{figure}



The protocols for arbitrary threshold and remainder predicates satisfying the conditions
of \cref{thm:atomic:helpers}, and their correctness proofs, are given in \Cref{app:threshold}. Note that the threshold
protocol is very similar to the protocol for linear inequalities given in Section 6 of \cite{BEJ18}. Thus, as an
example, we will instead describe how to handle the remainder predicate $5x - y \equiv_7 4$. Note, that the
predicate can be rewritten as $\left(5x + 6y \geq 4 \pmod{7}\right)\land \left(5x + 6y \not\geq 5
\pmod{7}\right)$. As we can handle negations and conjunctions separately in \cref{subsec:finsets}, we will now
explain the protocol for $\varphi \defeq 5x + 6y \geq 4 \pmod{7}$. The protocol is partially depicted in
\cref{fig:petri:remainder} using Petri net conventions for the graphical representation.

The protocol has an \emph{input state} $\q{x}$ for each variable $x \in X$,
\emph{output states} $\q{f}$ and $\q{t}$, a \emph{neutral state} $\q{0}$, and
\emph{numerical states} of the form $\q{+2^i}$ for every $0 \leq i
\leq n$, where $n$ is the smallest number such that $2^n > \norm{\varphi}$.
Initially, (at least) one helper is set to $\q{f}$ and (at least) $2n$ helpers set to $\q{0}$.
In order to compute $5x + 6y \geq 4 \pmod{7}$ for $x:= r$ and $y:= s$, we
initially place $r$ and $s$ agents in the states $\q{x}$ and $\q{y}$, i.e.,
the agents in state $\q{x}$ encode the number $r$ in unary, and similarly for $\q{y}$.
The blue transitions on the left of \cref{fig:petri:remainder} ``convert'' each agents in input states to a binary
representation of their corresponding coefficient.
In our example, agents in state $\q{x}$ are converted to $\vec{a}(x) = 5 = 0101_2$ by putting one agent in $\q{4}$ and another one in $\q{1}$.
Since two agents are needed to encode $5$, the transition ``recruits'' one helper from state $\q{0}$.
Observe that, since the inputs can be arbitrarily large, but a protocol can only use a constant number of helpers, the protocol must reuse helpers in order to convert all agents in input states.
This happens as follows.
If two agents are in the same power of two, say $\q{+2^i}$, then one of them can be ``promoted'' to $\q{+2^{i+1}}$, while the other one moves to state $\q{0}$, ``liberating'' one helper.
This allows the agents to represent the overall value of all converted agents in the most efficient representation.
That is, from any configuration, one can always reach a configuration where there is at most one agent in each place $2^0, \ldots, 2^{n-1}$, there are at most the number of agents converted from input places in place $2^n$, and hence there are at least $n$ agents in place $0$, thus ready to convert some agent from the input place.
Similar to promotions, ``demotions'' to smaller powers of two can also happen.
Thus, the agents effectively shift through all possible binary representations of the overall value of all converted agents.
The $\equiv_7$ transition in \cref{fig:petri:remainder} allows 3 agents in states $\q{4}$, $\q{2}$ and $\q{1}$ to ``cancel out'' by moving to state $\q{0}$, and it moves the
output helper to $\q{f}$.
Furthermore, there are RDI-transitions that allow to revert the effects of conversion and cancel transitions.
These are not shown in \cref{fig:petri:remainder}.

We have to show that this protocol computes $\varphi$ under reversible dynamic initialization.
First note, that while dynamic initialization has not terminated, all transitions have a corresponding reverse transition.
Thus, it is always possible to return to wrong initial configurations.
However, reversing the conversion transitions can create more agents in input states than the protocol effectively received.
To forbid this, each input agent is ``tagged'' with its variable (see tokens in \cref{fig:petri:remainder}).
Therefore, in order to reverse a conversion transitions, the original input agent is needed.
This implies, that the protocol is reversible.

Next, we need to argue that the protocol without the RDI-transitions computes $\varphi$ once the dynamic
initialization has terminated.
The agents will shift through the binary representations of the overall value.
Because of fairness, the $\equiv_7$ transition will eventually reduce the overall value to at most $6$.
There is a $\geq 4$-transition which detects the case where the final value is at least $4$ and moves the output helper from $\q{f}$ to state $\q{t}$. Notice that whenever transition $\equiv_7$ occurs, we reset the output by moving the output helper to state $\q{f}$.

\section{Succinct protocols for small populations}\label{sec:small}

We show that for every predicate $\varphi$
and constant $\cutoffvar = \O(|\varphi|^3)$,
there exists a succinct protocol $\PP_{< \cutoffvar}$ that computes the predicate $(|\vec{v}| < \cutoffvar) \rightarrow \varphi(\vec{v})$.
In this case, we say that $\PP_{< \cutoffvar}$ \emph{computes $\varphi$ for small inputs}. Further, we say that a number $n \in \N$
(resp. an input $\vec{v}$) is small with respect to $\varphi$ if  $n \leq \cutoffvar$ (resp. $|\vec{v}|\leq \cutoffvar$).
We present the proof strategy in a top-down manner.

\begin{itemize}
\item Section \ref{subsec:fixedsize} proves:  There is a succinct leaderless protocol $\PP$ that computes $\varphi$ for  small inputs (A), if for every small $n$ some succinct  protocol $\PP_n$ computes $\varphi$ for all inputs of size $n$ (B). Intuitively, constructing a succinct protocol for all small inputs reduces to the simpler problem of constructing a succinct protocol for all small inputs of a fixed size.
\item Section \ref{subsec:halting} introduces halting protocols. It shows: There is a succinct protocol that computes $\varphi$ for inputs of size $n$, if for every \emph{atomic} predicate $\psi$ of $\varphi$ some halting succinct protocol computes $\psi$ for inputs of size $n$ (C). Thus, constructing protocols for arbitrary predicates reduces to constructing \emph{halting} protocols for atomic predicates.
\item Section \ref{subsec:atomic} proves (C). Given a threshold or remainder predicate $\varphi$ and a small $n$, it shows how to construct a succinct halting protocol that computes $\varphi$ for inputs of size $n$.
\end{itemize}
\smallskip
Detailed proofs for this section can be found in \Cref{app:small}.

\subsection{From fixed-sized protocols with one leader to leaderless protocols}
\label{subsec:fixedsize}

\renewcommand{\cond}[2]{( #1 \mid #2 )}

We now define when a population protocol computes a predicate {\em for inputs of a fixed size}.
Intuitively,
it should compute the correct value for every initial configurations of this size;
for inputs of other sizes, the protocol may converge to the wrong result, or may not converge.

\begin{definition}
Let $\varphi$ be a predicate and let $i \geq 2$.
A protocol $\PP$ \emph{computes $\varphi$ for inputs of size $i$}, denoted \emph{``$\PP$ computes $\cond{\varphi}{i}$}'', if for every input $\vec{v}$ of size $i$, every fair execution of $\PP$ starting at $C_{\vec{v}}$ stabilizes to $\varphi(\vec{v})$.
\end{definition}

We show that if, for each small number $i$, some succinct protocol computes $\cond{\varphi}{i}$, then there is a single succinct protocol that computes $\varphi$ for all small inputs.

\begin{restatable}{theorem}{propRemovingLeaders}\label{thm:removing:leaders}
Let $\varphi$ be a predicate over a set of variables $X$, and let $\cutoffvar \in \N$.
Assume that for every $i \in \{2,3, \ldots, \cutoffvar-1\}$, there exists a protocol with at most one leader and at most $m$ states that computes $\cond{\varphi}{i}$. Then, there is a leaderless population protocol with $\O(\cutoffvar^4 \cdot m^2 \cdot |X|^3)$ states that computes $(\vec{x} <  \cutoffvar) \rightarrow \varphi(\vec{x})$.
\end{restatable}

\begin{proof}[Proof sketch]
Fix a predicate $\varphi$ and $\cutoffvar \in \N$. For every
$2 \leq i < \cutoffvar$, let $\PP_i$ be a protocol computing
$\cond{\varphi}{i}$. We describe the protocol $\PP=(Q, T, X, I, O)$ that computes $(\vec{x} \geq \cutoffvar) \lor \varphi(\vec{x}) \equiv (\vec{x} <  \cutoffvar) \rightarrow \varphi(\vec{x})$.
The input mapping $I$ is the identity. During the computation, agents never forget their initial state -- that is, all successor states of an agent are annotated with their initial state. The protocol initially performs a leader election. Each provisional leader stores how many agents it has ``knocked out'' during the leader election in a counter from $0$ to $\cutoffvar-1$. After increasing the counter to a given value $i<\cutoffvar$, it resets the state of $i$ agents and itself to the corresponding initial state of $\PP_{i+1}$, annotated with $X$, and initiates a simulation of $\PP_{i+1}$.
    When the counter of an agent reaches $\cutoffvar-1$, the agent knows that the population size must be $\geq \cutoffvar$, and turns the population into a permanent $1$-consensus.
    Now, if the population size $i$ is smaller than $\cutoffvar$, then eventually a leader gets elected who resets the population to the initial population of $\PP_i$. Since $\PP_i$ computes $\cond{\varphi}{i}$, the simulation of $\PP_i$ eventually yields the correct output.
\end{proof}

\subsection{Computing boolean combinations of predicates for fixed-size inputs}
\label{subsec:halting}

We want to produce a population protocol $\PP$
for a boolean combination $\varphi$ of atomic predicates $(\varphi_i)_{i \in [k]}$ for which we have population protocols $(\PP_i)_{i \in [k]}$.
As in Section~\ref{subsec:finat}, we cannot use a standard ``product protocol'' that executes $\PP_1, \ldots, \PP_k$ synchronously because
the number of states would be exponential in $k$.
Instead, we want to simulate the \emph{concatenation} of
$(\PP_i)_{i \in [k]}$. However, this is only possible if for all $i \in [k]$, the executions of $\PP_i$ eventually ``halt'', i.e.
some agents are eventually certain that the output of the protocol
will not change anymore, which is not the case in general population protocols.
For this reason we restrict our attention to ``halting'' protocols.

\begin{definition}
Let $\PP$ be a simple protocol with output states $\q{f}$ and $\q{t}$. We say that $\PP$ is a \emph{halting protocol} if every configuration $C$ reachable from an initial configuration satisfies:
\begin{itemize}
\item $C(\q{f}) = 0 \lor C(\q{t}) = 0$,

\item $C \trans{*} C' \land C(q) > 0 \Rightarrow C'(q) > 0$ for every $q \in \{\q{f}, \q{t}\}$ and every configuration $C'$.
\end{itemize}
\end{definition}

Intuitively, a halting protocol is a simple protocol in which states $\q{f}$ and $\q{t}$ behave like ``final states'': If an agent reaches $q \in \{\q{f}, \q{t}\}$, then the agent stays in $q$ forever.
In other words, the protocol reaches consensus $0$ (resp.\ $1$) if{}f an agent ever reaches $\q{f}$ (resp.\ $\q{t}$).

\begin{restatable}{theorem}{thmHaltingConjDis}\label{lemma:halting-conj-disj}
Let $k, i \in \N$.
Let $\varphi$ be a boolean combination of atomic predicates $(\varphi_j)_{j \in [k]}$.
Assume that for every $j \in [k]$, there is a  simple halting  protocol $\PP_j=(Q_j, L_j, X, T_j, I_j, O_j)$ with one leader computing $\cond{\varphi_j}{i}$. Then there exists
a simple halting protocol $\PP$ that computes $\cond{\varphi}{i}$,
with one leader and $\mathcal{O}\left(|X| \cdot  \left(\len(\varphi) +|Q_1| + \ldots + |Q_k|\right) \right)$ states.
\end{restatable}

\begin{proof}[Proof sketch]
\newcommand{\pres}[1]{\text{pres}(#1)}
We only sketch the construction
for $\varphi=\varphi_1 \wedge \varphi_2$.
The main intuition is that, since $\PP_1$ and $\PP_2$ are halting, we can construct a protocol that, given an input $\vec{v}$, first simulates $\PP_1$ on $\vec{v}$, and, after $\PP_1$ halts, either halts if $\PP_1$ converges to $0$, or simulates $\PP_2$ on $\vec{v}$ if $\PP_1$ converges to $1$.
Each agent remembers in its state the input variable it corresponds to, in order to simulate $\PP_2$ on $\vec{v}$.
\end{proof}

\subsection{Computing atomic predicates for fixed-size inputs}
\label{subsec:atomic}

We describe a halting protocol that computes a given threshold predicate for fixed-size inputs.

\begin{restatable}{theorem}{thmHaltingThreshold}\label{thm:halting:threshold}
Let $\varphi(\vec{x}, \vec{y}) \defeq  \vec{\alpha} \cdot \vec{x} - \vec{\beta}\cdot \vec{y} > 0$.
For every $i \in \N$, there exists a halting protocol with one leader and $\O(i^2 (|\varphi| + \log i)^3)$
states that computes $\cond{\varphi}{i}$.
\end{restatable}

\newcommand{\GreaterSum}{\sf{Greater}\text{-}\sf{Sum}}

\noindent We first describe a sequential algorithm $\GreaterSum(\vec{x}, \vec{y})$, that
for every input $\vec{x}, \vec{y}$ satisfying $|\vec{x}|+|\vec{y}|=i$
decides whether $\vec{\alpha} \cdot \vec{x} - \vec{\beta}\cdot \vec{y} > 0$ holds. Then we simulate $\GreaterSum$
by means of a halting protocol with $i$ agents.

Since each agent can only have $\O(\log i + \log |\varphi|)$ bits of memory (the logarithm of the number of states),
$\GreaterSum$ must use at most $\O(i \cdot (\log i + \log |\varphi|))$ bits of memory, otherwise it cannot be
simulated by the agents. Because of this requirement, $\GreaterSum$ cannot just compute, store, and then compare
$\vec{\alpha} \cdot \vec{x}$ and $\vec{\beta}\cdot \vec{y}$; this uses too much memory.

$\GreaterSum$ calls procedures $\textit{Probe}_1(j)$ and $\textit{Probe}_2(j)$ that return the $j$-th bits of
$\vec{\alpha} \vec{x}$ and $\vec{\beta} \vec{y}$,  respectively, where $j=1$ is the most significant bit. Since $|\vec{x}| \leq  i$,
and the largest constant in $\vec{\alpha}$ is  at most $||\varphi||$, we have $\vec{\alpha} \cdot \vec{x} \leq i \cdot ||\varphi||$, and so
$\vec{\alpha} \cdot \vec{x}$ has at most $m \defeq |\varphi|+\lfloor\log(i)\rfloor+1$ bits, and the same holds for $\vec{\beta} \vec{y}$.
So we have $ 1 \leq j \leq m$.  Let us first describe $\GreaterSum$, and then  $\textit{Probe}_1(j)$; the procedure $\textit{Probe}_2(j)$ is similar.

$\GreaterSum(\vec{x}, \vec{y})$ loops through $j=1, \ldots, m$. For each $j$, it calls $\textit{Probe}_1(j)$ and $\textit{Probe}_2(j)$.
If $\textit{Probe}_1(j)>\textit{Probe}_2(j)$, then it answers $\varphi(\vec{x}, \vec{y})=1$, otherwise it moves to $j+1$.
If $\GreaterSum$ reaches the end of the loop, then it answers $\varphi(\vec{x}, \vec{y})=0$. Observe that
 $\GreaterSum$ only needs to store the current value of $j$ and the bits returned by $\textit{Probe}_1(j)$ and $\textit{Probe}_2(j)$.
 Since $j \leq m$, $\GreaterSum$ only needs $\O(\log(|\varphi|+\log i))$ bits of memory.

$\textit{Probe}_1(j)$
uses a decreasing counter $k = m, \ldots, j$ to successively compute the bits $b_{1}(k)$ of $\vec{\alpha} \cdot \vec{x}$,
starting at the least significant bit. To compute $b_{1}(k)$, the procedure stores the carry $c_k \leq i$
of the computation of $b_1(k+1)$; it then computes the sum $s_k := c_k + \vec{\alpha}(k) \cdot \vec{x}$
(where $\vec{\alpha}(k)$ is the $k$-th vector of bits of $\vec{\alpha}$), and sets $b_{k} := s_k \bmod 2$ and
$c_{k-1} := s_k \div 2$.
The procedure needs $\O(\log (|\varphi|+\log i))$ bits of memory for counter $k$,
$\log(i)+1$ bits for encoding $s_k$, and
$\O(\log(i))$ bits for encoding $c_k$. So it only uses $\O(\log(|\varphi|+\log i))$ bits of memory.

Let us now simulate $\GreaterSum(\vec{x}, \vec{y})$ by a halting protocol with one leader agent.
Intuitively, the protocol proceeds in rounds corresponding to the counter $k$.
The leader stores in its state the value $j$ and the current values of the program counter, of counter $k$, and of variables $b_k$, $s_k$, and $c_k$.
The crucial part is the implementation of the instruction
$s_k := c_k + \vec{\alpha}(k) \cdot \vec{x}$ of $\textit{Probe}_1(j)$.
In each round, the leader adds input agents one by one. As the protocol only needs to work for populations with $i$ agents,
it is possible for each agent to know if it already interacted with the leader in this round, and for the leader to count the number of agents it has interacted with this round, until it reaches $i$ to start the next round.

\section{Conclusion and further work}\label{sec:conc}
We have proved that every predicate $\varphi$ of quantifier-free Presburger arithmetic (QFPA) is computed by a leaderless protocol with $\poly(|\varphi|)$ states.  
Further, the protocol can be computed in polynomial time. The number of states of previous constructions was exponential both in the bit-length of the coefficients of $\varphi$, and in the number of occurrences of boolean connectives. Since QFPA and PA have the same expressive power, 
 every computable predicate has a succinct leaderless protocol. This result completes the work initiated in \cite{BEJ18}, which also constructed succinct protocols, but only for some predicates, and with the help of leaders.
 
It is known that protocols with leaders can be exponentially faster than leaderless protocols. Indeed, every QFPA predicate is computed by a protocol with leaders whose expected
time to consensus is polylogarithmic in the number of agents \cite{AngluinAE08a},
while every leaderless protocol for the majority predicate needs at least linear time in the number of agents \cite{AlistarhAEGR17}. Our result shows that, if there is also an exponential gap in state-complexity, then it must be because some family of predicates have protocols with leaders of logarithmic size, while all leaderless families need polynomially many states. The existence of such a family is an open problem. 
 
The question of whether protocols with $\poly(|\varphi|)$ states exist for every PA formula $\varphi$, possibly with quantifiers, also remains open. However, it is easy to prove (see \Cref{appendix_PA}) that no algorithm for the construction of protocols from PA formulas runs in time $2^{p(n)}$ for any polynomial $p$:

\begin{restatable}{theorem}{thmHardnessPA}\label{thm:hardness:PA}
For every polynomial $p$, every algorithm that accepts a formula $\varphi$ of PA as input, and returns a population protocol computing $\varphi$, runs in time $2^{\omega(p(|\varphi|))}$.
\end{restatable}

\noindent Therefore, if PA also has succinct protocols, then they are very hard to find.

Our succinct protocols for QFPA have slow convergence (in the usual parallel time model, see e.g.~\cite{AlistarhG18}), since they often rely on exhaustive exploration of a number of alternatives, until the right one is eventually hit. The question of whether every QFPA predicate has a succinct \emph{and} fast protocol is very challenging, and we leave it open for future research.

\bibliography{references}

\appendix
\clearpage
\section{Equivalence of simple and standard population protocols} \label{appendix_output}
Recall that a simple population protocol (SPP), has two unique states 
$\q{f}, \q{t} \in Q$ with outputs $O(\q{f}) = 0$ and $O(\q{t}) = 1$ and all other states $q$ have output $O(q) = \bot$.

In the standard definition of population protocols used in the literature, all states $q$ have an output 
$O(q) \in \{0,1\}$. In this section we call such a protocol a \emph{full output population protocols} (FOPP). In a FOPP,  a configuration $C$ is a consensus configuration if
$O(p) = O(q)$ for every $p, q \in \supp{C}$. If $C$ is a consensus configuration, then its output $O(C)$ is the unique output of its states, otherwise it is $\bot$. 
An execution $\sigma = C_0 C_1 \cdots$ stabilizes to $b \in \{0, 1\}$ 
if $O(C_i) = O(C_{i+1}) = \cdots = b$ for some $i \in \mathbb{N}$. 
The output of $\sigma$ is $O(\sigma) = b$ if it stabilizes to $b$, and $O(\sigma)= \bot$ otherwise. 
A consensus configuration $C$ is
stable if every configuration $C'$ 
reachable from $C$ is a consensus configuration such that
$O(C') = O(C)$. It is easy to see that a fair execution of a FOPP stabilizes to $b \in \{0, 1\}$ if and only if it contains a stable configuration whose output is $b$.

A FOPP $\PP$ computes a predicate $\varphi \colon
\N^X \to \{0, 1\}$ if for every $\vec{v} \in \N^X$ every fair
execution $\sigma$ starting from $C_{\vec{v}}$ stabilizes to $\varphi(\vec{v})$.

In the rest of the section we show that every FOPP has an equivalent SPP, and vice versa.
Both translations have linear blow-up.

\parag{FOPP $\rightarrow$ SPP}
Let $\PP = (Q, T, L, X, I, O)$ be a FOPP computing a predicate $\varphi$.  We obtain a SPP protocol $\PP'$ by adding two output states $\{\q{f},\q{t}\}$ to $\PP$, plus a new state $\q{\bot}$. 
The output function of $\PP'$ is the mapping $O': q \mapsto (0 \text{ if } q = \q{f} \text{ else } 1
\text{ if } q = \q{t} \text{ else } \bot)$. The set $L'$ of leaders of $\PP'$ is obtained by 
adding one leader to $L$, initially in state $\q{\bot}$. Finally, the set $T'$ of transitions is obtained by adding to $T$, for all $b \in \{\q{f},\q{t},\q{\bot}\}$, a transition 
$\multiset{q, b} \mapsto \multiset{q, \q{f}}$ for every state $q \in Q$ such that $O(q)=0$ and 
a transition $\multiset{q, b} \mapsto \multiset{q, \q{t}}$ for every state $q \in Q$ such that $O(q)=1$. 

We show that $\PP'$ also computes $\varphi$.
Let $C'_0 C'_1 \cdots$ be a fair execution of $\PP'$ from $C_0'$. Projecting it onto the set of states of $\PP$ yields a fair execution of $\PP$. Since $\PP$ computes $\varphi$, the execution outputs some $b \in \{0,1\}$. Assume that $b=0$ (the case $b=1$ is symmetric). 
Let $i \in \N$ such that the output of every state populated by $C_j$ is $b$ for every $j \geq i$. 
Now, no matter the state populated by the additional leader in $C_j$ (which is one of $\{\q{f},\q{t},\q{\bot}\}$), the transition $\multiset{q, b} \mapsto \multiset{q, \q{f}}$
is enabled for every state $q$ such that $C_j(q)=0$. By fairness, the leader will thus eventually move to state $\q{f}$ and it will be stuck there, and $\PP'$ outputs $0$ as well.

\medskip
\parag{SPP $\rightarrow$ FOPP}
Let $\PP = (Q, T, L, X, I, O)$ be an SPP with output states $\q{f}, \q{t} \in Q$ computing a predicate $\varphi$. Let $\PP'$ be the FOPP with two disjoint copies $Q_0,Q_1$ of $Q$ as states. For $a \in Q$, let $a_b$ denote the copy of $a$ in $Q_b$, for $b \in \{0,1\}$.
We define $O'(q)=b$ for all $q \in Q_b$.
The set $T'$ of transitions is the following. First,
for every transition $\multiset{x, y} \mapsto \multiset{z, u}$ of $T$, the set $T'$
contains a transition $\multiset{x_b, y_c} \mapsto \multiset{z_d, u_e}$
for every $b,c,d,e \in \{0,1\}$ such that if $b=c$ then $d=e=b=c$. 
Further, $T'$ also contains a set $T''$ of transitions consisting of
$\multiset{\q{f}_1} \mapsto \multiset{\q{f}_0}$,
$\multiset{\q{t}_0} \mapsto \multiset{\q{t}_1}$,
$\multiset{a_1,\q{f}_0} \mapsto \multiset{a_0,\q{f}_0}$ 
and
$\multiset{a_0,\q{t}_1} \mapsto \multiset{a_1,\q{t}_1}$ 
for every $a \in Q$.

The input mapping and leader multiset of $\PP'$ are the ``0'' copies of the input mapping and leader multiset of $\PP$.
Hence, for any input $\vec{v}$ the initial configuration $C'_{\vec{v}}$ of $\PP'$ is the ``0'' copy of the initial configuration $C_{\vec{v}}$ in $\PP$. 

We show that $\PP'$ also computes $\varphi$. Let $C'_0 C'_1 \cdots$ be a fair execution of $\PP'$ from $C'_0$. For every $i \in \mathbb{N}$, let $C_i=\pi(C'_i)$, where $\pi$ is the mapping defined by $\pi(q_b)=b$ for every $q_b \in Q_0 \cup Q_1$.
It is easy to see that $C_0 C_1 \cdots$
is a fair execution of $\PP$, with possible repetitions $C_i=C_{i+1}$
when the transition from $C'_i$ to $C'_{i+1}$ is in $T''$.
Hence $C_0 C_1 \cdots$ eventually stabilizes to an output $b$.
Assume that $b=0$ (the case $b=1$ is symmetric).
By fairness, because of the transitions 
$\multiset{\q{f}_1} \mapsto \multiset{\q{f}_0}$,
and
$\multiset{a_1,\q{f}_0} \mapsto \multiset{a_0,\q{f}_0}$, the execution $C'_0 C'_1 \cdots$ eventually reaches, and gets trapped in, configurations of $\mathbb{N}^{Q_0}$. So the execution also stabilizes to the output $0$.


\section{Proofs of Section \ref{sec:large}: Protocols for large populations} \label{app:large}

\subsection{Proof of Theorem \ref{thm:remove:helpers}}\label{app:rem:helpers}
\thmRemoveHelpers*

We first define a leaderless protocol $\overline{\PP}$, introduce some auxiliary definitions and propositions, and finally prove that $\overline{\PP}$ computes $(|\vec{v}| \geq \cutoffvar)
 \rightarrow \varphi(\vec{v})$.
 
\parag{The protocol $\overline{\PP}$} As mentioned in the main text,
by~\cite[Lemma~3]{BEJ18}, $\PP$ can be transformed into a 2-way
population protocol (with helpers) also computing $\varphi$, and with
at most $|Q| + 3\lambda \cdot |T|$ states, where $\lambda \defeq
\max\{|\vec{p}| : (\vec{p}, \vec{q}) \in T\}$. Thus, we assume that
$\PP$ is 2-way in the rest of this section, implicitly keeping in mind
the polynomial increase in the number of states.

Let $h_1, h_2, \ldots, h_\cutoffvar \in Q$ be the helpers of $\PP$
in some arbitrary but fixed order. For example, if $L
= \{p, 3 \cdot q\}$, then we can have $h_1 = p$, $h_2 = q$, $h_3 = q$
and $h_4 = q$. Let $\overline{\PP} \defeq (\overline{Q}, \overline{T}, {\vec{0}}, X, \overline{I}, \overline{O})$ be the
population protocol such that:
\begin{align*}
 \overline{Q} &\defeq (X \times [\cutoffvar]) \cup Q^\ms{2}, \\
  \overline{T} &\defeq \overline{T}_\text{count} \cup \overline{T}_\text{init} \cup \overline{T}_\text{simul}, \\
  \overline{I} &\defeq x \mapsto (x, 1), \\
  \overline{O} &\defeq
  \begin{cases}
    (x, i) \mapsto 1
    & \text{for every } (x, i) \in X \times [\cutoffvar], \\
    \vec{q} \mapsto O(\vec{q})
    & \text{for every } \vec{q} \in Q^\ms{2},
  \end{cases}
\end{align*}
where 
\begin{itemize}
\item $\overline{T}_\text{count}$ is the set consisting of the following transitions:
\begin{align*}
  \multiset{(x, i), (y, i)} &\mapsto \multiset{(x, i+1), (y, i)}
  && \text{for every } x, y \in X \text{ and } i < \cutoffvar,
\end{align*}
\item  $\overline{T}_\text{init}$ is the set consisting of the following transitions:
\begin{align*}
  \multiset{(x, \cutoffvar), (y, i)}
  &\mapsto \multiset{(I(x), h_\cutoffvar), (I(y), h_i)}  
  && \text{for every } x, y \in X \text{ and } i \leq \cutoffvar, \\
  \multiset{\vec{q}, (y, i)}
  &\mapsto \multiset{\vec{q}, (I(y), h_i)}
  && \text{for every } y \in X, i \leq \cutoffvar, \text{ and } \vec{q} \in
  Q^\ms{2},
\end{align*}
\item $T_\text{simul}$ is the consisting of the following
transitions:
\begin{align*}
  \multiset{\vec{p}, \vec{q}} \mapsto \multiset{\vec{p}', \vec{q}'}
  && \text{ for every } \vec{p}, \vec{q}, \vec{p}', \vec{q}'
  \in\ Q^\ms{2} \text{ such that } (\vec{p} + \vec{q})
  \trans{} (\vec{p}' + \vec{q}') \text{ in } \PP.
\end{align*}
\end{itemize}

\parag{Auxiliary definitions and propositions} The intended behavior
of $\overline{\PP}$ is to first fire $\overline{T}_\text{count}$, then $\overline{T}_\text{init}$, and
then $\overline{T}_\text{simul}$. Although $\overline{\PP}$ may fire sequences not
respecting this order, there always exist an equivalent sequence
respecting the order, in the following sense:

\begin{proposition} \label{prop:helpers:reorder}
  For every configurations $C$ and $D$ such that $C \trans{*} D$,
  there exist $x \in \overline{T}_\text{count}^*$, $y \in \overline{T}_\text{init}^*$ and $z
  \in \overline{T}_\text{simul}^*$ such that $C \trans{xyz} D$.
\end{proposition}

\begin{proof}
  Let $w \in \overline{T}^*$ be such that $C \trans{w} D$. The sequence $xyz$ is
  simply obtained by reordering the transitions of $w$. Firability of
  $xyz$ follows from inspection of $\overline{T}$.
\end{proof}

Observe that firing $\overline{T}_\text{count}$, until no further possible,
counts the number of agents up to $\cutoffvar$:

\begin{proposition} \label{prop:helper:count}
  Let $C$ and $D$ be configurations such that $C$ is initial, $C
  \trans{\overline{T}_\text{count}^*} D$ and $\overline{T}_\text{count}$ is disabled in
  $D$. We have $\supp{D} \cap (X \times \{j\}) \neq \emptyset \iff |C|
  \geq j$ for every $j \in [\cutoffvar]$.
\end{proposition}

\begin{proof}
  \newcommand{\pos}[1]{\mathrm{pos}(#1)}

  Let $P_j \defeq X \times \{j\}$ for every $j \in [\cutoffvar]$. For every
  configuration $E$, let $\pos{E} \defeq \{j \in [\cutoffvar] : C(P_j) >
  0\}$. We define a relation $\prec$ on configurations:
  $$E \prec E' \defiff \pos{E'} = \pos{E} \lor \pos{E'} = \pos{E} \cup
  \{\max(\pos{E}) + 1\}.$$ Observe that $E \trans{T_\text{count}} E'$
  implies $E \prec E'$. Consequently, since $\pos{C} = \{1\}$, we have
  $\pos{D} = \{1, 2, \ldots, m\}$ for some $m \in [\cutoffvar]$. To
  complete the proof, it suffices to show that $m = \min(|C|, \cutoffvar)$.

  Clearly, $m \leq \min(|D|, \cutoffvar) = \min(|C|, \cutoffvar)$ holds. Let us show
  that $m \geq \min(|C|, \cutoffvar)$. If $m = \cutoffvar$, then we are
  done. Therefore, assume $m < \cutoffvar$. Since $\overline{T}_\text{count}$ is
  disabled in $D$, we have $D(P_i) = 1$ for every $1 \leq i \leq m$
  and $D(P_i) = 0$ for every $m < i \leq \cutoffvar$. Thus, $m = |D| =
  |C| \geq \min(|C|, \cutoffvar)$.
\end{proof}

For every configuration $C$ of $\overline{\PP}$, let $\widehat{C}$ be the
configuration of $\PP$ obtained by ``projecting'' $C$ onto
$Q$, \ie\ the configuration such that
\begin{align*}
  \widehat{C}(p)
  &\defeq \sum_{\vec{q} \in Q^\ms{2}} \vec{q}(p) \cdot C(\vec{q})
  && \text{ for every } p \in Q.
\end{align*}
We extend this notation to executions, i.e., to sequences of configurations. The
following correspondence follows immediately from the definitions:

\begin{proposition} \label{prop:helpers:exec}
  For every (fair) execution $\sigma$ of $\overline{\PP}$, $\widehat{\sigma}$ is
  a (fair) execution of $\PP$.
\end{proposition}

\parag{Main proof} We  prove that $\overline{\PP}$ computes $\varphi$.

\begin{proof}[Proof of \cref{thm:remove:helpers}]
  Let $\vec{v} \in \N^X$ and let $\sigma$ be a fair execution of $\overline{\PP}$
  from $C_{\vec{v}}$. Observe that, by definition of $\overline{T}$, the number
  of transitions from $\overline{T} \setminus \overline{T}_\text{simul}$ occurring along
  $\sigma$ must be finite. Let $i \in \N$ be some index such that $\overline{T}
  \setminus \overline{T}_\text{simul}$ is disabled in $\sigma_j$ for every $j
  \geq i$. By \cref{prop:helpers:reorder}, there exist $x \in
  \overline{T}_\text{count}^*$, $y \in \overline{T}_\text{init}^*$, $z \in
  \overline{T}_\text{simul}^*$, and configurations $C$ and $D$ such that
  $\sigma_0 \trans{x} C \trans{y} D \trans{z} \sigma_i$. By
  \cref{prop:helper:count}, the following holds for every $j \in [\cutoffvar]$:
  \begin{align}
    \supp{C} \cap (X \times \{j\}) \neq \emptyset \iff |\vec{v}| \geq
    j. \label{eq:all:helpers}
  \end{align}
  Let us now show that $\overline{O}(\sigma)$ is as expected, by making a case
  distinction on whether $|\vec{v}| \geq \cutoffvar$.\medskip

  \noindent\emph{Case $|\vec{v}| < \cutoffvar$.} By~\eqref{eq:all:helpers}, we
  have $C(x, \cutoffvar) = 0$ for every $x \in X$. Thus, we have $y = z =
  \varepsilon$ since no transition of $\overline{T}_\text{init} \cup
  \overline{T}_\text{simul}$ is enabled in $C$. This implies that $C = \sigma_i =
  \sigma_{i+1} = \cdots \in X \times [\cutoffvar]$. Hence, $\overline{O}(\sigma) = \overline{O}(C) =
  1$ which is the expected output. \medskip

  \noindent\emph{Case $|\vec{v}| \geq \cutoffvar$.} By~\eqref{eq:all:helpers},
  $C(x, \cutoffvar) > 0$ for some $x \in X$. Thus, fairness enforces
  sequence $y$ to convert every agent from states $X \times [\cutoffvar]$ to
  states $\overline{Q}^\ms{2}$. Thus, we have
  $D \geq \multiset{h_1, h_2, \ldots, h_\cutoffvar}$
  by~\eqref{eq:all:helpers}, which implies
  $D \in (L' \mminus L) \mplus C_{\vec{v}}$ for some
  $L' \succeq L$, and consequently
  $\sigma_i, \sigma_{i+1}, \ldots \in \N^{\overline{Q}^\ms{2}}$.

  Let $m \defeq |z|$ and let $D_0, D_1, \ldots, D_m$ be the
  configurations such that $D = D_0 \trans{} D_1 \trans{} \cdots
  \trans{} D_m = \sigma_i$. Let $\pi \defeq D_0 D_1 \cdots D_m
  \cdot \sigma$. By fairness of $\sigma$ and by
  \cref{prop:helpers:exec}, $\widehat{\pi}$ is a fair execution of
  $\PP$, which implies that $O(\widehat{\pi}) =
  \varphi(\vec{v})$. Therefore, we have $\overline{O}(\pi) = \varphi(\vec{v})$ by
  definition of $\overline{O}$. Since $\sigma$ and $\pi$ share a common
  (infinite) suffix, we have $\overline{O}(\sigma) = \overline{O}(\pi)$, which completes the
  proof.
\end{proof}

\subsection{Proof of Proposition \ref{prop:mult}}\label{app:mult}
\PropMult*

\begin{proof}
Let $\atomic(P)$ be the set of atomic predicates in $P$.
Consider a forest of binary trees of boolean operations encoding $\varphi$ (negations have only one child), with atomic predicates at the leaves. There are at most $\len(\varphi)+\size(\varphi)$ nodes in that forest (roots correspond to different predicates of $\varphi$).
Consider the set $P'$ made of every predicate corresponding to nodes of the forest.
We call such $P'$ a {\em full} set of predicates.
We have $\size(P') \leq \len(\varphi)+ \size(\varphi) \leq |\varphi|$,
$\norm{P'}=\norm{\varphi}$ and $\len(P') \leq \len(\varphi)^2 \leq |\varphi|^2$.
We prove by induction on $\len(P')$ that every {\em full} $P'$ is computed by some multi-output population
protocol with 
$\O( \len(P') + |\atomic(P')|^5)$ helpers, states and transitions.

If $\len(P') = 0$, then
each predicate is atomic, and the claim is true by hypothesis.

Let $P'$ be a full set with $\len(P')= k > 0$, and assume that the claim holds for every full set $P''$ with $\len(P'') < k$.
Let $\varphi \in P'$ with $\len(\varphi)$ maximal.
Let us consider the case where $\varphi = \psi \land \psi'$ for some predicates $\psi, \psi'$. The case of disjunction and negation are handled similarly. Let $P'' \defeq P' \setminus \{\varphi\}$. Note that $\len(P'') < \len(P')$,
and that $P''$ is full because $\len(\varphi)$ is maximal.
Thus, by induction
hypothesis, we obtain a simple multi-output population protocol
$\PP'' = (Q, T, L, X, I, O)$ that computes $P''$. 
Assume w.l.o.g.\ that the indices of $O$ associated to $\psi$ and $\psi'$ are $|P|$ and $|P| + 1$ respectively.
Let $q_0, q_1, r_0, r_1 \in Q$ be the unique states such that $O_{|P|}(q_b) = b$ and
$O_{|P|+1}(r_b) = b$ for $b \in \{0, 1\}$. These states exist
since $\PP''$ is simple. Let $\PP' = (Q', T', L', X', I', O')$ be the
multi-output protocol such that:
  \begin{align*}
    Q' &\defeq Q \cup \{o_0, o_1\}, \\
    T' &\defeq T \cup \{(\multiset{q_a, r_b, o_{\neg c}}, \multiset{q_a, r_b,
      o_c}) : a, b, c \in \{0, 1\}, a \land b = c\}, \\
    L' &\defeq L \mplus \multiset{o_0}, \\
    I' &\defeq I, \\
    O'_i &\defeq
    \begin{cases}
      O_i
      & \text{for every } 1 \leq i < |P|, \\
      q \mapsto (b \text{ if } q = o_b \text{ else } \bot)
      & \text{for $i = |P|$}.
    \end{cases}
  \end{align*}
  We claim that $\PP'$ computes $P'$. 
  Note that $\PP'$ behaves exactly as
  $\PP$ on $Q$. This implies that $\PP'$ computes each predicate of
  $P' \setminus \{\varphi\}$. Thus, it suffices to show that it also
  computes $\varphi$. Let $\sigma$ be a fair execution of $\PP'$
  starting from some initial configuration $C_{\vec{v}}$. Since $\PP$
  is simple and computes both $\psi$ and $\psi'$, there exists $i \in
  \N$ such that for every $j \geq i$:
  $$ \sigma_j(q_{\psi(\vec{v})}) > 0,\ \sigma_j(r_{\psi'(\vec{v})}) >
  0\ \text{ and }\ \sigma_j(q_{\neg \psi(\vec{v})}) = \sigma_j(r_{\neg
    \psi'(\vec{v})}) = 0.
  $$ Thus, by fairness, there exists $i'\geq i$ such that
  $\sigma_j(o_{\psi(\vec{v}) \land \psi'(\vec{v})}) > 0$ and
  $\sigma_j(o_{\neg(\psi(\vec{v}) \land \psi'(\vec{v}))}) = 0$ for
  every $j \geq i'$. This implies that $O_{|P|}(\sigma) =
  \psi(\vec{v}) \land \psi'(\vec{v}) = \varphi(\vec{v})$.

Concerning the number of states and helpers, the protocol $\PP'$ uses two states plus the states of $\PP''$, and one helper plus the helpers of $\PP''$, which ends the proof by induction as $\atomic(P'')=\atomic(P')=\atomic(P)$. 

In terms of $|\varphi|$, we obtain a protocol with 
$\O( \len(\varphi) + |\atomic(\varphi)|^5)=
\O(|\varphi|^5)$ helpers, states and transitions.
 \end{proof}

\subsection{Proof of Theorem \ref{thm:atomic:mult:helpers}}\label{app:atomic:mult}
\thmAtomicMultHelpers*

Let $P=\{\varphi_1, \varphi_2, \ldots, \varphi_k\}$.  For every $i \in
[k]$, let $\PP_i = (Q_i, T_{\infty i}, T_{\dagger i}, L_i,
\overline{X} \cup \underline{X}, I_i, O_i)$ be the simple RDI-protocol
with helpers computing $\tilde{\varphi}_i$. Recall that each $\PP_i$
has two input variables $\overline{x}$ and $\underline{x}$ for each
input variable $x \in X$. Recall further that the transitions of
$T_{\dagger i}$ are called RDI-transitions.

We first define a simple multi-output protocol $\PP$. Then we
introduce some auxiliary definitions and propositions, and finally we
prove that $\PP$ computes $P$.

\parag{Notations} For every RDI-transition $t = (\vec{p}, \vec{q})$
and for every $x \in X$, let $t^{x}$ be the transition defined as
$t^{x} \defeq (\vec{p} + \multiset{x}, \vec{q} + \multiset{x})$. In
other words, $t^{x}$ has the same effect as $t$, but is ``guarded'' by
$X$, i.e., it can only occur if some agent is in state $x$. We say
that $t^{x}$ is a \emph{guarded transition}. Given a set $U$ of
transitions, we define the sets $U^{\textrm{g}}$ of guarded transitions
and $U^{-1}$ of guarded reversal transitions as:
$$U^{\textrm{g}} \defeq \{t^{x} : t \in U, x
\in X\} \text{ and } U^{-1} \defeq \{(\vec{q}, \vec{p}) : (\vec{p},
\vec{q}) \in U^g\}.$$

\parag{The protocol} The $k$-output population protocol
with helpers $\PP = (Q, T, L, X, I, O)$ is defined as follows:
\begin{itemize}
\setlength\itemsep{4pt}

\item  $Q \defeq X \cup H \cup Q_1 \cup Q_2 \cup \cdots \cup Q_k$, where
$H \defeq \{h_x : x \in X\}$. \\
Intuitively, $X$ are the input states, $H$ are auxiliary states used to distribute agents
to the atomic protocols, and $Q_1, \ldots, Q_k$ are the states of the atomic protocols themselves.\smallskip

\item $T \defeq S \cup S^{-1} \cup (T_{\infty 1} \cup T_{\infty 2} \cup \cdots \cup
  T_{\infty k}) \cup (T_{\dagger 1} \cup T_{\dagger 2} \cup \cdots \cup
  T_{\dagger k})^{\textrm{g}}$, where \\[5pt] $S \defeq \{\overline{s}_x, \underline{s}_x : x \in X\}$ and
\begin{alignat*}{2}
  \overline{s}_x & \defeq\ \ \ &
  \multiset{k \cdot x} &\mapsto \multiset{I_i(\overline{x}) : i \in [k]}, \\
  \underline{s}_x &\defeq\ \ \ & 
  \multiset{x, (k-1) \cdot h} &\mapsto \multiset{I_i(\underline{x}) :
    i \in [k]}.
\end{alignat*}
Intuitively, the transitions of $S$ allow $\PP$ to distribute agents to $\PP_1, \ldots, \PP_k$.
Transition $\overline{s}_x$ collects $k$ agents from the input state of $\PP$ for $x$, and sends one agent to each of the input states of $\PP_1, \ldots, \PP_k$ for $\overline{x}$. Similarly,
$\underline{s}_x$ collects one agent from $x$ and $(k-1)$ helpers, and sends one agent to each of the input places of $\PP_1, \ldots, \PP_k$ for $\underline{x}$.

Transitions of $S^{-1}$ allow $\PP$  to collect agents back if they were not distributed properly. They are guarded to ensure that the agents are not collected when the distribution is correct. 

The rest of the transitions are the transitions of $\PP_1, \ldots, \PP_k$, with an additional guard on the transitions of $T_{\dagger 1}, \ldots, T_{\dagger k}$. The guards ensure that $\PP_1, \ldots, \PP_k$ stop returning agents to the input states once the correct distribution is achieved.

\item $L \defeq \multiset{(k-1)^2 \cdot h_x : x \in X} \mplus L_1 \mplus
  L_2 \mplus \cdots \mplus L_k$. \\
The helpers of $\PP$ are those of $\PP_1, \ldots, \PP_k$, plus $(k-1)^2$ helpers for each input variable.
\item $I \defeq x \mapsto x$. 

\item The output mapping for $\varphi_i$ is given by $O (i, q) \defeq (O_i(q) \text{ if } q \in Q_i \text{ else } \bot)$. 
\end{itemize}

\parag{Auxiliary definitions and propositions} 
\begin{itemize}
\setlength\itemsep{4pt}

\item For every $i \in [k]$
and every configuration $C \in \N^Q$ of $\PP$, let $C^i \in \N^{Q_i}$ be the
configuration of $\PP_i$ such that $C^i(q) \defeq C(q)$ for every $q
\in Q_i$. 
\item For every $i \in [k]$ and every sequence $w \in T^*$, let $w^i$ be the
projection of $w$ onto the transitions of $S \cup S^{-1} \cup T_{\infty i}
\cup (T_{\dagger i})^{\textrm{g}}$. 
  
\item For every $w \in T^*$, let $\vec{w} \in \Z^{\overline{X} \cup \underline{X}}$
be the vector such that for every $x \in X$:
\begin{align*}
  \vec{w}(\overline{x})
  &\defeq |w|_{\overline{s}_x} - \sum_{ t \in \{ \overline{s}_x \}^g  } |w|_{t}, &
  \vec{w}(\underline{x})
  &\defeq |w|_{\underline{s}_x} - \sum_{ t \in \{ \underline{s}_x \}^g  } |w|_{t}.
\end{align*}
In other words, $\vec{w}$ records the difference between the number of
occurrences of transition $\overline{s}_x$ and its guarded reversals, for each
variable $x$, and similarly for $\underline{s}_x$.

\item Observe that the set of input variables of $\PP$ is $X$, while the set of input variables
of $\PP_i$ is $\overline{X} \cup \underline{X}$. Given  $\vec{v} \in \N^X$ and $\vec{w} \in \N^{(\overline{X} \cup \underline{X})}$, we let $\vec{v} \equiv \vec{w}$ denote that $\vec{v}(x) = k \cdot \vec{w}(\overline{x}) + \vec{w}(\underline{x})$ for every $x \in X$.
\end{itemize}

Let us prove the following observations on the executions of $\PP$:

\begin{proposition} \label{prop:atomic:mult}
  Let $\vec{v} \in \N^X$, $\hat{L} \succeq L$ and $\hat{C}_{\vec{v}}
  \defeq \hat{L} \mplus \{\vec{v}(x) \cdot x : x \in X\}$. 
  Let
  $\hat{C}_{\vec{v}} \trans{w} C$ be a finite execution of $\PP$. We
  have:
  \begin{enumerate}
  \item There exists an initialization sequence from
    $\hat{C}_{\vec{v}}^i$ to $C^i$ in $\PP_i$ with effective input
    $\vec{w} \geq \vec{0}$. \label{itm:mult:rev}
  
  \item If $C(X) = 0$, then $\vec{v} \equiv \vec{w}$. \label{itm:mult:equiv} 

  \item There exists a configuration $D$ such that $C \trans{*} D$ and
    $D(X) = 0$. \label{itm:mult:empty}
  \end{enumerate}
\end{proposition}

\begin{proof}
  \leavevmode
  \begin{enumerate}
  \item Let $i \in [k]$. The only transitions that change the number
    of agents over the states of $Q_i$ are those of $S \cup
    S^{-1} \cup T_{\infty i} \cup (T_{\dagger
      i})^{\textrm{g}}$. Transitions $S \cup S^{-1}$
    have the same effect as the transitions of $\Ipt$ and $\Opt$.
    Transitions $T_{\infty i} \cup T_{\dagger i}$ form precisely the
    set of transitions of $\PP_i$, and the effects of the transitions
    of $T_{\dagger i}$ and $(T_{\dagger i})^{\textrm{g}}$
    coincide. Moreover, we have $\hat{C}_{\vec{v}}^i = \hat{L}^i
    \succeq L^i = L_i$. Therefore, $w_i$ yields an initialization
    sequence of $\PP^i$ from $\hat{C}_{\vec{v}}^i$ to $C^i$ with
    effective input $\vec{w}$. Since $\PP^i$ is an RDI-protocol,
    $C^i(I_i(x)) \leq \vec{w}(x)$ holds for every $x \in \overline{X}
    \cup \underline{X}$. Hence, we must have $\vec{w} \geq \vec{0}$ as
    a configuration cannot hold any negative amount of agents.\medskip

  \item An induction on $|w|$ shows that $C(x) = \vec{v}(x) - k \cdot
    \vec{w}(\overline{x}) - \vec{w}(\underline{x})$ for every $x \in
    X$. Thus, if $C(X) = 0$, then $\vec{v}(x) = k \cdot
    \vec{w}(\overline{x}) + \vec{w}(\underline{x})$ for every $x \in
    X$. Hence, since $\vec{w} \geq \vec{0}$ by~\eqref{itm:mult:rev},
    we have $\vec{v} \equiv \vec{w}$.\medskip

  \item For every configuration $A$ of $\PP$, let $\DiEq{A} \defeq \{B
    : B^i \in \DiEq{A^i}\} \text{ for every } i \in [k]\}$. Note that:
    \begin{align}
      A(X) = B(X) \text{ for every } B \in \DiEq{A}.\label{eq:x:equal}
    \end{align}
    Let $C_j \trans{t_j} C_{j+1}$ be $j^\text{th}$ step of
    $\hat{C}_{\vec{v}} \trans{w} C$. For every $D_{j+1} \in
    \DiEq{C_{j+1}}$, we construct a sequence $w_j \in T^*$ such that
    $D_{j+1} \trans{w_j} D_j$ for some $D_j \in \DiEq{C_j}$. In other
    words, we show how to reverse $t_j$, up to a possible
    redistribution of the output agents. The validity of the main
    claim follows by \eqref{eq:x:equal} and a straightforward
    induction. We may assume without loss of generality that
    $C_{j+1}(X) > 0$, as otherwise the main claim would already be
    satisfied. Since $C_{j+1}(X) > 0$, guarded transitions of $\PP$
    are equivalent to their unguarded counterparts, \ie\ a transition
    $u$ is enabled at $C_{j+1}$ if and only if $u^\textrm{g}$ is
    enabled at $C_{j+1}$. Thus, we may reverse $t_j$ as
    follows:\smallskip
    \begin{itemize}
      \item If $t_j \in S$, then we pick $w_j \in S^{-1}$
        as the guarded reversal of $t_j$;

      \item If $t_j \in S^{-1}$, then we pick $w_j \in S$
        as the counterpart transition of $t_j$;

      \item If $t_j \in T_{\infty i} \cup T_{\dagger i}$ for some $i
        \in [k]$, then we proceed as follows. By~\eqref{itm:mult:rev},
        there is an initialization sequence from $\hat{C}_{\vec{v}}^i$
        to $C^i$ in $\PP_i$ with effective input $\vec{w}$. Moreover,
        $$C^i \trans{*} C_j^i \trans{t_j} C_{j+1}^i \text{ in }
        \PP_i.$$ Hence, since $\PP_i$ is an RDI-protocol, there exists
        $w_j \in (T_{\infty i} \cup T_{\dagger i})^*$ such that
        $D_{j+1}^i \trans{w_j} E$ in $\PP_i$ for some $E \in
        \DiEq{C_j^i}$. Thus, we have $D_{j+1} \trans{w_j} D_j$ in
        $\PP$ for some $D_j \in \DiEq{C_j}$.\qedhere
    \end{itemize}

  \end{enumerate}
\end{proof}

\parag{Main proof} We proceed to prove that $\PP$ indeed computes
$\{\varphi_1, \varphi_2, \ldots, \varphi_k\}$.

\begin{proof}[Proof of~\cref{thm:atomic:mult:helpers}]
  Let $\vec{v} \in \N^X$, $\hat{L} \succeq L$, and let $\sigma$ be a
  fair execution of $\PP$ starting from $\hat{C}_{\vec{v}} \defeq
  \hat{L} \mplus \{\vec{v}(x) \cdot I(x) : x \in X\}$. By
  \cref{prop:atomic:mult}~\eqref{itm:mult:empty} and by fairness,
  there exists $j \in \N$ such that $\sigma_j(X) = 0$. By definition
  of $T$, if $X$ is emptied, then it remains permanently emptied, as
  none of the guarded reversals can be fired. Thus, we have:
  \begin{align}
    \sigma_j(X) = \sigma_{j+1}(X) = \cdots = 0. \label{eq:mult:empty}
  \end{align}
  Let $\widehat{\sigma} \defeq \sigma_j^i \sigma_{j+1}^i
  \cdots$. Consider protocol $\PP_i$ for some $i$, and let $\vec{w}$ be the 
  effective input of the initialization sequence
  $\widehat{\sigma}^i$ of $\PP_i$.  By~\eqref{eq:mult:empty}, 
  $\widehat{\sigma}^i$ only contains transitions
  of $T_{\infty i}$, and is consequently a fair execution of protocol
  $\PP_{\infty i}$. By hypothesis, and by definition of RDI-protocols, $\PP_{\infty i}$ computes 
  $\tilde{\varphi}_i$. Hence, we have $O_i(\widehat{\sigma}^i) =
  \tilde{\varphi}_i(\vec{w})$. We are done since, by
  \cref{prop:atomic:mult}~\eqref{itm:mult:equiv}, we have $\vec{v}
  \equiv \vec{w}$, which implies $\varphi_i(\vec{v}) =
  \tilde{\varphi}_i(\vec{w})$.
\end{proof}

\subsection{Proof of Theorem \ref{thm:atomic:helpers}}\label{app:threshold}
\thmAtomicHelpers*

In Section \ref{subsec:app-threshold} we describe the protocol
for threshold predicates, and prove its correctness. 
Section \ref{subsec:app-remainder} does the same for remainder predicates. 

\subsubsection{Threshold protocols}\label{subsec:app-threshold}

Let us fix a threshold predicate $\varphi$ over variables $X$. Without
loss of generality\footnote{If $b \leq 0$, then we can instead
  consider the equivalent predicate $\neg(-\vec{a} \cdot \vec{v} \geq
  -b + 1)$, construct a protocol for $-\vec{a} \cdot \vec{v} \geq -b +
  1$ and handle the negation separately in \cref{subsec:finsets}.}, we
have $\varphi(\vec{v}) = \vec{a} \cdot \vec{v} \geq b$ where $\vec{a}
\in \Z^X$ and $b > 0$. We construct a simple
population protocol $\PP_{\text{thr}}$ that computes $\varphi$ under
reversible dynamic initialization, and prove its correctness.

\parag{Notations} Let $n$ be the smallest number such that $2^n >
\norm{\varphi}$. Let $P \defeq \{\q{+2^i}, \q{-2^i} : 0 \leq i \leq
n\}$, $Z \defeq \{\q{0}\}$, $N \defeq P \cup Z$ and $B \defeq \{\q{f},
\q{t}\}$, where $P$, $Z$, $N$ and $B$ respectively stand for
``$P$owers of two'', ``$Z$ero'', ``$N$umerical values'' and
``$B$oolean values''. For every set $S$ and every $x \in X$, let $S_x
\defeq \{q_x : q \in S\}$ and $S_X \defeq S \cup \bigcup_{x \in X}
S_x$.

For every $d \in \N$, let $\mathrm{bits}(d)$ denote the unique set $J
\subseteq \N$ such that $d = \sum_{j \in J} 2^j$,
\eg\ $\mathrm{bits}(13) = \mathrm{bits}(1101_2) = \{3, 2, 0\}$. The \emph{canonical representation} of an integer $d \in \Z$ is the multiset
$\mathrm{rep}(d)$ defined as follows:

$$
\mathrm{rep}(d) \defeq
\begin{cases}
  \multiset{\q{+2^i} : i \in \mathrm{bits}(d)}  & \text{if } d > 0, \\
  \multiset{\q{-2^i} : i \in \mathrm{bits}(|d|)} & \text{if } d < 0, \\
  \multiset{\q{0}}                               & \text{if } d = 0.
\end{cases}
$$

\parag{The protocol} The RDI-protocol $\PP_{\text{thr}} = (Q, T_\infty, T_\dagger, L, X, I, O)$ is defined as
follows:
\begin{itemize}
\setlength\itemsep{4pt}

\item  $Q \defeq X \cup N_X \cup B$. \\
Intuitively, the states of $X$ are the ``ports'' through which the agents for each variable
enter and exit the protocol.

\item  $I \defeq x \mapsto \q{x}$. \\
That is, the initial state for variable $x$ is $x$.

\item $L \defeq \multiset{2n \cdot \q{0}, \q{f}}$. \\
So, we have $2n$ helpers in state $\q{0}$, and one helper in state $\q{f}$, i.e., initially 
the protocol assumes that the predicate does not hold.

\item $O(q) \defeq q \mapsto (0 \text{ if } q = \q{f} \text{ else } 1
\text{ if } q = \q{t} \text{ else } \bot)$. \\
That is, the output of the protocol is completely determined by the number of agents in states
$\q{t}$ and $\q{f}$

\item $T_\infty$ is the following set of (``permanent'') transitions:
\begin{alignat*}{3}
	\mathsf{add}_{x}:\ &&
	\multiset{x, |\mathrm{rep}(\vec{a}(x))| \cdot \q{0}}
	& \mapsto \multiset{\q{0}_x} \mplus \mathrm{rep}(\vec{a}(x))
	&& \quad \text{for all $x \in X$,} \\
	\mathsf{up}_i^\circ:\ &&
	\multiset{\q{\circ2^i}, \q{\circ2^i}}
	&  \mapsto \multiset{\q{\circ2^{i+1}}, \q{0}}
	&& \quad\text{for all $0 \leq i < n$ and $\circ \in \{+, -\}$,} \\
	\mathsf{down}_i^\circ:\ &&
	\multiset{\q{\circ2^i}, \q{0}}
	&  \mapsto \multiset{\q{\circ2^{i-1}}, \q{\circ{2^{i-1}}}}
	&& \quad\text{for all $0 < i \leq n$ and $\circ \in \{+, -\}$,} \\
	\mathsf{cancel}_{i, q}:\ &&
	\multiset{\q{+2^i}, \q{-2^i}, q}
	&  \mapsto \multiset{\q{0}, \q{0}, \q{f}}
	&& \quad\text{for all $0 \leq i \leq n$ and $q \in B$,} \\
	\mathsf{swap}_{p, q}^x:\ &&
	\multiset{p, q_x}
	& \mapsto \multiset{p_x, q}
	&& \quad \text{for all $p, q \in N$ and $x \in X$,} \\
	\mathsf{equal}:\ &&
	\mathrm{rep}(b) \mplus \multiset{\q{f}}
	& \mapsto \mathrm{rep}(b) \mplus \multiset{\q{t}}, \\
	\mathsf{false}:\ &&
	\multiset{\q{f}, \q{t}} & \mapsto \multiset{\q{f}, \q{f}}.
\end{alignat*}
Intuitively, $\mathsf{add}_{x}$ converts an agent which arrived via port $x$ into 
the canonical representation of $\vec{a}(x)$. Transitions of the form $\mathsf{up}_i^\circ, 
\mathsf{down}_i^\circ$ and $\mathsf{cancel}_{i, q}$ allow the protocol to change
the representation of a value, without changing the value itself.
Transition $\mathsf{equal}$ allows the protocol to detect that the current value
of $\vec{a}\cdot \vec{x}$, for the current input $\vec{x}$, is at least $b$, which moves 
a helper from state $\q{f}$ to $\q{t}$.
\item Finally, $T_{\dagger}$ is the following set of RDI-transitions:
\begin{alignat*}{3}
	\mathsf{add}_{x,q}^{-1}:\ &&
	\multiset{\q{0}_x, q} \mplus \mathrm{rep}(\vec{a}(x))
	& \mapsto \multiset{x, \q{f}, |\mathrm{rep}(\vec{a}(x))| \cdot \q{0}}
	&& \quad \text{for all $x \in X$ and $q \in B$,} \\
	\mathsf{cancel}_{i,q}^{-1}:\ &&
	\multiset{\q{0}, \q{0}, q}
	&  \mapsto \multiset{\q{+2^i}, \q{-2^i}, \q{f}}
	&& \quad\text{for all $0 \leq i \leq n$ and $q \in B$,}  
	\\
	\mathsf{reset}:\ &&
	\multiset{\q{t}}
	& \mapsto \multiset{\q{f}}.
\end{alignat*}
The first two transitions are needed to reverse the changes of
$\mathsf{add}$ and $\mathsf{cancel}$ transitions while the dynamic
initialization is not finished. Both types of transitions reset the
output of the protocol by leaving an agent in the default output state
$\q{f}$. The $\mathsf{reset}$ transition resets the output by moving
agents from $\q{t}$ to $\q{f}$.
\end{itemize}

\noindent Let $\PP_\infty = (Q, T_\infty, L, X, I , O)$. 
Let $T \defeq T_\infty \cup T_\dagger$. For the sake of readability, 
we will sometimes omit the subscripts and superscripts from transitions names 
when they are irrelevant, \eg\ ``a
$\mathsf{swap}$ transition is enabled'' instead of ``there exist $p, q
\in N$ and $x \in X$ such that $\mathsf{swap}_{p,q}^x$ is enabled''.\medskip

\parag{Size} Note that $\PP_{\text{thr}}$ has $|Q| = |X| + |N_X| + |B|
= |X| + (2n+3) \cdot (|X|+1) + 2 \in \O(\log \norm{\varphi} \cdot
|X|)$ states and $|L| = 2n+1 \in \O(\log \norm{\varphi})$
helpers. Moreover, since families of transitions are parameterized by
$X$, $B$, $N$ or $N^2$, and $\{+, -\}$, there are $\O(|N|^2 \cdot |X|)
= \O(\log^2 \norm{\varphi} \cdot |X|) \subseteq \O(|\varphi|^3)$
transitions. Finally, each transition uses at most
$\O(|\mathrm{rep}(\norm{\varphi})|) = \O(\log \norm{\varphi})
\subseteq \O(|\varphi|)$ states.

\medskip

\parag{Auxiliary definitions and observations} Before proving that
$\PP_{\text{thr}}$ works as intended, let us first introduce auxiliary
definitions. Let $\val \colon Q \to \N$ be the function that
associates a value to each state as follows:
\begin{align*}
  \val(\q{0}) &= \val(\q{f}) = \val(\q{t}) \defeq 0, \\
  \val(x) &\defeq \vec{a}(x)
  && \text{for every } x \in X, \\
  \val(\q{\circ2^i}) &\defeq \circ2^i
  && \text{for every } 0 \leq i \leq n \text{ and } \circ \in \{+, -\}, \\
  \val(q_x) &\defeq \val(q)
  && \text{for every } q \in N \text{ and } x \in X.
\end{align*}
\noindent So, for example, for the predicate $3x - 4y \geq 2$ we have $\val(x) = 3$ and
$\val(y) = -4$. 
For every configuration $C$ and every set of states $S \subseteq Q$, let
$$\val_S(C) \defeq \sum_{q \in S} \val(q) \cdot C(q).$$ In particular,
let $\val(C) \defeq \val_Q(C)$. Intuitively, $C$ can be seen as an encoding of the value
$\val(C)$. The following properties, relating
values and configurations, can be derived from the above definitions:
\begin{proposition}\label{prop:thr:prop}
  For every initialization sequence $\pi$ with effective input $\vec{w}$ 
  such that $L' \trans{\pi} C$ for some $L' \succeq L$, the following holds:
  \begin{enumerate}
  \item $\val(C) = \vec{a} \cdot \vec{w}$, \label{itm:val:inv}  
  
  \item $|C| = C(N) + C(B) + |\vec{w}|$,\label{itm:config:size} 
    
  \item $C(N) \geq L(N)$ and $C(B) \geq L(B)$, \label{itm:preservation} 
 
  \item $C(N_x) + C(x) = \vec{w}(x)$ for every $x \in
    X$. \label{itm:preservation:x}
  \end{enumerate}
\end{proposition}

\noindent In particular, \eqref{itm:config:size} states that the
number of agents is always equal to the number of helpers plus the
\emph{net} amount of agents that dynamically entered the population.

\parag{Auxiliary propositions} We say that a configuration $C$ is
\emph{clean} if the following holds for every $p, q \in
P_X$:\smallskip
\begin{itemize}
\item If $\val(p) + \val(q) = 0$, then $C(p) = 0$ or $C(q) = 0$. \\
For example, a configuration with agents in $\q{+2^i}_x$ and 
$\q{-2^i}_y$ is not clean, since $\val(\q{+2^i}_x) + \val(\q{-2^i}_y) = 0$.
Intuitively, no pair of agents can cancel in a clean configuration.\medskip

\item If $\val(p) = \val(q)$ and $\val(p) \not\in \{-2^n, +2^n\}$,
  then $C(\{p, q\}) \leq 1$. \\
For example, a configuration with two agents in $\q{+2^i}_x$, where $i < n$,
is not clean. Intuitively, in a clean configuration no agent can be promoted to
a higher power of 2.
\end{itemize}

We show that any configuration can be cleaned using only permanent transitions. 
This implies that once the dynamic initialization has terminated, every fair 
execution visits clean configurations infinitely often.

\begin{proposition}\label{prop:cleaning}
  For every initialization sequence $\pi$ such that $L' \trans{\pi} C$ for some 
  $L' \succeq L$, there exists a clean configuration $D$ 
  such that $C \trans{T_{\infty}^*} D$.
\end{proposition}

\begin{proof}
  If $C$ is clean, then we pick $D \defeq C$. Otherwise, at least one
  of the following holds:
  \begin{enumerate}[(a)]
    \item $C(p) > 0$, $C(q) > 0$ and $\val(p) + \val(q) = 0$ for some
      $p, q \in P_X$; \label{itm:clean:eq}

    \item $C(\{p, q\}) \geq 2$ for some $p, q \in P_X$ such that
      $\val(p), \val(q) \not\in \{-2^n, +2^n\}$. \label{itm:clean:up}
  \end{enumerate}
  
  We claim there exists a configuration $C'$ such that $C \trans{T_{\infty}^*}
  C'$ and $C(P_X) > C'(P_X)$. Let us show that if the claim is true then the result holds. If $C'$ is clean, then we are
  done. Otherwise, this process is repeated until a clean
  configuration $D$ has been reached. The process terminates as the
  number of agents in $P_X$ cannot become negative.

  Let us now prove the claim. Suppose~\eqref{itm:clean:eq}
  holds. By \cref{prop:thr:prop}~\eqref{itm:preservation}, we have
  $C(N) \geq L(N) \geq 2$ and hence it is possible to consecutively
  fire at least two $\mathsf{swap}$ transitions. Note that they do not
  change the amount of agents in $P_X$. For this reason, we may assume
  without loss of generality that $p = \q{+2^i}$ and $q = \q{-2^i}$
  for some $0 \leq i \leq
  n$. By \cref{prop:thr:prop}~\eqref{itm:preservation}, we have
  $C(B) \geq L(B) > 0$. Thus, there exists $r \in B$ such that $C(r) >
  0$. Therefore, firing transition $\mathsf{cancel}_{i, r}$ decreases
  $C(P_X)$ by two.

  Similarly, if case~\eqref{itm:clean:up} holds, then we may assume
  without loss of generality that $C(\q{\circ2^i}) \geq 2$ for some $0
  \leq i < n$ and $\circ \in \{+, -\}$. Thus, firing transition
  $\mathsf{up}_i^\circ$ decreases $C(P_X)$ by one.
\end{proof}

We now bound the number of agents in states from $X \cup P_X$ in a
clean configuration.

\begin{proposition} \label{prop:top:small}
  For every initialization sequence $\pi$ with effective input $\vec{w}$ 
  such that $L' \trans{\pi} C$ for some $L' \succeq L$, if $C$ is clean, 
  then $C(X) + C(P_X) \leq |\vec{w}| + n$.
\end{proposition}

\begin{proof}  
  Let $S_{\circ 2^n} \defeq \{q \in P_X : \val(q) = \circ2^n\}$ for
  both $\circ \in \{+, -\}$. Since $C$ is clean, we have $C(S_{\circ
    2^n}) = 0$ for some $\circ \in \{+, -\}$. Let us consider the case
  where $\circ = -$. The other case is proven analogously.

  Let $\vec{u} \in \N^X$ be such that $\vec{u}(x) \defeq C(x)$ for
  every $x \in X$. Note that $|\vec{u}| = C(X)$, and that $\vec{u}
  \leq \vec{w}$ by \cref{prop:thr:prop}~\eqref{itm:preservation:x}. Since $C$ is clean, we
  have $C(P_X \setminus S_{+2^n}) \leq n$. Thus, it suffices to show
  that $C(S_{+2^n}) \leq |\vec{w}| - |\vec{u}|$. Suppose this is not
  the case. This yields a contradiction:
  \begin{align*}
    \val_{P_X}(C)
    &> 2^n \cdot C(S_{+2^n}) - 2^n
    && \text{(since $C$ is clean)} \\
    &\geq 2^n \cdot (|\vec{w}| - |\vec{u}| + 1) - 2^n
    && \text{(by assumption)} \\
    &= 2^n \cdot (|\vec{w}| - |\vec{u}|) \\
    &\geq \vec{a} \cdot (\vec{w} - \vec{u}) && \text{(since $2^n >
      \norm{\vec{a}}$ and $\vec{w} \geq \vec{u}$)} \\    
    &= \vec{a} \cdot \vec{w} - \vec{a} \cdot \vec{u} \\
    &= \val(C) - \val_X(C)
    && \text{(by Prop.~\ref{prop:thr:prop}~\eqref{itm:val:inv} and
      def.\ of $\vec{u}$)} \\    
    &= \val_{P_X}(C)
    && \text{(by def.\ of $\val$)} \qedhere
  \end{align*}
\end{proof}

The following corollary shows that the number of agents in state
$\q{0}$ can always be increased back to at least $n$. This will later
be useful in arguing that the number of agents in $X$ can eventually
be decreased to zero.

\begin{corollary} \label{cor:enough:zero}
  For every initialization sequence $\pi$ with effective input $\vec{w}$ such 
  that $L' \trans{\pi} C$ for some $L' \succeq L$, there exists a clean 
  configuration $D$ such that $C \trans{T_{\infty}^*} D$ and $D(\q{0}) \geq n$.
\end{corollary}

\begin{proof}
  By \cref{prop:cleaning}, there exists a clean configuration $C'$ 
  such that $C \trans{\pi' \in T_\infty^*} C'$. Let us first prove that $C'(Z_X) \geq n$. 
  Note that $\pi\pi'$ is an initialization sequence with effective input 
  $\vec{w}$ such that $L' \trans{\pi\pi'} C'$. Thus:
  \begin{align*}
    C'(Z_X)
    &= |C'| - C'(X) - C'(P_X) - C'(B)
    && \text{(by def.\ of $Q$)} \\
    &= (C'(N) + C'(B) + |\vec{w}|) - C'(X) - C'(P_X) - C'(B)
    && \text{(by Prop.~\ref{prop:thr:prop}~\eqref{itm:config:size})} \\    
    &\geq (L(N) + C'(B) + |\vec{w}|) - C'(X) - C'(P_X) - C'(B)
    && \text{(by Prop.~\ref{prop:thr:prop}~\eqref{itm:preservation})} \\    
    &\geq (L(N) + C'(B) + |\vec{w}|) - (|\vec{w}| + n) - C'(B)
    && \text{(by Prop.~\ref{prop:top:small})} \\
    &= L(N) - n \\
    &\geq n
    && \text{(by def.\ of $L$)}
  \end{align*}
  Now, by \cref{prop:thr:prop}~\eqref{itm:preservation}, we have
  $C'(N) \geq L(N) \geq n$. Thus, using $\mathsf{swap}$ transitions, we
  can swap $n$ agents from $Z_X$ to $\q{0}$. This way, we obtain a
  configuration $D$ such that $C' \trans{T_{\infty}^*} D$ and $D(\q{0}) \geq
  n$. We are done since $\mathsf{swap}$ transitions preserve
  cleanness.
\end{proof}

For every configuration $C$, let
\begin{align*}
  \val^+(C) &\defeq \sum_{\substack{q \in P_X \\ \val(q) > 0}} \val(q)
  \cdot C(q), &
  \val^-(C) &\defeq \sum_{\substack{q \in P_X \\ \val(q) < 0}} \val(q)
  \cdot C(q).
\end{align*}
We now show that, once dynamic initialization has terminated, fair executions stabilize to configurations of a
certain ``normal form''.

\begin{proposition}\label{prop:threshold:stabilize}
  For every initialization sequence $\pi$ with effective input $\vec{w}$ 
  such that $L' \trans{\pi} C$ for some $L' \succeq L$ and for every fair execution 
  $\sigma$ of $\PP_\infty$ starting from $C$, 
  there exist $i \in \N$, $m_+ \geq 0$ and $m_- \leq 0$ such that:
  \begin{enumerate}
  \item $\sigma_i(X) = \sigma_{i+1}(X) = \cdots = 0$,
    
  \item $\val^\circ(\sigma_i) = \val^\circ(\sigma_{i+1}) = \cdots =
    m_\circ$ for both $\circ \in \{+, -\}$,

  \item $m_+ = 0 \lor m_- = 0$.
  \end{enumerate}
\end{proposition}

\begin{proof}
  For the sake of contradiction, assume there exist infinitely many
  indices $j$ such that $\sigma_j(X) > 0$. Let $j \in \N$ be such an
  index. By \cref{cor:enough:zero}, there exists a configuration $C_j$
  such that $\sigma_j \trans{T_{\infty}^*} C_j$ and $C_j(\q{0}) \geq n$. Hence,
  there exists $x \in X$ such that transition $\mathsf{add}_x$ is
  enabled in $C_j$. Since this holds for infinitely many indices and
  since $X$ is finite, fairness implies that some $\mathsf{add}$
  transition can be enabled infinitely often and hence occurs
  infinitely often along $\sigma$. This is impossible since the number
  of agents in $X$ cannot be increased by any transition in $T_\infty$, 
  and thus would eventually drop below zero. 
  Therefore, there exists $h \in \N$ such that
  $\sigma_h(X) = \sigma_{h+1}(X) = \cdots = 0$.

  Since $X$ is permanently empty from index $h$, the $\mathsf{add}$ transitions 
  are permanently disabled. No other transition in $T_\infty$ can increase 
  the absolute value of $\val^\circ$ for any $\circ \in \{+, -\}$. Thus, we have
  $|\val^\circ(\sigma_h)| \geq |\val^\circ(\sigma_{h+1})| \geq
  \cdots$ for both $\circ \in \{+, -\}$. Therefore, there exist $i
  \geq h$, $m_+ \geq 0$ and $m_- \leq 0$ such that
  \begin{align}
    \val^+(\sigma_i) &= \val^+(\sigma_{i+1}) = \cdots = m_+,\label{eq:stab1} \\
    \val^-(\sigma_i) &= \val^-(\sigma_{i+1}) = \cdots = m_-.\label{eq:stab2}
  \end{align}

  It remains to show that $m_+ = 0$ or $m_- = 0$. For the sake of
  contradiction, suppose this is not the case. For every $j \geq i$,
  \cref{cor:enough:zero} yields a configuration $C_j$ such that
  $\sigma_j \trans{T_{\infty}^*} C_j$ and $C_j(\q{0}) \geq n$. Thus, by fairness,
  there exist infinitely many indices $j \geq i$ such that
  $\sigma_j(\q{0}) \geq n$. Let $j$ be such an index. Let $0 \leq d,
  d' \leq n$ be the largest indices for which there exist states $q,
  q' \in P_X$ such that $\sigma_j(q) > 0$, $\sigma_j(q') > 0$,
  $\val(q) = 2^{d}$ and $\val(q') = -2^{d'}$. Note that these indices
  exist because $m_+ \neq 0$ and $m_- \neq 0$.

  Assume without loss of generality that $d \geq d'$, as the other
  case is symmetric. By \cref{prop:thr:prop}~\eqref{itm:preservation},
  there exists $r \in B$ such that $C(r) > 0$. Since $C_j(\q{0}) \geq
  n \geq d - d'$, the sequence of transitions $\mathsf{down}_d^+,
  \mathsf{down}_{d-1}^+, \ldots, \mathsf{down}_{d'+1}^+$ can be fired
  from $C_j$. From there, we can fire $\mathsf{cancel}_{d', r}$ which
  leads to a configuration $D_j$ such that $|\val^\circ(D_j)| <
  |m_\circ|$ for both $\circ \in \{+, -\}$. Since there are infinitely
  many such indices $j$, fairness implies that some such configuration
  $D_j$ occurs (infinitely often) along $\sigma$, which contradicts
  both~(\ref{eq:stab1}) and~(\ref{eq:stab2}).
\end{proof}

\parag{Main proof} We are now ready to prove that $\PP_{\text{thr}}$
works as intended.

\newcommand{\false}{\mathrm{false}}

\begin{theorem}
  $\PP_{\text{thr}}$ computes $\varphi$ with helpers and under
  reversible dynamic initialization.
\end{theorem}

\begin{proof}
  We first show that $\PP_{\text{thr}}$ is input reversible, and then
  that it correctly computes $\varphi$.
  
  \medskip\noindent\emph{Input reversibility.} Let $\pi$ be an
  initialization sequence with effective input $\vec{w}$ such that
  $L' \trans{\pi} C$ for some $L' \succeq L$. By
  \cref{prop:thr:prop}~\eqref{itm:preservation:x}, we have $C(I(x)) =
  C(x) = \vec{w}(x) - C(N_x) \leq \vec{w}(x)$ for every $x \in X$,
  which proves the first required property.

  For every configuration $C$, let $\false(C) \defeq D$ where $D(t)
  \defeq 0$, $D(f) \defeq C(t) + C(f)$ and $D(q) \defeq C(q)$ for
  every $q \in Q \setminus \{f, t\}$. Observe that for every
  configuration $C$, the following holds:
  \begin{align}
    C \trans{\mathsf{reset}^{C(\q{t})}} \false(C) \text{ and }
    \false(C) \in [C].\label{eq:reset:false}
  \end{align}
  It remains to show that if $C \trans{\pi_\triangleright \in T^*} D$
  and $D' \in [D]$, then $D' \trans{\pi_\triangleleft \in T^*} C'$ for
  some $C' \in \DiEq{C}$. By~\eqref{eq:reset:false}, it is enough to
  argue that $\false(D) \trans{T^*} \false(C)$.

  Let $C_i \trans{t_i} C_{i+1}$ be the $i^\text{th}$ step of
  $\pi_\triangleright$. We argue that $\false(C_{i+1}) \trans{t_i'}
  \false(C_i)$ for some $t_i' \in T \cup \{\varepsilon\}$. By
  induction, this implies $\false(D) \trans{T^*} \false(C)$ as
  desired. If $t_i$ is an $\mathsf{equal}$, $\mathsf{false}$ or
  $\mathsf{reset}$ transition, then we already have $\false(C_{i+1}) =
  \false(C_i)$. Otherwise we revert the step as follows, where ``$s
  \mapsto u$'' indicates that if $t_i = s$, then we reverse it with
  $t_i' = u$:
  \begin{align*}
    \mathsf{add}_x &\mapsto \mathsf{add}_{x,\q{f}}^{-1}
    && \text{for every $x \in X$,} \\
    \mathsf{add}_{x,q}^{-1} &\mapsto \mathsf{add}_x
    && \text{for every $x \in X$ and $q \in B$,} \\
    \mathsf{up}_i^\circ &\mapsto \mathsf{down}^\circ_{i+1}
    && \text{for every $0 \leq i < n$ and $\circ \in \{+, -\}$,}\\
    \mathsf{down}^\circ_i &\mapsto \mathsf{up}_{i-1}^\circ
    && \text{for every $0 < i \leq n$ and $\circ \in \{+, -\}$,} \\
    \mathsf{cancel}_{i,q} &\mapsto \mathsf{cancel}_{i,\q{f}}^{-1}
    && \text{for every $0 \leq i \leq n$ and $q \in B$,} \\
    \mathsf{cancel}_{i,q}^{-1} &\mapsto \mathsf{cancel}_{i,\q{f}}
    && \text{for every $0 \leq i \leq n$ and $q \in B$,}\\
    \mathsf{swap}_{p, q}^x &\mapsto \mathsf{swap}_{q, p}^x
    && \text{for every $p, q \in N$ and $x \in X$.}
  \end{align*}
  Note that $t_i'$ is not the exact reverse transition of $t_i$, as it
  may differ over $B$. Indeed, $t_i'$ may require an agent in state
  $\q{f}$, which may not have been produced by $t_i$. However, this is
  not an issue since, by definition of $\false$ and by
  \cref{prop:thr:prop}~\eqref{itm:preservation}, we have:
  $$\false(C_{i+1})(\q{f}) = C_{i+1}(B) \geq L(B) > 0.$$ Thus, we have
  $\false(C_{i+1}) \trans{t_i'} \false(C_i)$ as desired, which
  completes the proof.
	
  \medskip\noindent\emph{Correctness.} Let $\pi$ be an initialization
  sequence with effective input $\vec{w}$ such that $L' \trans{\pi} C$ 
  for some $L' \succeq L$. 
  Let $\sigma$ be a fair execution of $\PP_\infty$ starting from $C$. 
  By \cref{prop:threshold:stabilize}, there exist 
  $i \in \N$, $m_+ \geq 0$ and $m_- \leq 0$ such that:\smallskip
  \begin{enumerate}[(a)]
    \item $\sigma_i(X) = \sigma_{i+1}(X) = \cdots =
      0$, \label{itm:input:stab}

    \item $\val^\circ(\sigma_i) = \val^\circ(\sigma_{i+1}) = \cdots =
      m_\circ$ for both $\circ \in \{+, -\}$, \label{itm:val:stab}

    \item $m_+ = 0 \lor m_- = 0$.
  \end{enumerate}
  
  First, let us show that $O(\sigma) \in \{0, 1\}$. We make a case
  distinction on whether $m_+ \geq b$.\medskip

  \noindent\emph{Case $m_+ \geq b$.} We show that $O(\sigma) =
  1$. Note that $m_+ \geq b > 0$ which implies that $m_- = 0$. Thus,
  by~\eqref{itm:val:stab}, no $\mathsf{cancel}$ transition is enabled
  in $\sigma_j$ for every $j \geq i$. Thus, it suffices to show that
  $\sigma_j(\q{f}) = 0$ for some $j \geq i$. For the sake of
  contradiction, suppose this is not the case. Let $j \geq i$ be such
  that $\sigma_j(\q{f}) > 0$. We claim that $\sigma_j$ can reach some
  configuration $D_j$ enabling transition $\mathsf{equal}$,
  \ie\ larger or equal to $\mathrm{rep}(b)$. This claim, together with
  fairness, yields a contradiction since this transition can move all
  agents in state $\q{f}$ to state $\q{t}$.

  Let us prove the claim. By \cref{cor:enough:zero}, we have $\sigma_j
  \trans{T_\infty^*} C_j$ where $C_j(\q{0}) \geq n$. Note that $\val^+(C_j) =
  \val^+(\sigma_j)$. If $\val^+(C_j) = b$, then, by cleanness, $C_j$
  contains precisely the binary representation of $b$, and hence $C_j
  \geq \mathrm{rep}(b)$. Thus, assume $\val^+(C_j) > b$. Let $0 \leq d
  \leq n$ be the largest exponent for which there exists a state $q
  \in P_X$ such that $C_j(q) > 0$, $\val(q) = 2^d$ and $d \not\in
  \mathrm{bits}(b)$. Since $C_j(\q{0}) \geq n \geq d$, the sequence of
  transitions $\mathsf{down}_d^+ \cdot \mathsf{down}_{d-1}^+ \cdots
  \mathsf{down}_1^+$ can be fired from $C_j$, which yields a
  configuration $D_j \geq \mathrm{rep}(b)$.\medskip

  \noindent\emph{Case $m_+ < b$.} We show that $O(\sigma) = 0$. First note
  that $\mathsf{equal}$ is disabled in $\sigma_j$ for every
  $j \geq i$, as otherwise we would have
  $\sigma_j \geq \mathrm{rep}(b)$ which implies that $m_+
  = \val^+(\sigma_j') \geq b$. Thus, it suffices to show that there
  are infinitely many indices $j$ such that $\sigma_j(\q{f}) >
  0$. Indeed, if this is the case, then, by fairness, $\mathsf{false}$
  permanently moves all agents in state $\q{t}$ to state $\q{f}$.
  
  For the sake of contradiction, suppose the claim does not hold. 
  Let $\sigma' \defeq \pi\sigma$. Let $j \in \N$ be the largest index 
  such that $\sigma_j'(\q{f}) > 0$. 
  Note that this configuration exists as $\sigma_0' = L' \succeq L$ and $L(\q{f}) > 0$. 
  The only transition that reduces the number of agents in $\q{f}$ is $\mathsf{equal}$. 
  Thus, $\sigma_j' \trans{\mathsf{equal}} \sigma_{j+1}'$ and $\val^+(\sigma_j') \geq b$. 
  As finally, $m_+ < b$, some $\mathsf{cancel}$ or $\mathsf{add}^{-1}$ transition 
  must be fired in $\sigma'_{j'}$ for some $j' > j$. 
  In both cases there is afterwards an agent in state $\q{f}$. 
  This contradicts the maximality of $j$. \medskip

  We are done proving $O(\sigma) \in \{0, 1\}$. It remains to
  argue that $O(\sigma) = \varphi(\vec{w})$. We have:
  \begin{align*}
    \vec{a} \cdot \vec{w}
    &= \val(\sigma_i)
    && \text{(by Prop.~\ref{prop:thr:prop}~\eqref{itm:val:inv})}\\
    &= \val^+(\sigma_i) + \val^-(\sigma_i)
    && \text{(by~\eqref{itm:input:stab})} \\
    &= m_+ + m_-.
    && \text{(by~\eqref{itm:val:stab})} 
  \end{align*}
  Recall that $m_+ \geq 0$, $m_- \leq 0$ and $(m_+ = 0 \lor m_- =
  0)$. If $\vec{a} \cdot \vec{w} \geq b$, then we must have $m_+ \geq
  b > 0$ and $m_- = 0$. Therefore, the first case above holds, and
  hence $O(\sigma) = 1$, which is correct. If $\vec{a} \cdot \vec{w} <
  b$, then we must have $m_+ < b$. Therefore, the second case above
  holds, and hence $O(\sigma) = 0$, which is also correct.
\end{proof}

\subsubsection{Remainder protocols}\label{subsec:app-remainder}

This section describes a family of protocols with helpers computing
remainder predicates under reversible dynamic initialization. The
construction, its correctness proof and its intermediary propositions
are similar to those presented in \cref{subsec:app-threshold} for
the case of threshold predicates. For completeness, we repeat and
adapt them in full details.

Let us fix a remainder predicate $\varphi$ over variables $X$. 
Let $\varphi(\vec{v}) \defeq \vec{a} \cdot \vec{v} \equiv_m b$ 
where $\vec{a} \in \Z^X$ and $m \in \N_{\geq 2}$ and $b \in \Z$. 
Without loss of generality\footnote{If this is not the case for some coefficient
	 $\vec{a}(x)$, then we can replace it by $\vec{a}(x) \bmod m$, 
	 which yields an equivalent predicate.}, 
we may assume $0 \leq b < m$ and $0 \leq \vec{a}(x) < m$ for each $x \in X$.

Instead of directly constructing a protocol for $\varphi(\vec{v})$, we rewrite the predicate. This yields a different but equivalent predicate $\varphi'(\vec{v})$:
\begin{align*}
\varphi'(v) \defeq 
	\begin{cases}
		\vec{a} \cdot \vec{v} \not \geq 1~(\text{mod } m) 
			& \text{ if } b = 0, \\
		\vec{a} \cdot \vec{v} \geq b ~(\text{mod } m) \land \vec{a} \cdot \vec{v} \not\geq b + 1 ~(\text{mod } m) 
			& \text{ if } b > 0. \\
	\end{cases}
\end{align*}

As we can handle negations and conjunctions separately in \cref{subsec:finsets}, 
it is enough to describe a protocol for the predicate 
$\varphi(\vec{v}) \defeq \vec{a} \cdot \vec{v} \geq b~(\text{mod } m)$ where 
$\vec{a} \in \N^X$, $m \in \N_{\geq 2}$, $0 < b < m$ and $0 \leq \vec{a}(x) < m$ for each $x \in X$.

We construct a simple population protocol $\PP_{\text{rem}}$ with helpers that 
computes $\varphi$ under reversible dynamic initialization, and prove its correctness. 

\parag{Notation} Let $n$ be the smallest number such that $2^n >
\norm{\varphi}$. Let $P \defeq \{\q{2^i} : 0 \leq i \leq
n\}$, $Z \defeq \{\q{0}\}$, $N \defeq P \cup Z$ and $B \defeq \{\q{f},
\q{t}\}$, where $P$, $Z$, $N$ and $B$ respectively stand for
``$P$owers of two'', ``$Z$ero'', ``$N$umerical values'' and
``$B$oolean values''. For every set $S$ and every $x \in X$, let $S_x
\defeq \{q_x : q \in S\}$ and $S_X \defeq S \cup \bigcup_{x \in X}
S_x$. 

For every $d \in \N$, let $\mathrm{bits}(d)$ denote the unique set $J
\subseteq \N$ such that $d = \sum_{j \in J} 2^j$,
\eg\ $\mathrm{bits}(13) = \mathrm{bits}(1101_2) = \{3, 2, 0\}$. The \emph{canonical representation} of an integer $d \in \Z$ is the multiset $\mathrm{rep}(d)$ defined as follows:
$$
\mathrm{rep}(d) \defeq
\begin{cases}
\multiset{\q{2^i} : i \in \mathrm{bits}(d)}   & \text{if } d > 0, \\
\multiset{\q{0}}                              & \text{if } d = 0.
\end{cases}
$$

\parag{The protocol} The RDI-protocol $\PP_{\text{rem}} = (Q, T_\infty, T_\dagger, L, X, I, O)$ is defined as
follows:
\begin{itemize}
\setlength\itemsep{4pt}

  \item  $Q \defeq X \cup N_X \cup B$. \\
	Intuitively, the states of $X$ are the ``ports'' through which the agents for each variable
	enter and exit the protocol. 
	
	\item  $I \defeq x \mapsto \q{x}$. \\
	That is, the initial state for variable $x$ is $x$.
	
	\item $L \defeq \multiset{2n \cdot \q{0}, \q{f}}$. \\
	So, we have $2n$ helpers in state $\q{0}$, and one helper in state $\q{f}$, i.e., initially 
	the protocol assumes that the predicate does not hold.
	
	\item $O(q) \defeq q \mapsto (0 \text{ if } q = \q{f} \text{ else } 1
	\text{ if } q = \q{t} \text{ else } \bot)$. \\
	That is, the output of the protocol is completely determined by the number of agents 
	in states $\q{t}$ and $\q{f}$
	
	\item $T_\infty$ is the following set of (``permanent'') transitions:
	\begin{alignat*}{3}
		\mathsf{add}_{x}:\ &&
		\multiset{x, |\mathrm{rep}(\vec{a}(x))| \cdot \q{0}}
		& \mapsto \multiset{\q{0}_x} \mplus \mathrm{rep}(\vec{a}(x))
		&& \qquad \text{for every $x \in X$,} \\
		\mathsf{up}_i:\ &&
		\multiset{\q{2^i}, \q{2^i}}
		&  \mapsto \multiset{\q{2^{i+1}}, \q{0}}
		&& \qquad\text{for every $0 \leq i < n$,} \\
		\mathsf{down}_i:\ &&
		\multiset{\q{2^i}, \q{0}}
		&  \mapsto \multiset{\q{2^{i-1}}, \q{{2^{i-1}}}}
		&& \qquad\text{for every $0 < i \leq n$,} \\
		\mathsf{modulo}_q:\ &&
		\mathrm{rep}(b) \mplus \multiset{q}
		&  \mapsto \mathrm{rep}(b) \mplus \multiset{\q{f}}
		&& \qquad\text{for every $q \in B$,} \\
		\mathsf{swap}_{p, q}^x:\ &&
		\multiset{p, q_x}
		& \mapsto \multiset{p_x, q}
		&& \qquad \text{for every $p, q \in N$ and $x \in X$,} \\
		\mathsf{equal}:\ &&
		\mathrm{rep}(b) \mplus \multiset{\q{f}}
		& \mapsto \mathrm{rep}(b) \mplus \multiset{\q{t}}, \\
		\mathsf{false}:\ &&
		\multiset{\q{f}, \q{t}} & \mapsto \multiset{\q{f}, \q{f}}.
	\end{alignat*}
	Intuitively, $\mathsf{add}_{x}$ converts an agent which arrived via port $x$ into 
	the canonical representation of $\vec{a}(x)$. Transitions of the form $\mathsf{up}_i, 
	\mathsf{down}_i$ allow the protocol to change
	the representation of a value, without changing the value itself.
	The $\mathsf{modulo}$ transition reduces the overall value by $m$.
	Transition $\mathsf{equal}$ allows the protocol to detect that the current value 
	is at least $b$, which moves a helper from state $\q{f}$ to $\q{t}$.
	\item Finally, $T_{\dagger}$ is the following set of RDI-transitions:
	\begin{alignat*}{3}
		\mathsf{add}_{x,q}^{-1}:\ &&
		\multiset{\q{0}_x, q} \mplus \mathrm{rep}(\vec{a}(x))
		& \mapsto \multiset{x, \q{f}, |\mathrm{rep}(\vec{a}(x))| \cdot \q{0}}
		&& \quad \text{for every $x \in X$ and $q \in B$,} \\
		\mathsf{modulo}_q^{-1}:\ &&
		\mathrm{rep}(b) \mplus \multiset{q}
		&  \mapsto \mathrm{rep}(b) \mplus \multiset{\q{f}}
		&& \quad\text{for every $q \in B$,} \\
		\mathsf{reset}:\ &&
		\multiset{\q{t}}
		& \mapsto \multiset{\q{f}}.
	\end{alignat*}
	The first two transitions are needed to reverse the changes of
	$\mathsf{add}$ and $\mathsf{modulo}$ transitions while the dynamic
	initialization is not finished. Both types of transitions reset the output of the protocol
	by leaving an agent in the default output state $\q{f}$. The
	$\mathsf{reset}$ transition resets the output by moving agents from
	$\q{t}$ to $\q{f}$.
\end{itemize} 

\noindent Let $\PP_\infty = (Q, T_\infty, L, X, I , O)$. 
Let $T \defeq T_\infty \cup T_\dagger$. For the sake of readability, 
we will sometimes omit the subscripts and superscripts from transitions names 
when they are irrelevant, \eg\ ``a
$\mathsf{swap}$ transition is enabled'' instead of ``there exist $p, q
\in N$ and $x \in X$ such that $\mathsf{swap}_{p,q}^x$ is enabled''.\medskip

\parag{Size} Note that $\PP_{\text{rem}}$ has $|Q| = |X| + |N_X| + |B|
= |X| + (n+2) \cdot (|X|+1) + 2 = \O(\log \norm{\varphi} \cdot |X|)$
states and $|L| = 2n+1 = \O(\log \norm{\varphi})$ helpers. Moreover,
since families of transitions are parameterized by $X$, $B$ and $N$ or $N^2$,
there are $\O(|N|^2 \cdot |X|) = \O(\log^2 \norm{\varphi} \cdot |X|) 
\subseteq \O(|\varphi|^3)$ transitions. Finally, each transition uses at most
$\O(|\mathrm{rep}(\norm{\varphi})|) = \O(\log \norm{\varphi})
\subseteq \O(|\varphi|)$ states.\medskip

\parag{Auxiliary definitions and observations} Before proving that
$\PP_{\text{rem}}$ works as intended, let us first introduce auxiliary
definitions. Let $\val \colon Q \to \N$ be the function that
associates a value to each state as follows:
\begin{align*}
	\val(\q{0}) &= \val(\q{f}) = \val(\q{t}) \defeq 0, \\
	\val(x) &\defeq \vec{a}(x)
	&& \text{for every } x \in X, \\
	\val(\q{2^i}) &\defeq 2^i
	&& \text{for every } 0 \leq i \leq n, \\
	\val(q_x) &\defeq \val(q)
	&& \text{for every } q \in N \text{ and } x \in X.
\end{align*}
\noindent So, for example, for the predicate $5x + 6y \geq 4\ (\bmod\ m)$ we have $\val(x) = 5$ and
$\val(y) = 6$. 
For every configuration $C$ and set of states $S \subseteq Q$, let
$$\val_S(C) \defeq \sum_{q \in S} \val(q) \cdot C(q).$$ 
In particular, let $\val(C) \defeq \val_Q(C)$. 
Intuitively, $C$ can be seen as an encoding of the value $\val(C)$. 
The following properties, relating values and configurations, can be derived from 
the above definitions:
\begin{proposition}\label{prop:rem:prop}
	For every initialization sequence $\pi$ with effective input $\vec{w}$ such 
	that $L' \trans{\pi} C$ for some $L' \succeq L$, the following holds:
	\begin{enumerate}
		\item $\val(C) \equiv_m \vec{a} \cdot \vec{w}$, \label{itm:rem:val:inv}
		
		\item $\val(C) \leq \vec{a} \cdot \vec{w}$, \label{itm:rem:val:dec}
		
		\item $|C| = C(N) + C(B) + |\vec{w}|$, \label{itm:rem:config:size}
		
		\item $C(N) \geq L(N)$ and $C(B) \geq L(B)$, \label{itm:rem:preservation}
		
		\item $C(N_x) + C(x) = \vec{w}(x)$ for every $x \in
		X$. \label{itm:rem:preservation:x}
	\end{enumerate}
\end{proposition}

\parag{Auxiliary propositions} We say that a configuration $C$ is
\emph{clean} if for every $p, q \in P_X$ with $\val(p) = \val(q)$ and $\val(p) 
\not= 2^n$, it holds that $C(\{p, q\}) \leq 1$. Intuitively, in a clean configuration 
no agent can be promoted to a higher power of 2.

We show that any configuration can be cleaned using only permanent transitions. 
This implies that once the dynamic initialization has terminated, every fair 
execution visits clean configurations infinitely often.

\begin{proposition}\label{prop:rem:cleaning}
	For every initialization sequence $\pi$ such that $L' \trans{\pi} C$ 
	for some $L' \succeq L$, there exists a clean configuration $D$ 
	such that $C \trans{T_{\infty}^*} D$.
\end{proposition}
\begin{proof}
	If $C$ is clean, then we pick $D \defeq C$. Otherwise, we claim there 
	exists a configuration $C'$ such that $C \trans{*} C'$ and $C(P_X) > 
	C'(P_X)$. If $C'$ is clean, then we are	done. Otherwise, this process
	is repeated until a clean configuration $D$ has been reached. This must 
	terminate as the number of agents in $P_X$ cannot become negative.
	
	Let us prove the claim. If $C$ is not clean, then $C(\{p, q\}) \geq 2$
	for some $p, q \in P_X$ such that $\val(p) = \val(q)$ and $\val(p) 
	\not= 2^n$. By \cref{prop:rem:prop}~\eqref{itm:rem:preservation}, we have
	$C(N) \geq L(N) \geq 2$ and hence it is possible to consecutively
	fire at least two $\mathsf{swap}$ transitions. Note that they do not
	change the amount of agents in $P_X$. For this reason, we may assume
	without loss of generality that $p = q = \q{2^i}$ for some $0 \leq i \leq
	n$. Therefore, firing transition $\mathsf{up}_i$ decreases
	$C(P_X)$ by one.
\end{proof}

We now bound the number of agents in states from $X \cup P_X$ in a
clean configuration.

\begin{proposition} \label{prop:rem:top:small}
	For every initialization sequence $\pi$ with effective input $\vec{w}$ 
	such that $L' \trans{\pi} C$ for some $L' \succeq L$, if $C$ is clean, 
	then $C(X) + C(P_X) \leq |\vec{w}| + n$.
\end{proposition}

\begin{proof}  
	Let $S_{2^n} \defeq \{q \in P_X : \val(q) = 2^n\}$. Let $\vec{u} \in \N^X$ 
	be such that $\vec{u}(x) \defeq C(x)$ for every $x \in X$. 
	Note that $|\vec{u}| = C(X)$, and that $\vec{u} \leq \vec{w}$ by 
	\cref{prop:rem:prop}~\eqref{itm:rem:preservation:x}. 
	Since $C$ is clean, we have $C(P_X \setminus S_{2^n}) \leq n$. 
	Thus, it suffices to show	that $C(S_{2^n}) \leq |\vec{w}| - |\vec{u}|$. 
	Suppose this is not	the case. This yields a contradiction:
	\begin{align*}
		\val_{P_X}(C)
		&\geq 2^n \cdot C(S_{2^n}) \\
		&> 2^n \cdot (|\vec{w}| - |\vec{u}|)
		&& \text{(by assumption)} \\
		&\geq \vec{a} \cdot (\vec{w} - \vec{u}) && \text{(since $2^n >
			\norm{\vec{a}}$ and $\vec{w} \geq \vec{u}$)} \\    
		&= \vec{a} \cdot \vec{w} - \vec{a} \cdot \vec{u} \\
		&\geq \val(C) - \val_X(C)
		&& \text{(by Prop.~\ref{prop:rem:prop}~\eqref{itm:rem:val:dec} and
			def.\ of $\vec{u}$)} \\    
		&= \val_{P_X}(C)
		&& \text{(by def.\ of $\val$)} \qedhere
	\end{align*}
\end{proof}

The following corollary shows that the number of agents in state
$\q{0}$ can always be increased back to at least $n$. This will later
be useful in arguing that the number of agents in $X$ can eventually
be decreased to zero.

\begin{corollary} \label{cor:rem:enough:zero}
	For every initialization sequence $\pi$ with effective input $\vec{w}$ such 
	that $L' \trans{\pi} C$ for some $L' \succeq L$, there exists a clean 
	configuration $D$ such that $C \trans{T_{\infty}^*} D$ and $D(\q{0}) \geq n$.
\end{corollary}

\begin{proof}
	By \cref{prop:rem:cleaning}, there exist a clean configuration $C'$ such 
	that $C \trans{\pi' \in T_\infty^*} C'$. Let us first prove that $C'(Z_X) \geq n$. 
	Note that $\pi\pi'$ is an initialization sequence with effective input 
	$\vec{w}$ such that $L \trans{\pi\pi'} C'$. Thus:
	\begin{align*}
		C'(Z_X)
		&= |C'| - C'(X) - C'(P_X) - C'(B)
		&& \text{(by def.\ of $Q$)} \\
		&= (C'(N) + C'(B) + |\vec{w}|) - C'(X) - C'(P_X) - C'(B)
		&& \text{(by Prop.~\ref{prop:rem:prop}~\eqref{itm:rem:config:size})} \\    
		&\geq (L(N) + C'(B) + |\vec{w}|) - C'(X) - C'(P_X) - C'(B)
		&& \text{(by Prop.~\ref{prop:rem:prop}~\eqref{itm:rem:preservation})} \\    
		&\geq (L(N) + C'(B) + |\vec{w}|) - (|\vec{w}| + n) - C'(B)
		&& \text{(by Prop.~\ref{prop:rem:top:small})} \\
		&= L(N) - n \\
		&\geq n
		&& \text{(by def.\ of $L$)}
	\end{align*}
	Now, by \cref{prop:rem:prop}~\eqref{itm:rem:preservation}, we have
	$C'(N) \geq L(N) \geq n$. Thus, using $\mathsf{swap}$ transitions, we
	can swap $n$ agents from $Z_X$ to $\q{0}$. This way, we obtain a
	configuration $D$ such that $C' \trans{T_{\infty}^*} D$ and $D(\q{0}) \geq
	n$. We are done since $\mathsf{swap}$ transitions preserve
	cleanness.
\end{proof}
We now show that, once dynamic initialization has terminated, fair executions 
stabilize to configurations of a certain ``normal form''.

\begin{proposition}\label{prop:remainder:stabilize}
	For every initialization sequence $\pi$ with effective input $\vec{w}$ such that 
	$L' \trans{\pi} C$ for some $L' \succeq L$ and for every fair execution 
	$\sigma$ of $\PP_\infty$ starting from $C$, there exist 
	$i \in \N, r \geq 0$ such that
	\begin{enumerate}
		\item $\sigma_i(X) = \sigma_{i+1}(X) = \cdots = 0$,
		
		\item $\val(\sigma_i) = \val(\sigma_{i+1}) = \cdots =
		r < m$.
	\end{enumerate}
\end{proposition}

\begin{proof}
	For the sake of contradiction, assume there exist infinitely many
	indices $j$ such that $\sigma_j(X) > 0$. Let $j \in \N$ be such an
	index. By \cref{cor:rem:enough:zero}, there exists a configuration $C_j$
	such that $\sigma_j \trans{T_{\infty}^*} C_j$ and $C_j(\q{0}) \geq n$. Hence,
	there exists $x \in X$ such that transition $\mathsf{add}_x$ is
	enabled in $C_j$. Since this holds for infinitely many indices and
	since $X$ is finite, fairness implies that some $\mathsf{add}$
	transition can be enabled infinitely often and hence occurs
	infinitely often along $\sigma$. This is impossible since the number
	of agents in $X$ cannot be increased by any transition in $T_\infty$, 
	and thus would eventually drop below zero. 
	Therefore, there exists $h \in \N$ such that
	$\sigma_h(X) = \sigma_{h+1}(X) = \cdots = 0$.
	
	Transitions $\mathsf{modulo_{\q{t}}}$ and $\mathsf{modulo_{\q{f}}}$ 
	reduce the value of a configuration by $m > 0$. 
	$\mathsf{add}$ transitions are disabled in every configuration $\val(\sigma_{i})$ 
	with $i \geq h$ and all other transitions in $T_\infty$ 
	do not change the value of a configuration. 
	Moreover, the value of a configuration is always non-negative. 
	Thus, there exist $i \geq h$, $r \geq 0$ such that 
	$\val(\sigma_{i}) =\val(\sigma_{i+1}) = \cdots = r$. 
	
	For the sake of contradiction, assume that $r \geq m$. 
	As the execution is infinite but there are only finitely many 
	different configurations for a fixed number of agents, there 
	exists a configuration $D$ with $D(X) = 0$ and $\val(D) = r$ 
	that occurs infinitely often in $\sigma$. 
	We claim that $D$ can reach a configuration that enables a 
	$\mathsf{modulo}$ transition. 
	This claim, together with fairness, yields a contradiction because 
	the overall value would drop below $r$. 
	
	Let us prove the claim. By \cref{cor:rem:enough:zero}, we have $D \trans{*} D'$ 
	where $D'$ is clean and $D'(\q{0}) \geq n$. 
	Furthermore, $D'(X) = 0$ as $D(X) = 0$ and the number of agents in $X$ cannot be increased. 
	By \cref{prop:rem:prop} \eqref{itm:rem:preservation}, we have $D'(B) \geq L(B) = 1$. 
	Thus, it suffices to show that $D' \trans{T_{\infty}^*} D''$ for some $D'' \geq \mathrm{rep}(m)$.
        
	If $\val(D') = m$, then $D'$ contains precisely the binary representation of $m$, 
	because $D'$ is clean and $D'(X) = 0$. 
	Hence, $D' \geq \mathrm{rep}(m)$. Thus, assume $\val(D') > m$. 
	Let $0 \leq d \leq n$ be the largest exponent for which there exists a state 
	$q \in P_X$ such that $D(q) > 0$, $\val(q) = 2^d$ and $d \not\in \mathrm{bits}(m)$. 
	Since $D'(\q{0}) \geq n \geq d$, the sequence of transitions $\mathsf{down}_d \cdot 
	\mathsf{down}_{d-1} \cdots \mathsf{down}_1$ can be fired from $D'$, which yields a 
	configuration $D'' \geq \mathrm{rep}(m)$.
\end{proof}

\parag{Main proof} We are now ready to prove that $\PP_{\text{rem}}$
works as intended.

\begin{theorem}
	$\PP_{\text{rem}}$ computes $\varphi$ under reversible dynamic initialization.
\end{theorem}

\begin{proof}
	We first show that $\PP_{\text{rem}}$ is input reversible, and then
	that it correctly computes $\varphi$.
	
	\medskip\noindent\emph{Input reversibility.} Let $\pi$ be an
	initialization sequence with effective input $\vec{w}$ such that
	$L' \trans{\pi} C$ for some $L' \succeq L$. By
	\cref{prop:rem:prop}~\eqref{itm:rem:preservation:x}, we have $C(I(x)) =
	C(x) = \vec{w}(x) - C(N_x) \leq \vec{w}(x)$ for every $x \in X$,
	which proves the first required property.
	
	For every configuration $C$, let $\false(C) \defeq D$ where $D(t)
	\defeq 0$, $D(f) \defeq C(t) + C(f)$ and $D(q) \defeq C(q)$ for
	every $q \in Q \setminus \{f, t\}$. Observe that for every
	configuration $C$, the following holds:
	\begin{align}
		C \trans{\mathsf{reset}^{C(\q{t})}} \false(C) \text{ and }
		\false(C) \in [C].\label{eq:rem:reset:false}
	\end{align}
	It remains to show that if $C \trans{\pi_\triangleright \in T^*} D$
	and $D' \in [D]$, then $D' \trans{\pi_\triangleleft \in T^*} C'$ for
	some $C' \in \DiEq{C}$. By~\ref{eq:rem:reset:false}, it is enough to
	argue that $\false(D) \trans{T^*} \false(C)$.
	
	Let $C_i \trans{t_i} C_{i+1}$ be the $i^\text{th}$ step of
	$\pi_\triangleright$. We argue that $\false(C_{i+1}) \trans{t_i'}
	\false(C_i)$ for some $t_i' \in T \cup \{\varepsilon\}$. By
	induction, this implies $\false(D) \trans{T^*} \false(C)$ as
	desired. If $t_i$ is an $\mathsf{equal}$, $\mathsf{false}$ or
	$\mathsf{reset}$ transition, then we already have $\false(C_{i+1}) =
	\false(C_i)$. Otherwise we revert the step as follows, where ``$s
	\mapsto u$'' indicates that if $t_i = s$, then we reverse it with
	$t_i' = u$:
	\begin{align*}
		\mathsf{add}_x &\mapsto \mathsf{add}_{x,\q{f}}^{-1}
		&& \text{for every $x \in X$,} \\
		\mathsf{add}_{x,q}^{-1} &\mapsto \mathsf{add}_x
		&& \text{for every $x \in X$ and $q \in B$,} \\
		\mathsf{up}_i &\mapsto \mathsf{down}_{i+1}
		&& \text{for every $0 \leq i < n$,}\\
		\mathsf{down}_i &\mapsto \mathsf{up}_{i-1}
		&& \text{for every $0 < i \leq n$,} \\
		\mathsf{modulo}_{q} &\mapsto \mathsf{modulo}_{\q{f}}^{-1}
		&& \text{for every $q \in B$,} \\
		\mathsf{modulo}_{q}^{-1} &\mapsto \mathsf{modulo}_{\q{f}}
		&& \text{for every $q \in B$,}\\
		\mathsf{swap}_{p, q}^x &\mapsto \mathsf{swap}_{q, p}^x
		&& \text{for every $p, q \in N$ and $x \in X$.}
	\end{align*}
	Note that $t_i'$ is not the exact reverse transition of $t_i$, as it
	may differ over $B$. Indeed, $t_i'$ may require an agent in state
	$\q{f}$, which may not have been produced by $t_i$. However, this is
	not an issue since, by definition of $\false$ and by
	\cref{prop:thr:prop}~\eqref{itm:preservation}, we have:
	$$\false(C_{i+1})(\q{f}) = C_{i+1}(B) \geq L(B) > 0.$$ Thus, we have
	$\false(C_{i+1}) \trans{t_i'} \false(C_i)$ as desired, which
	completes the proof.

	\medskip\noindent\emph{Correctness.} Let $\pi$ be an initialization
	sequence with effective input $\vec{w}$ such that $L' \trans{\pi} C$ 
	for some $L' \succeq L$. 
	Let $\sigma$ be a fair execution of $\PP_\infty$ starting from $C$. 
	By \cref{prop:remainder:stabilize}, there exist $i \in \N, r \geq 0$ such that:
	\begin{enumerate}[(a)]
		\item $\sigma_i(X) = \sigma_{i+1}(X) = \cdots = 0$, \label{itm:rem:input:stab}
		
		\item $\val(\sigma_i) = \val(\sigma_{i+1}) = \cdots =
		r < m$. \label{itm:rem:val:stab}
	\end{enumerate}
	
	Let us first show that $O(\sigma) \in \{0, 1\}$. We make a case distinction on 
	whether $v \geq b$.\medskip
	
	\noindent\emph{Case $r \geq b$.} We show that $O(\sigma) = 1$. Note that $r < m$. 
	The $\mathsf{modulo}$ transitions reduce the overall value by $m$. 
	As the value of a configuration is never negative, no $\mathsf{modulo}$ transition 
	can be fired again. 
	Thus, it suffices to show that $\sigma_j(\q{f}) = 0$ for some $j \geq i$. 
	If $\sigma_i(\q{f}) = 0$ then we are done. 
	For the sake of contradiction, suppose this is not the case. 
	As $\sigma$ is infinite but there are only finitely many different configurations 
	for a fixed number of agents, there exists a configuration $D$ that occurs infinitely 
	often in $\sigma$ such that $D(X) = 0$, $\val(D) = r$ and $D(\q{f}) > 0$. 
	We claim that $D$ can reach a configuration that enables the transition $\mathsf{equal}$. 
	This claim, together with fairness, yields a contradiction because $\mathsf{equal}$ 
	can move all agents form state $\q{f}$ to state $\q{t}$.
	
	Let us prove the claim. 
	By \cref{cor:rem:enough:zero}, we have $D \trans{T_{\infty}^*} D'$ where $D'$ is clean and 
	$D'(\q{0}) \geq n$. 
	Furthermore, $D'(X) = 0$ as $D(X) = 0$ and the number of agents in $X$ cannot be 
	increased by transitions in $T_\infty$.  
	We show that $D' \trans{*} D''$ for some $D'' \geq \mathrm{rep}(m)$. 
	If $\val(D') = m$, then $D'$ contains precisely the binary representation of $m$, 
	because $D'$ is clean and $D'(X) = 0$. 
	Hence, $D' \geq \mathrm{rep}(m)$. Thus, assume $\val(D') > m$. 
	Let $0 \leq d \leq n$ be the largest exponent for which there exists a state $q \in P_X$ 
	such that $D(q) > 0$, $\val(q) = 2^d$ and $d \not\in \mathrm{bits}(m)$. 
	Since $D'(\q{0}) \geq n \geq d$, the sequence of transitions 
	$\mathsf{down}_d \cdot \mathsf{down}_{d-1} \cdots \mathsf{down}_1$ can be fired from $D'$, 
	which yields a configuration $D'' \geq \mathrm{rep}(m)$. 
	If $D''(\q{f}) = 0$, then the claim holds because the only transition that reduces the 
	number of agents in state $\q{f}$ is transition $\mathsf{equal}$. 
	If $D''(\q{f}) > 0$, then the claim holds because $D''$ enables $\mathsf{equal}$. \medskip
	
	\noindent\emph{Case $r < b$.} We show that $O(\sigma) = 0$. First note
	that $\mathsf{equal}$ is disabled in $\sigma_j$ for every
	$j \geq i$, as otherwise we would have
	$\sigma_j \geq \mathrm{rep}(b)$ which implies that $r
	= \val(\sigma_j) \geq b$. Thus, it suffices to show that there
	are infinitely many indices $j$ such that $\sigma_j(\q{f}) >
	0$. Indeed, if this is the case, then, by fairness, $\mathsf{false}$
	permanently moves all agents in state $\q{t}$ to state $\q{f}$.
	
	For the sake of contradiction, suppose the claim does not hold. 
	Let $\sigma' \defeq \pi\sigma$. Let $j \in \N$ be the largest index 
	such that $\sigma_j'(\q{f}) > 0$. 
	Note that this configuration exists as $\sigma_0' = L' \succeq L$ and $L(\q{f}) > 0$. 
	The only transition that reduces the number of agents in $\q{f}$ is $\mathsf{equal}$. 
	Thus, $\sigma_j' \trans{\mathsf{equal}} \sigma_{j+1}'$ and $\val(\sigma_j') \geq b$. 
	As finally, $r < b$, some $\mathsf{modulo}$ or $\mathsf{add}^{-1}$ transition 
	must be fired in $\sigma'_{j'}$ for some $j' > j$. 
	In both cases there is afterwards an agent in state $\q{f}$. 
	This contradicts the maximality of $j$. \medskip
	
	We are done proving $O(\sigma) \in \{0, 1\}$. 
	It remains to argue that $O(\sigma) = \varphi(\vec{w})$.
	We have $m > r = \val(\sigma_i) \equiv_m \vec{a} \cdot \vec{w}$ 
	by \cref{prop:rem:prop}~\eqref{itm:rem:val:inv}. 
	If $\vec{a} \cdot \vec{w} \geq b~(\text{mod } m)$, then $r \geq b$ and 
	hence $O(\sigma) = 1$, which is correct. 
	If $\vec{a} \cdot \vec{w} < b~(\text{mod } m)$, then $r < b$ and 
	hence $O(\sigma) = 0$, which is also correct.
\end{proof}


\section{Proofs of Section \ref{sec:small}: Protocols for small populations} \label{app:small}

\subsection{Proof of Theorem \ref{thm:removing:leaders}}
\propRemovingLeaders*

The proof proceeds in two steps. Lemma  \ref{lem:fixed-size} shows that, under the assumptions of the proposition, there is a protocol with one leader computing
$\varphi$ for all small populations. Lemma \ref{lemma:kill-the-leader} shows how to transform this protocol into a leaderless one. The bound on the number
of states follows directly from the composition of the bounds given in the lemmas.

\begin{lemma}\label{lem:fixed-size}
Let $\varphi$ be a predicate over a set of variables $X$ and let $\cutoffvar \in \N$.
Assume that for every $i \in \{2,3, \ldots, \ell-1\}$, there exists a protocol with at most one leader and at most $m$ states that computes $\cond{\varphi}{i}$.
Then there exists a protocol with one leader and $\O(\cutoffvar \cdot m \cdot |X|)$ states that computes
$\vec{x} < \cutoffvar \rightarrow \varphi(\vec{x})$.
\end{lemma}

\begin{proof}
 Let $\PP_1, \PP_2, \ldots, \PP_\cutoffvar$ be such that each $\PP_i$ computes $\cond{\varphi}{i}$. Without loss of generality, assume the states of the $\PP_i$ to be pairwise disjoint.
We construct a protocol $\PP=(Q, T, L, X, I, O)$ with one leader and $\O(\cutoffvar \cdot m \cdot |X|)$ states that computes $\vec{x} \leq \cutoffvar \rightarrow \varphi(\vec{x})$. Intuitively, the protocol $\PP$ works as follows: the leader stores a lower-bound estimate of the current population size. When the leader meets a new agent it has not met, the leader increments its estimate. Whenever the estimate changes to some value $i$, the leader resets $i$ agents in the population to initial of $\PP_i$ and lets the agents simulate the computation of $\PP_i$. When the estimate reaches $\cutoffvar$, the leader knows that the precondition $|X| < \cutoffvar$ is not satisfied, and it converts every agent to $\top$, a state that converts any other state to $\top$, thus yielding a stable $1$-consensus. The agents' states are annotated with their initial input, which allows the leader to reset states to the correct value.

    \parag{States and associated mappings} Let $Q_i \defeq X \times Q^{\PP_i}$ for every $i \in [\cutoffvar]$, where $Q^{\PP_i}$ denotes the states of $\PP_i$. The leader assumes a state from  the \emph{leader states} defined as $Q_L \defeq \{0, 1, \ldots, \cutoffvar\} \times \{0, 1\}$. The states of $\PP$ are defined as:
    \begin{align*}
    Q \defeq X \cup (X \times (Q_1 \cup \ldots \cup Q_{\cutoffvar-1})) \cup \{\top\} \cup Q_L.
    \end{align*}
    For the size of the protocol we thus have
    \[|Q| = |X| + |(X \times (Q_1 \cup \ldots \cup Q_{\cutoffvar-1}))| + 1 + |Q_L| \leq |X| + |X| \cdot m \cdot l + 1 + 2\cdot \cutoffvar, \] which is in $\O(\cutoffvar \cdot m \cdot |X|)$.

    We set the leader multiset to:
    \begin{align*}
    L \defeq \multiset{(0,0)}.
    \end{align*}
    The input mapping $I$ is defined as the identity function. The output mapping is given by:
    \begin{align*}
    O(x)      & \defeq 0 &&  \text{ for every } x \in X,\\
    O((x, q)) & \defeq O^{\PP_i}(q) && \text{ for every } i \in [\cutoffvar] \text{ and every } (x, q) \in Q_i, \\
    O((i, b)) & \defeq b && \text{ for every } (i, b) \in Q_L,\\
    O(\top) & \defeq 1. &&
    \end{align*}

    \parag{Transitions} The set of transitions $T$ of $\PP$ is given by
    $T \defeq T_\top  \cup T_\mathsf{sim} \cup T_\mathsf{incr}$ where
    $T_\top$,  $T_\mathsf{sim}$, and $T_\mathsf{incr}$ are defined as follows.
    \begin{itemize}
    \item
    $T_\top$ contains precisely the transitions:
    \begin{alignat*}{3}
        \mathsf{true}:\ &&
        \multiset{\top, q}
        & \mapsto \multiset{\top, \top}
        && \qquad \text{for every $q \in Q$,} \\
        \mathsf{threshold}:\ &&
        \multiset{(\cutoffvar, b), q}
        & \mapsto \multiset{\top, \top}
        && \qquad \text{for every $b \in \{0, 1\}, q \in Q$.}
    \end{alignat*}
    Intuitively, $T_\top$ contains transitions that ensure stabilization to $1$ if  $|X| < \cutoffvar$ is not satisfied: $\mathsf{threshold}$ initiates converting everyone to $\top$ as soon as the threshold $\cutoffvar$ is reached in the leader agent. The transitions $\mathsf{true}$ then convert everyone to $\top$.

    \item
    $T_\mathsf{sim}$ is given by $\bigcup_{i \in [\cutoffvar]} T_{\mathsf{sim}_i}$, and
    $T_{\mathsf{sim}_i}$ contains precisely the following transitions for every $x, y \in X$ and every $t \colon (\multiset{q, r}, \multiset{q', r'}) \in T^{\PP_i}$:
    \begin{alignat*}{3}
        \mathsf{sim}:\ &&
        \multiset{(x, q), (y, r)}
        & \mapsto \multiset{(x, q'), (y, r')}
        && \qquad \text{}
    \end{alignat*}
    Intuitively, the transitions in $T_{\mathsf{sim}_i}$ simulate the transitions of the individual protocols $\PP_i$ in $\PP$.

    \item

    $T_\mathsf{conv}$ contains precisely the following transitions for every $x \in X$:
    \begin{alignat*}{3}
        \mathsf{incr}_i:\ &&
        \multiset{x, (i, b)} &
        \mapsto \multiset{(x, I^{\PP_{i+1}}(x)), (i+1, b)}
        && \qquad \text{for every $i \in [1, \cutoffvar-1], b \in \{0, 1\}$},\\
        \mathsf{conv}_i:\ &&
        \multiset{(x, q), (i, b)} &
        \mapsto \multiset{(x, I^{\PP_{i}}(x)), (i, b)}
        && \qquad \text{for every $i \in [1, \cutoffvar-1], q \not \in Q_i,$} \\
        \mathsf{bool}_i:\ &&
        \multiset{(x, q), (i, b)} &
        \mapsto \multiset{(x, q), (i, O(q))}
        && \qquad \text{for every $i \in [1, \cutoffvar-1], q \in Q_i$}.
    \end{alignat*}
    Intuitively, the transitions in $T_\mathsf{conv}$ implement interactions with the leader whose role is to convert every agent to the current protocol: $\mathsf{incr}_i$ and $\mathsf{conv}_i$ take care of converting agents to the next protocol, while  the transitions $\mathsf{bool}_i$ convert the leader's opinion to the opinion of the current protocol.
    \end{itemize}
\parag{Correctness}
    Before we prove correctness of $\PP$, we state without proof some propositions that follow by inspection of the transitions of $\PP$:
    \begin{proposition}
        \label{prop:inv_top}
        For every $C, C' \in \N^Q$, the following invariant holds:
        If $C \trans{} C'$, then $C'(\top) \geq C(\top)$.
    \end{proposition}
    \begin{proposition}
        \label{prop:conv}
        Let $\vec{v} \in \N^X$.
        In every fair run $\sigma$ of $\PP$ starting in $C_{\vec{v}}$,
        the transition $\mathsf{conv}_i$ is taken precisely once in $\sigma$ for every $i \leq \text{max}(|\vec{v}|, \cutoffvar)$.
    \end{proposition}
    \begin{proposition}
        \label{prop:stab}
        In every fair run $\sigma$ of $\PP$ such that $\sigma_k((i, 0)) + \sigma_k((i, 1)) = 1$ for all but finitely many indices $k$,
        we have for all but finitely many indices $k$ that  $\sigma_k(\vec{q}) = 0$ for every $\vec{q} \in Q \setminus (Q_l \cup (X \times Q_i))$.
    \end{proposition}

    We now prove correctness of $\PP$. Let $\vec{v} \in \N^X$ and let $\sigma$ be a fair execution of $\PP$ starting in $C_{\vec{v}}$.
    We consider two cases: $|\vec{v}| \geq \cutoffvar$ and $|\vec{v}| < \cutoffvar$.

    Let us first consider the case  $|\vec{v}| \geq \cutoffvar$.
    If $|\vec{v}| \geq \cutoffvar$, then by fairness of $\sigma$ and Proposition~\ref{prop:conv} we have that the transitions $\mathsf{incr}_i$ are fired $\cutoffvar$ times in $\sigma$, and thus there is some $j$ such that $\sigma_j((\cutoffvar, b)) > 0$ for some $b \in \{0, 1\}$. By fairness of $\sigma$ and by construction of $T$, we then have that $\sigma_j \trans{\mathsf{true}} \sigma_{j+1}$ and $\sigma_{j+1}(\top) > 0$ for some $j$, and by Proposition~\ref{prop:inv_top} $\mathsf{true}$ is fired infinitely often in $\sigma$, and thus $\sigma$ stabilizes to a configuration where every agent is in state $\top$. Thus, if $|\vec{v}| \geq \cutoffvar$, then $\PP$ stabilizes to output $1$.

    Now consider the case $|\vec{v}| = m < \cutoffvar$. By fairness and Proposition~\ref{prop:conv}, the transitions $\mathsf{incr}_i$ are fired $m$ times in $\sigma$, until every agent leaves its initial state. Let $j$ be the largest index such that $\sigma_j \trans{\mathsf{incr}_{m-1}} \sigma_{m}$. Since the transitions $\mathsf{incr}_i$ are the only transitions that change the first component of the leader, we have $\sigma_k((m, 1))=1$ or $\sigma_k((m, 0))=1$ for every $k > j$. From this, Proposition~\ref{prop:stab} and fairness of $\sigma$, we have that eventually all non-leader agents are in a state from $X \times Q_m$. The transitions $T_{\mathsf{sim}}$ then guarantee by fairness of $\sigma$ that the non-leader agents stabilize to the output of $\PP_m$, while the transition $\mathsf{bool}_m$ ensures that the leader agent stabilizes to the same output, and thus, since by assumption $\PP_m$ stabilizes to $\varphi(\vec{v})$ if $|\vec{v}| = m$, and $|\vec{v}| = m$ holds by assumption, we have $O(\sigma)=\varphi(\vec{v})$. This completes the proof.
\end{proof}

The following lemma shows how to get rid of the leaders in halting protocols.

\begin{lemma}
    \label{lemma:kill-the-leader}
    Let $\cutoffvar \in \N$.
    For every protocol $\PP = (Q, T, L, X, I, O)$  with leaders that computes some predicate $\varphi$, there exists a leaderless protocol $\PP'$  with $\O(|Q|^{|L|+1} \cdot |X| \cdot \cutoffvar^2)$ states that computes $(|\vec{x}| < \cutoffvar) \rightarrow \varphi(\vec{x})$.
\end{lemma}
\begin{proof}
    Let $\PP = (Q, T, L, X, I, O)$ be a protocol, with $|L|$ leaders in states $\multiset{l_1, \ldots, l_{|L|}} = L$, that computes some predicate $\varphi$.
    We construct a leaderless protocol $\PP' = (Q', T', X, I', O')$, with $\O(|Q|^{|L|+1} \cdot |X| \cdot \cutoffvar^2)$ states, that computes $\varphi(\vec{x}) \lor (|\vec{x}| > \cutoffvar)$, which is equivalent to the desired predicate $(|\vec{x}| < \cutoffvar) \rightarrow \varphi(\vec{x})$.
    \newcommand{\frozen}{\mathtt{frozen}}
    \newcommand{\act}{\mathtt{active}}
    \newcommand{\leaderagent}{\vec{l}}
    \newcommand{\leader}{\mathtt{leader}}
    \newcommand{\popsize}{\mathtt{popsize}}
    \newcommand{\resetcounter}{\mathtt{resetcounter}}
    \newcommand{\init}{\mathtt{init}}
    \newcommand{\leaderstate}{\mathtt{leader}}

    The $|L|$ leaders are simulated by one agent we refer to as \emph{leader agent}. The leader agent is determined by a leader election. In general, the agents cannot know when the leader agent is finally elected, and so agents cannot wait for the leader election to be finished before starting their computation. However, as long as the population is sufficiently small, the leader agent may count the population size and reset the population to initial before starting the computation.

    The leader agent simulates both the $|L|$ leaders plus an additional regular agent in a \emph{multi-leader} state. Multi-leader states are defined as $Q_L \defeq \left(Q^{|L|} \times [\cutoffvar] \times [\cutoffvar] \times X \times Q\right)$. Multi-leader states contain a representation of the states of the $|L|$ leaders, plus meta-data needed for additional bookkeeping: The leader agent stores a lower-bound estimate of the population size and a number that indicates how many agents need to be reset to initial. The estimate of the population size indicates how many agents need to be reset after the leader agent has been elected, while resetting agents to initial ensures that the computation starts from a sane configuration. In the case where the estimate of the population size exceeds $\cutoffvar$, the leader agent moves to state $\top$ that converts everyone to true, thereby ensuring stabilization to consensus $1$. Multi-leader states are thus tuples of the form $\leaderagent = (\leaderstate_1, \ldots \leaderstate_{|L|}, \popsize, \resetcounter, \init, q) \in Q_L$ where:

    \begin{itemize}
        \item $\leaderstate_i$ is the current state of the $i^\text{th}$ leader simulated by the leader agent (where $1 \leq i \leq |L|$),

        \item $\popsize \in [\cutoffvar]$ is a counter for the population size,

        \item $\resetcounter \in [\cutoffvar]$ counts how many agents have been reset,

        \item $\init \in X$ stores the initial input of the regular agent simulated by the leader agent,

        \item $q \in Q$ stores the current state of the regular agent represented by the leader agent.
    \end{itemize}

    For every $\leaderagent \in Q_L$, we denote by $\leaderagent[\textit{attr}:= x]$ the state $\leaderagent'$ that is identical to $\leaderagent$, except that $\leaderagent'(\textit{attr}) = x$. For example, $\leaderagent[\popsize := 10]$ denotes the state $\leaderagent'$ that is identical to $\leaderagent$, except that $\leaderagent'(\popsize) = 10$.

    \parag{States} The set of states is
    \[Q' \defeq Q_L \cup \left(X \times Q \times \{\frozen, \act\}\right) \cup \{ \top \}.\]

\noindent An agent is thus either:
\begin{itemize}
\item a leader in a multi-leader state of $Q_L$;

\item a leader or non-leader in state $\top$, which converts every agent to $\top$; or

\item a non-leader in a state of the form $(x, q, s)$ with $s \in \{\frozen, \act\}$. The value of $x$ is the initial input the agent came from, and it never changes. The value of $q$ represents the current state of the agent from the original protocol. The value of $s$ determines whether the agent can interact with other non-leader agents: If $s = \frozen$, then the agent is ``frozen'' and cannot interact, otherwise it can interact freely with other agents.
\end{itemize}

\parag{Inputs} For every $x \in X$, we set the input mapping to:
    \[ I'(x) \defeq \left(l_1, \ldots, l_{|L|}, 1, 1, x, I(x)\right).  \]

\noindent So initially every agent is a leader agent.

\parag{Outputs} We set the opinions of the states to:
    \begin{align*}
        O'(\leaderagent) & \defeq O(\leaderagent(q)) && \text{ for every } \leaderagent \in Q_L, \\
        O'((i, q, x)) & \defeq O(q) && \text{ for every } (i, q, x) \in I\times Q \times \{\act, \frozen \}, \\
        O'(\top) & \defeq 1.
    \end{align*}

\parag{Election of the leader agent}
    For every $\leaderagent, \leaderagent' \in Q_L$ s.t.\ $\text{max}(\leaderagent(\popsize), \leaderagent'(\popsize)) < \cutoffvar$, we add the following transition to $T'$ :
    \[\multiset{\leaderagent, \leaderagent'} \mapsto \multiset{(l_1, \ldots, l_{|L|}, \leaderagent(\popsize) + \leaderagent'(\popsize), 1, \leaderagent(\init), \leaderagent(\init)),\ (\leaderagent'(\init), I(\leaderagent'(\init)), \frozen)}.
    \]
    This implements a leader election; by fairness, one leader agent eventually remains.

    \parag{Initiating conversion to $\top$}
    For every $\leaderagent, \leaderagent' \in Q_L$ s.t.\ $\text{max}(\leaderagent(\popsize), \leaderagent'(\popsize)) = \cutoffvar$, we add the following transition to $T'$:
    \[\multiset{\leaderagent, \leaderagent'} \mapsto \multiset{\top, \top}. \]
    This transition ensures that if the population size is at least $\cutoffvar$, then all agents are eventually converted to $\top$, thus yielding a stable $1$-consensus.

    \parag{Conversion to $1$-consensus}
    For every $q \in Q'$, we add the following transition to $T'$:
    \[\multiset{\top, q} \mapsto \multiset{\top, \top}. \]
    This transition ensures that all agents eventually move to to $\top$ when one agent reaches $\top$.

    \parag{Interactions with leaders}
    For every $\leaderagent \in Q_L$ s.t.\ $\leaderagent(\popsize) = \leaderagent(\resetcounter)$, every $x \in X$ and every $r \in Q$, we add the following transitions to $T'$:
    \begin{align*}
      & \multiset{\leaderagent, (x, r, \act)} \mapsto \multiset{\leaderagent[q := q'], (x, r', \act)} && \text{ for every } \multiset{\leaderagent(q), r} \mapsto \multiset{q', r'} \in T, \\[5pt]
      & \multiset{\leaderagent, (x, r, \act)} \mapsto \multiset{\leaderagent[\leader_i := l'], (x, r', \act)} && \text{ for every } 1 \leq i \leq |L| \colon  \\
      & && \multiset{\leaderagent(\leader_i), r} \mapsto \multiset{l', r'} \in T.
    \end{align*}
    This simulates interactions with leaders.

    \parag{Interactions among regular agents}
    For every $(x, q), (y, r) \in X \times Q$ such that $\multiset{q, r} \mapsto \multiset{q', r'} \in T$ for some $q', r'$, we add the following transition to $T'$:
    \[\multiset{(x, q, \act), (y, r, \act)} \mapsto \multiset{(x, q', \act), (y, r', \act)}.\]
    This simulates interactions between non-leader agents.

    \parag{Freezing agents}
    For every $\leaderagent \in Q_L$ such that $\leaderagent(\resetcounter) < \leaderagent(\popsize)$, and every $(x, q) \in X \times Q$, we add the following transitions to $T'$:
    \begin{align*}
      \multiset{\leaderagent, (x, q, \act)} & \mapsto \multiset{\leaderagent[\resetcounter := 1],\ (x, I(x), \frozen)}, \\[5pt]
      \multiset{\leaderagent, (x, q, \frozen)} & \mapsto \multiset{\leaderagent[\resetcounter := \leaderagent(\resetcounter) + 1],\ (x, q, \act)}.
      &&
    \end{align*}
    These transitions take care of freezing/activating agents and resetting agents to initial. Intuitively, the leader agent resets active agents by resetting their states to initial while simultaneously freezing them. Thus the following invariant is maintained: whenever an agent is frozen, it is in its initial state. The reset counter indicates how many frozen agents need to be activated: If the counter equals $i$, then $\leaderagent(\popsize) - i$ agents must be activated. Ideally, the leader agents resets the population by first freezing all agents, one after another, and then activating each agent one by one. Of course, this order of steps cannot be guaranteed, but if it is is violated, then by fairness $\leaderagent(\resetcounter)$ is eventually set to $1$, and freezing/resetting is reinitiated.
    \qedhere
\end{proof}

\subsection{Proof of Theorem \ref{lemma:halting-conj-disj}} \label{app:construction-halting-conj-dis}

\thmHaltingConjDis*

\newcommand{\leadertag}{\square}
\newcommand{\tags}{{Y_{\texttt{tag}}}}

\begin{proof}
	\newcommand{\pres}[1]{\text{pres}(#1)}
	Without loss of generality, we assume that the state sets of $\PP_1, \ldots, \PP_k$ are pairwise disjoint.
	For every $j \in [k]$, let $\q{f}_j$ and $\q{t}_j$ denote the output states of protocol $\PP_j$.
	Further let $l_j$ be the initial leader state of protocol $\PP_j$, i.e.\ $\multiset{l_j} = L_j$.

	Remember that our final protocol $\PP$ should evaluate the outcomes of the individual protocols in succession. To this end, we enrich all states of $Q_1 \cup \ldots \cup Q_k$
	with a tag in $\tags \defeq X \cup \left\{\leadertag\right\}$. Intuitively, each agent is ``tagged'' with the input variable it corresponds to or with $\leadertag$ if it was the leader. This way, when a protocol $\PP_j$ halts, one can rewind to the initial configuration and start the next protocol $\PP_{j+1}$.

	Formally, for a given protocol $\PP_j$, the \emph{tagged protocol} $\PP_j^Y$ is $(Q_j^Y, L_j^Y, X, T_j^Y, I_j^Y, O^Y_j)$ where
	\begin{align*}
		& Q_j^Y && \defeq Y \times Q_j, & \\
		& L_j^Y && \defeq \multiset{(\leadertag, l_j)}, &\\
		& T_j^Y &&\defeq \{\multiset{(x, q), (y, r)} \mapsto \multiset{(x, q'), (y, r')} \mid \multiset{q, r} \mapsto \multiset{q', r'} \in T_j \},& \\
		& I_j^Y(x)  && \defeq (x, I_j(x))  \qquad \text{for every $x \in X$}, \\
		& O_j^Y((x, q)) && \defeq O_j(q)  \qquad \text{for every $(x, q) \in Q_j^Y$.}
	\end{align*}

	Notice that $\PP_j^Y$ is no longer simple, because we will have multiple states $(x, \q{f_j})$
	with output $0$ and  $(x,\q{t_j})$ with output $1$, one per $x \in \tags$. We will say that the intermediate tagged protocols $\PP_j^Y$ are {\em tagged-simple}.
	However, it is easy to recover a simple halting protocol from a tagged-simple halting protocol: we can add two states $\q{f},\q{t}$ and transitions $\multiset{(x, \q{f})} \mapsto \multiset{\q{f}}$ and $\multiset{(x, \q{t})} \mapsto \multiset{\q{t}}$ for all $x \in \tags$. It thus suffices to show how to combine the individual tagged-simple protocols to a tagged-simple protocol of appropriate size.

\medskip

We show by induction on $\len(\varphi)$:
for every boolean combination of $\varphi$ of atomic predicates $\varphi_1, \ldots, \varphi_k$, there exists a tagged-simple halting protocol $\PP'$ with
$\O(|X| \cdot (\len(\varphi) + |Q_1| + \cdots + |Q_k|))$ states and one leader that computes $\cond{\varphi}{i}$. By the previous remark, the claim entails the theorem to be shown.
\medskip

\noindent The case $\len(\varphi)=0$ is trivial, since $\varphi$ is computed by $\PP^Y_j$ for some $j$ if $\len(\varphi)=0$ holds.

\noindent For the induction, consider $\varphi= \varphi_1 \oplus \varphi_2$
for $\oplus \in \{ \wedge, \vee \}$, and assume the existence of tagged-simple protocols $\PP'_1, \PP'_2$ that satisfy the claim for $\varphi_1$ and $\varphi_2$, respectively.
We construct a protocol $\PP_\oplus=(Q, T, L, X, I, O)$ that computes $\cond{\varphi}{i}$ as follows.\medskip

	\parag{States and associated mappings}
	We define states of $\PP_\oplus$ as:
	\[Q \defeq \left(Q'_1 \cup Q'_2 \right) \cup \{(x,\q{f}), (x,\q{t}) \mid x \in \tags\}.\] The leader multiset corresponds to the tagged leader multiset of $\PP'_1$: \[L \defeq \multiset{(\leadertag,l_1)}.\]
	The output mapping is given by $O((x, \q{f})) \defeq 0$, $O((x, \q{t})) \defeq 1$ for every $x \in Y$, and $O(\vec{q}) \defeq \bot$ for every other $\vec{q} \in Q$.
	The input mapping is defined as $I(x) \defeq I'_1(x)$ for every $x \in X$.\medskip

	\parag{Transitions} The set of transitions is $T=T'_1 \uplus T'_2 \uplus T'$, where $T'$ is constructed as follows:  For every
$(x,q) \in Q'_1$, $(y,r) \in Q'_2$, we add the following transitions to $T'$:
	\begin{align*}
	  \multiset{(x, q), (y, r)} & \mapsto
	  	\begin{cases}
	  		\multiset{I'_2(x), (y, r)} & \text{ if } x \not= \leadertag \\
	  		\multiset{(x, l'_2), (y, r)} & \text{ if } x = \leadertag
  		\end{cases}
	\end{align*}

	These transitions make sure that once at least one agent is promoted to a state in the higher protocol $\PP'_2$, all agents eventually simulate the execution of protocol $\PP'_2$.

	Moreover, we add the following transitions to $T'$: for every $x \in \tags$, depending on the operator $\oplus$:\smallskip
	\begin{itemize}
		\item If $\oplus = \land$, we add the following transitions for each $x \in \tags$:
		      \begin{align*}
		      \multiset{(x,\q{f}'_j)} &\mapsto \multiset{(x,\q{f})} && \text{ for every } j \in \{1,2\}, \\
		      \multiset{(x,\q{t}'_1)} & \mapsto
		      	\begin{cases}
		      		\multiset{I'_2(x)} & \text{ if } x \not= \leadertag \\
		      		\multiset{(x, l'_2)} & \text{ if } x = \leadertag
		      	\end{cases}, \\
		      \multiset{(x,\q{t}'_2)} & \mapsto \multiset{(x,\q{t})}
				.
		      \end{align*}
	\end{itemize}

	\begin{itemize}
		\item If $\oplus = \lor$, we add the following transitions for each $x \in \tags$:
		      \begin{align*}
		      \multiset{(x,\q{t}'_j)} &\mapsto \multiset{(x,\q{t})} && \text{ for every } j \in \{1,2\}, \\
		      \multiset{(x,\q{f}'_1)} & \mapsto
		      	\begin{cases}
		      		\multiset{I'_2(x)} & \text{ if } x \not= \leadertag \\
		      		\multiset{(x, l'_2)} & \text{ if } x = \leadertag
		      	\end{cases}, \\
		      \multiset{(x,\q{f}'_2)} & \mapsto \multiset{(x,\q{f})}
				.
		      \end{align*}

	\end{itemize}

These transitions ensure that once an output state is reached in the simulation of a given protocol, then either its output
	is returned as final output (in the case where $\oplus = \lor$ and output of the protocol is $1$, or $\oplus = \land$ and output of the protocol is $0$),
or the simulation of the second protocol is initiated, until its output is returned, and $\PP$ satisfies the claim by induction hypothesis.
Note that each inductive call adds $2 |\tags|$ states, which results in the bound given in our theorem.
\end{proof}

\subsection{Proof of Theorem \ref{thm:halting:threshold}}\label{app:constr-halting-threshold}

We consider only the case
$\varphi(\vec{x}, \vec{y}) \defeq  \vec{\alpha} \cdot \vec{x} - \vec{\beta}\cdot \vec{y} > 0$; the general case $> c \in \mathbb{Z}$
is easily adapted from there.
We explain later how to adapt the proof to handle remainders predicate
$\varphi(\vec{x}) \defeq \left( \vec{\alpha} \cdot \vec{x} \equiv_m b \right)$ with $m \in \N$ and  $0 \leq b \leq m$.

\thmHaltingThreshold*
 \begin{proof}
    \newcommand{\sgn}{\mathtt{sgn}}
    \newcommand{\tgt}{\mathtt{tgt}}
    \newcommand{\pos}{\mathtt{pos}}
    \renewcommand{\val}{\mathtt{val}}
    \newcommand{\bin}{\mathtt{bin}}
    \newcommand{\reset}{\mathtt{reset}}
    \newcommand{\incr}{\text{incr}}
    \newcommand{\met}{\texttt{met}}

Let $A \defeq \{\alpha_j : j \in \{1, \ldots, |\vec{\alpha}|\}\}$ and  $B \defeq \{\beta_j : j \in \{1, \ldots, |\vec{\beta}|\}\}$. Let $m$ be the maximal bit length of any number in the set $\{x_1 + \ldots + x_n : x_j \in A\} \cup \{y_1 + \ldots + y_n : y_j \in B\}$.
Note that $m \in \O\left(\log \left(i \cdot 2^{|\varphi|}\right)\right) = \O(|\varphi| + \log i)$.
For any $a < 2^{m}$, we write $\bin(a)$ to denote the least-significant-bit-first binary representation of $a$, padded to length $m$ with leading $0$s.
Whenever $\bin(a) = b_m b_{m-1} \ldots b_1$, we write $\bin(a)_j$ to denote $b_j$ for every $j \in [m]$.

Consider the sequential algorithm {\sc{Greater-Sum}$(X,Y)$} shown in Figure~\ref{fig_probe}. We have $\vec{\alpha} \cdot \vec{x} - \vec{\beta}\cdot \vec{y} > 0$ if{}f  {\sc{Greater-Sum}$(A,B)$} returns true. So it suffices to exhibit a protocol that simulates the execution of
{\sc{Greater-Sum}$(A,B)$} for inputs of size $i$.  Intuitively, the protocol has a leader that executes the procedure. The leader stores the values of the variables defined in $\textsc{Greater-Sum}$. Regular agents store the input and one additional bit that indicates whether the leader has met the agent in the current round. The leader can set and unset this bit, which permits the implementation of a \emph{for all} loop: the leader stores how many agent it has met in the current iteration of the loop. Whenever the leader encounters an agent whose bit is set to $0$, it flips the bit to $1$, increments its counter, and performs the variable updates defined in the body of the loop. When the counter value reaches $i$, the leader knows that the current iteration of the loop is complete. The leader then unsets all bits of the regular agents while decrementing its counter agent by agent, before starting the next iteration of the loop when the counter value reaches zero.

We now define the protocol  $\PP = (Q, L, T, I, O)$ formally.

    \begin{figure}
    \begin{minipage}[t]{0.5\textwidth}
	\renewcommand{\algorithmicrequire}{\textbf{Input:}}
	\renewcommand{\algorithmicensure}{\textbf{Output:}}
	\begin{algorithmic}[1]
	\Require{\begin{tabular}[t]{l}
	Finite Multiset $Z \in \N^\N$, \\
	bit position $j \in [m]$
	\end{tabular}}
	\Ensure{$j^\text{th}$ bit of $\sum Z$}
	\Procedure{Probe}{$Z$, $j$}
		\State $\val \leftarrow 0$
		\For{$\pos = 1  \textbf{ to } j$}
			\State $\val \leftarrow \val \text{ div } 2$
			\ForAll{$z \in  Z$}
				\State $\val \leftarrow \val + \bin(z)_\pos$
			\EndFor
		\EndFor
		\State \Return{$\val \text{ mod } 2$}
	\EndProcedure
	\end{algorithmic}
	\end{minipage}
	\qquad
	\begin{minipage}[t]{0.5\textwidth}
	\renewcommand{\algorithmicrequire}{\textbf{Input:}}
	\renewcommand{\algorithmicensure}{\textbf{Output:}}
	\begin{algorithmic}[1]
	\Require{\begin{tabular}[t]{l} Finite multisets $X, Y \in \N^\N$ \\[0cm] {\ } \end{tabular}}
	\Ensure{boolean $\left(\sum X > \sum Y\right)$}
	\Procedure{Greater-Sum}{$X$, $Y$}
		\For{$\tgt = m \text{ to } 1$}
			\State $\val_X \leftarrow \mathtt{PROBE}(X, \tgt)$
			\State $\val_Y \leftarrow \mathtt{PROBE}(Y, \tgt)$
			\If{$\val_X \neq \val_Y$}
 				\State{\Return{$\val_X > \val_Y$}}
 			\EndIf
		\EndFor
		\State{\Return{false}}
	\EndProcedure
	\end{algorithmic}
	\end{minipage}
	\caption{Procedure {\sc{Probe}$(Z,j)$} probes the $j^\text{th}$ bit of the sum of the elements of $Z \in \N^\N$. It implements a binary adder, but instead of storing the result of the addition, it only keeps the carry in variable $\val$ when moving from one bit position to the next. \\
	Procedure {\sc{Greater-Sum}$(X,Y)$} compares the sums of the elements of $X, Y \in \N^\N$. It probes the bits of the two sums, starting with the most-significant bit, until it finds the first position at which the bits of the two sums differ. If there is no such position, the sums are equal and the algorithm returns \emph{false}.}
	\label{fig_probe}
	\end{figure}

    \parag{States}
    Let $Q \defeq Q' \cup Q_L$ where $Q' \defeq (A \cup B) \times \{0, 1\}$ is the set of states for regular agents, and $\{\q{f}, \q{t}\} \subseteq Q_L$ is a set of leader states yet to be specified.

    The leader multiset $L$ contains exactly one leader.
    Unless the leader is in one of the output states $\{\q{f}, \q{t}\}$, it stores the following values:
    \begin{itemize}
        \item $\tgt \in [m]$ : the target bit position to be probed; corresponds to the loop counter  $\tgt$ in line 2 of Procedure $\textsc{Greater-Sum}$.

        \item $\pos \in [m]$: the current bit position; corresponds to the loop counter $\pos$ in line 3 of Procedure $\textsc{Probe}$.

        \item $\met \in \{0, \ldots, i-1\}$: the number of agents the leader has met in the current round. This is needed for the implementation of the $\emph{forall}$ loop in Procedure $\textsc{Probe}$.

        \item $\reset \in \{0, 1\}$: indicates whether the bit flag of each regular agent should be reset.

        \item $\val_Z \in [i \cdot m]$ for every $Z \in \{X, Y\}$: storage for sum of bits from binary representations of numbers in $A$ and $B$, respectively; corresponds to $\val_X$ $\val_Y$ in Procedure $\textsc{Greater-Sum}$, and $\val$ in Procedure $\textsc{Probe}$.
    \end{itemize}

    Initially, the variables of the leader are set as follows: $\tgt = m, \pos = 1, \met = 0, \reset = 0, \val_{X} = \val_{Y} = 0$.
    Note that this corresponds to the initial values of the variables in the procedures $\textsc{probe}$ and $\textsc{Greater-Sum}$. Thus, we set:
    \begin{align*}
    Q_L &\defeq [m] \times [m] \times \{0, \ldots, i-1\} \times \{0, 1\} \times [i \cdot m], \\
    L &\defeq \multiset{(m, 1, 0, 0, 0, 0)}.
    \end{align*}

  	\parag{Size} The number of states is $|Q| = |Q'| + |Q_L| = 2*|A \cup B| + 2m^3i^2 \in \O(i^2 (|\varphi| + \log i)^3)$.

    \parag{Input and output mappings}
    We define the input mapping $I$ as:
    \begin{align*}
        I(x_j) &\defeq (\alpha_i, 0) && \text{ for every } 1 \leq j \leq |\vec{\alpha}|, \\
        I(y_j) &\defeq (\beta_i, 0) && \text{ for every } 1 \leq j \leq |\vec{\beta}|.
    \end{align*}
    The output mapping $O$ is defined as:
    \begin{align*}
        O(\q{f}) &\defeq 0, \\
        O(\q{t}) &\defeq 1, \\
        O(q) &\defeq \bot && \text{ for every } q \in Q \setminus \{\q{f}, \q{t}\}.
    \end{align*}

    \parag{Transitions}
    For a state $\vec{q} =(\gamma, b) \in Q'$, let $\vec{q}(\gamma)$ denote $\gamma$ and let $\vec{q}(b)$ denote $b$.

    To implement resetting the bit flag of the regular agents, we add the following transitions for every
    $\vec{q} \in Q'$ and $\vec{l} \in Q_L$ where $\vec{q}(b) = 1$ and $\vec{l}(\reset) = 1$:
    \[\multiset{\vec{q},\ \vec{l}} \mapsto \multiset{\vec{q}[b := 0],\ \vec{l}[\met := (\vec{l}(\met) + 1 \text{ mod } (i-1)); \ \reset := \text{min}(1, (\vec{l}(\met)
     + 1) \text{ mod } i )]}. \]

     \begin{table}[t]
     \begin{center}
    \begin{tabular}{|l|l|p{2cm}|}
        \hline
        \multicolumn{1}{|c|}{Conditions satisfied by $\vec{l}$} &
        \multicolumn{1}{|c|}{Value of $\vec{l}'$} &
        \multicolumn{1}{|c|}{Corresponds to}\tabularnewline
        \hline
        \hline
        \parbox{0.5cm}{\begin{align*}&\vec{l}(\met) < i-1\end{align*}} &
        \parbox{0.5cm}{\begin{align*}&\vec{l}_\incr[\met := \vec{l}(\met) + 1]\end{align*}}
         & \parbox[c]{\hsize}{Line 6 of Procedure $\textsc{probe}.$} \tabularnewline
        \hline
        \parbox{0.5cm}{\begin{align*}& \vec{l}(\met) = i-1\\
                                     & \vec{l}(\pos) < \vec{l}(\tgt)
                    \end{align*}} &
        \parbox{0.5cm}{\begin{align*}\vec{l}_\incr & [\val_X := \val_X \text{ div } 2; \\
                                    &  \val_Y := \val_Y \text{ div } 2; \\
                                    &  \pos := \vec{l}(\pos) + 1; \\
                                    &  \met := 0; \\
                                    & \reset := 1]
            \end{align*}} & \parbox[c]{\hsize}{Continuation of for loop in line 4 of Procedure $\textsc{probe}$.} \tabularnewline
        \hline
        \parbox{0.5cm}{\begin{align*}& \vec{l}(\met) = i-1,\\
                                     & \vec{l}(\pos) = \vec{l}(\tgt), \\
                                     & \vec{l}_\incr(\val_X) \neq \vec{l}_\incr(\val_Y)
                    \end{align*}}
         & \parbox{0.5cm}{\begin{align*} \q{t}\ & \text{ if } \vec{l}_\incr(\val_X) > \vec{l}_\incr(\val_Y),\\
                                                   \q{f}\ & \text{ otherwise.}
         \end{align*}} & \parbox[c]{\hsize}{Return statement in line 6 of Procedure $\textsc{Greater-Sum}$.} \tabularnewline
        \hline
        \parbox{0.5cm}{\begin{align*}& \vec{l}(\met) = i-1,\\
                                     & \vec{l}(\pos) = \vec{l}(\tgt) > 1, \\
                                     & \vec{l}_\incr(\val_X) = \vec{l}_\incr(\val_Y).
                    \end{align*}} &
                    \parbox{0.5cm}{
                    \begin{align*}
                    \vec{l} & [\tgt := \vec{l}(\tgt) - 1; \\
                    & \pos = 1; \\
                    & \met = 0; \\
                    & \reset = 1]
                    \end{align*}}
                      & \parbox[c]{\hsize}{Continuation of the for loop in Procedure $\textsc{Greater-Sum}$} \tabularnewline
        \hline
        \parbox{0.5cm}{\begin{align*}\text{Other}\end{align*}} & \parbox{0.5cm}{\begin{align*}\q{f}\end{align*}}  & \parbox[c][2cm]{\hsize}{Return statement in line 9 of Procedure $\textsc{Greater-Sum}$.} \tabularnewline
        \hline
    \end{tabular}
    \end{center}
    \vspace{0.5cm}
    \caption{Transitions of the protocol implementing {\sc Greater-Sum}.}
    \label{table:transitions}
\end{table}
    We now define the remaining transitions for the execution of procedure $\textsc{Greater-Sum}$.
    Let $\vec{q} \in Q'$ and $\vec{l} \in Q_L \setminus \{\q{f}, \q{t}\}$. Let us first establish some abbreviations.

    Let: \[Z \defeq \begin{cases} X & \text{ if } \vec{q}(\gamma) \in A, \\
                                  Y & \text{ if } \vec{q}(\gamma) \in B.
                   \end{cases}\]
    Further let $\vec{l}_\incr \defeq \vec{l}[\val_Z := \vec{l}(\val_Z) +\bin(\vec{q}(\gamma))_\pos]$.
    Intuitively, $\vec{l}_\incr$ represents the update to the leader state that results from the incrementation in line $6$ of Procedure $\textsc{probe}$.

    Whenever $\vec{q}(b) = 0$ and $\vec{l}(\reset) = 0$, we add the following transitions:
    \[\multiset{\vec{q}, \vec{l}} \mapsto \multiset{\vec{q}[b := 1], \vec{l}'}\]
    where $\vec{l}'$ is specified in Table~\ref{table:transitions}.
\end{proof}

\smallskip
\noindent \textbf{Remainder}.
Consider now a predicate $\varphi(\vec{x}) \defeq \left( \vec{\alpha} \cdot \vec{x} \equiv_m b \right)$ with $m \in \N$ and  $0 \leq b \leq m$. We show that for every $i \in \N$, there exists a halting protocol with one leader and $\O\left(\poly(|\varphi| + i) \right)$ states that computes $\cond{\varphi}{i}$.

The protocol in which a leader interacts with every other agent, storing in its state the value
of  $\vec{\alpha} \cdot \vec{v}' \bmod m$, where $\vec{v}'$ is the vector of the agents it has already interacted with,
does not work: For $m \in \Theta(|\varphi|)$, which can be the case, this requires $\O(2^{|\varphi|})$ states.
So we proceed in a different way.

\begin{theorem}
	\label{thm:halting:modulo}
Let $\varphi(\vec{x}) \defeq \left( \vec{\alpha} \cdot \vec{x} \equiv_m b \right)$ with $b, m \in \N$,  $0 \leq b < m$ and $\vec{\alpha} \in \Z^{|X|}$. For every $i \in \N$, there exists a halting protocol with one leader and
$\O\left(|X|\cdot i^3 (|\varphi| + \log i)^3)\right)$
states that computes $\cond{\varphi}{i}$.
\end{theorem}

\begin{proof}
Let $\vec{\alpha}' \defeq (\vec{\alpha}(1) \bmod m, \ldots, \vec{\alpha}(|X|) \bmod m)$. Since $\vec{\alpha} \cdot \vec{x} \equiv_m  \vec{\alpha}' \cdot \vec{x}$, setting $\varphi'(\vec{x}) \defeq (\vec{\alpha}' \cdot \vec{x} \equiv_m b)$ yields $\varphi(\vec{x}) = \varphi'(\vec{x})$ for every input $\vec{x}$. Consider an input $\vec{x}$ of size $i$. We have $\vec{\alpha}' \cdot \vec{x} \leq m \cdot i$, hence $\varphi'(\vec{x})=1$ if{}f  $\vec{\alpha}' \cdot \vec{x} \in \{b, m+b, \ldots, (i-1)m + b \}$, and consequently:
	$$\varphi(\vec{x}) \equiv \varphi'(\vec{x}) \equiv
	\bigvee_{j=0}^{i-1} \left( \vec{\alpha}' \cdot \vec{x} = j \cdot m + b \right),$$
\noindent which is a disjunction of $i$ threshold predicates $\varphi'_1, \cdots, \varphi'_i$. For every $j \in [i]$ it holds:
\begin{align*}
	|\varphi'_j| &\in \O\left(\log\norm{\varphi'_j} + \len(\varphi'_j) + |X|\right) \\
							&\subseteq \O\left(\log(\norm{\varphi} \cdot i) + \len(\varphi) + |X|\right) \\
							&= \O\left(\log i + \log\norm{\varphi} + \len(\varphi) + |X|\right) \\
							&= \O\left(\log i + |\varphi|\right)
\end{align*}
By Lemma~\ref{lemma:halting-conj-disj} there is a  protocol $\PP_j$ computing $\cond{\varphi'_j}{i}$ with
\begin{align*}
	&\O\left(i^2(|\varphi'_j| + \log i)^3\right) \\
	&\in \O\left(i^2(|\varphi| + 2\log i)^3\right) \\
	&\in \O\left(i^2(|\varphi| + \log i)^3\right)
\end{align*}
states.
By Theorem~\ref{thm:halting:threshold},  $\cond{\varphi'}{i}$
can be computed by a protocol with
\begin{align*}
	&\O\left(|X|\cdot (i + i \cdot (i^2(|\varphi| + \log i)^3))\right) \\
	&\in \O\left(|X|\cdot i^3(|\varphi| + \log i)^3 \right)
\end{align*}
	states.
\end{proof}


\section{Proof of Theorem \ref{thm:hardness:PA}} \label{appendix_PA}

\thmHardnessPA*

\begin{proof}
 We show that if such an algorithm runs in time $2^{p(n)}$ for some polynomial $p$, then the validity problem for PA formulas  is in EXPTIME, contradicting the fact that its complexity lies between 2-NEXP and 2-EXPSPACE~\cite{Berman80, Haase14}. Recall that the validity problem for PA formulas asks whether a given sentence, i.e., a formula without free variables, is true or false.

Let $\varphi$ be a sentence of PA, and let $n \defeq |\varphi|$. Consider the formula $\psi(x) \defeq (x \geq 2) \wedge \varphi$ (notice that the smallest possible size of a population is 2). Clearly, $\varphi$ is valid if{}f $\psi(2)$ holds. Assume there exists an algorithm that on input $\psi$ executes at most $f(n)$ steps and outputs a population protocol $\PP$ that computes $\psi$. Clearly, $\PP$ has at most $\O(f(n))$ states, and $\varphi$ is valid if{}f $\PP$ computes 1 for input $x = 2$.

Let $C$ be the initial configuration of $\PP$ for input $x = 2$. A configuration of $\PP$ with two agents can be stored in space $\O(\log f(n))$, and so $C$, and every configuration reachable from it, can be stored using $\O(\log f(n))$ space. Protocol $\PP$ computes $1$ from $C$ if{}f there exists a configuration $C'$ such that:
\begin{enumerate}[(i)]
\item $C \trans{*} C'$,

\item $C'$ has output 1,

\item for every configuration $C''$, if  $C' \trans{*} C''$, then $C'' \trans{*} C'$.
\end{enumerate}

Observe that (i)-(iii) can be expressed in FO(TC), \ie\ first-order logic with transitive-closure. By Immermann's theorem, deciding (i)-(iii) belongs to NSPACE($\log f(n)$)~\cite{Immerman99}, and so it can be solved in $\O(f(n)^k)$ deterministic time for some $k \geq 1$. Consequently, if there exists a polynomial $p$ such that $f \in \O(2^{p(n)})$, then the validity of $\varphi$ can be decided in time $2^{\O(p(n))}$, and so the validity problem for PA is in EXPTIME. The latter is impossible by the time hierarchy theorem.
\end{proof}

\end{document}